\documentclass[11pt]{article}

\usepackage[a4paper,margin=2.5cm]{geometry}
\usepackage{microtype}

\usepackage{amsmath,amssymb,amsfonts}
\usepackage{amsthm}
\usepackage{mathrsfs}
\usepackage{amsmath}
\usepackage{graphicx}
\usepackage{amssymb}
\usepackage[dvipsnames]{xcolor}
\usepackage{mathtools}
\usepackage{tikz}
\usepackage{bbold}
\usepackage{bbding}
\usepackage{comment}
\usepackage[none]{hyphenat}
\usepackage{enumitem}
\usepackage{etoolbox}
\usepackage[toc,page]{appendix}

\setlist[itemize]{itemsep=0pt, topsep=1pt}
\setlist[enumerate]{itemsep=0pt, topsep=1pt}

\usetikzlibrary{automata, positioning, arrows}
\tikzset{
	node distance=3cm,
	initial text=$ $
}

\theoremstyle{plain}
\newtheorem{theorem}{Theorem}[section]
\newtheorem{lemma}[theorem]{Lemma}

\newtheorem{definition}[theorem]{Definition}
\newtheorem{corollary}[theorem]{Corollary}

\theoremstyle{definition}
\newtheorem{example}[theorem]{Example}

\theoremstyle{remark}
\newtheorem{remark}[theorem]{Remark}

\newcommand{\arc}{\mathrm{arc}}
\newcommand{\amer}{\mathrm{amer}}

\newcommand{\epsfree}{\eps\text{-}\mathrm{free}}

\newcommand{\free}{\mathrm{free}}

\newcommand{\llong}{\mathrm{long}}

\newcommand{\op}{\operatorname{op}}
\newcommand{\eur}{\mathrm{eur}}

\newcommand{\prev}{\mathrm{prev}}
\newcommand{\QFL}{\mathrm{QFL}}

\newcommand{\Reg}{\mathrm{REG}}
\newcommand{\res}{\mathrm{res}}
\newcommand{\Run}{\mathrm{Run}}
\newcommand{\short}{\mathrm{short}}

\newcommand{\supp}{\mathrm{supp}}

\newcommand{\trop}{\mathrm{trop}}

\newcommand{\WFFA}{\mathrm{WFFA}}
\newcommand{\wt}{\mathrm{wt}}

\newcommand{\zero}{\mathbb{0}}
\newcommand{\one}{\mathbb{1}}
\newcommand{\eps}{\varepsilon}

\newcommand{\A}{\mathcal{A}}
\newcommand{\C}{\mathcal{C}}
\newcommand{\D}{\mathcal{D}}

\newcommand{\F}{\mathcal{F}}
\newcommand{\I}{\mathcal{I}}
\renewcommand{\L}{\mathcal{L}}
\renewcommand{\O}{\mathcal{O}}
\renewcommand{\P}{\mathcal{P}}
\newcommand{\R}{\mathcal{R}}

\newcommand{\T}{\mathcal{T}}
\newcommand{\W}{\mathcal{W}}

\renewcommand{\AA}{\mathbb{A}}
\newcommand{\bb}{\mathbb{b}}

\newcommand{\CC}{\mathbb{C}}

\newcommand{\DD}{\mathbb{D}}
\newcommand{\EE}{\mathbb{E}}
\newcommand{\FF}{\mathbb{F}}
\newcommand{\hh}{\mathbb{h}}

\renewcommand{\ll}{\mathbb{l}}

\newcommand{\mm}{\mathbb{m}}
\newcommand{\MM}{\mathbb{M}}
\newcommand{\NN}{\mathbb{N}}
\newcommand{\PP}{\mathbb{P}}

\newcommand{\RR}{\mathbb{R}}

\renewcommand{\SS}{\mathbb{S}}
\newcommand{\UU}{\mathbb{U}}

\newcommand{\lsem}{[\![}
\newcommand{\rsem}{]\!]}
\newcommand{\llangle}{\langle \! \langle}
\newcommand{\rrangle}{\rangle \! \rangle}
\newcommand{\beh}[1]{{\lsem #1 \rsem}}
\newcommand{\bind}[1]{{\llangle #1 \rrangle}}
\newcommand{\dist}[1]{{| \! | #1 | \! |}}
\newcommand{\clp}[1]{\langle #1 \rangle_{\oplus}}

\newcommand{\keywords}[1]{\par\noindent\textbf{Keywords: }#1}

\title{Weighted Automata and Regular Expressions for Financial Systems}

\author{
Manfred Droste\\
Institut f\"ur Informatik, Universit\"at Leipzig, Leipzig, Germany\\
\texttt{droste@informatik.uni-leipzig.de}
\and
Vitaly N\"urnberg\\
Independent Researcher, Herrischried, Germany\\
\texttt{vitaly.nuernberg@gmail.com}
}

\date{}
\begin{document}
\maketitle

\abstract{
We introduce \emph{weighted finite finance automata} (WFFA), 
a formal framework for modeling and analyzing quantitative properties of financial systems driven by uncertain economic variables such as stock prices, interest rates, and exchange rates. 
The model provides a compositional and language-theoretic approach to scenario-based financial analysis, enabling systematic evaluation of financial instruments and trading strategies.

To specify such systems, we introduce \emph{weighted finance regular expressions}, 
a declarative language for quantitative financial properties. 
We establish a Kleene--Sch\"utzenberger-type correspondence~\cite{Sch61,Kle56} between WFFAs and weighted finance regular expressions, together with effective translation procedures between the two formalisms.

On the algorithmic side, we investigate fundamental decision and optimization problems for WFFAs, including the computation of extremal payoffs, and identify expressive yet computationally tractable subclasses. 
These results provide a foundation for formal, compositional, and efficient analysis of financial systems under multiple market scenarios.
}


\maketitle

\section{Introduction}\label{sec:introduction}

Financial markets play a central role in modern economic systems and give rise to a wide range of complex financial instruments \cite{Fab08,Hul93,Lue99,Nat14}. 
To support decision-making, a variety of sophisticated mathematical models, in particular stochastic models, and software tools have been developed. However, given the inherent unpredictability of financial markets, these approaches cannot guarantee reliable outcomes. Moreover, they typically lack a compositional and formally analyzable structure that would allow quantitative reasoning across sequences of interacting market events.

A central challenge in computational finance is the efficient evaluation of financial instruments over large collections of possible market scenarios. In practice, such analyses rely on scenario generation, stress testing, and backtesting (see, e.g., \cite{Has16}), which often require ad hoc numerical procedures and offer limited support for compositional modeling. As a result, it remains difficult to reason about complex financial systems in a scalable and interpretable way. These limitations motivate the development of formal frameworks that support both expressive modeling of financial scenarios and efficient quantitative evaluation.

Formal methods provide a natural foundation for such frameworks. By offering mathematically precise models with well-defined semantics, they enable rigorous specification, analysis, and verification of complex systems. In recent years, formal methods have been successfully applied in financial contexts, including the design of exchange mechanisms \cite{GS22} and the analysis of DeFi smart contracts \cite{BDKJ23}, demonstrating their potential for improving reliability and risk management.

At the core of formal methods lie non-deterministic finite automata \cite{RS59,HU79} and regular expressions \cite{Kle56,HU79}, which serve as fundamental models for representing and analyzing discrete systems. Automata provide intuitive operational models and form the basis for verification techniques such as model checking \cite{BK08}, while regular expressions offer concise specification formalisms.

\emph{Weighted automata}, introduced by Sch\"utzenberger \cite{Sch61} (see also \cite{Eil74,KS86,SS78,DKV09}), extend classical automata by assigning weights from a semiring to transitions, which are aggregated along runs. They provide a general framework for modeling quantitative properties such as time~\cite{LBBFHP01}, costs and rewards~\cite{DP16}, probabilities~\cite{Rab63}, and energy consumption~\cite{BBFLMR21}. They have also found important applications in natural language processing and artificial intelligence~\cite{WZZS23}. 
The algebraic structure of weighted automata yields robust closure and decidability properties, making them a powerful tool for quantitative reasoning.
\emph{Weighted regular expressions} provide a complementary specification formalism, and the Kleene--Sch\"utzenberger theorem \cite{Kle56,Sch61} establishes their expressive equivalence with weighted automata.

Despite their expressive power, weighted automata have not yet been systematically applied to financial modeling, where system behavior depends on dynamically evolving external variables. 
Classical models do not directly capture such data-dependent dynamics. 
Bridging this gap enables the transfer of automata-theoretic techniques to financial applications, in particular for scenario-based analysis.

In this paper, we extend the theory of weighted automata and regular expressions to model financial systems driven by external quantitative data. Our framework provides a compositional and formally analyzable approach to scenario-based financial modeling.

{\em Our contributions:}

\begin{itemize}
\item We introduce {\em weighted finite finance automata (WFFA)} over semirings. 
WFFAs are non-deterministic finite automata that process finite sequences of valued letters, i.e., pairs consisting of a symbol from a finite alphabet and a value from a (possibly infinite) data domain. 
In contrast to classical weighted automata, transition weights are arithmetic expressions with variables (so-called {\em $\FF$-expressions}) referring to financial quantities at the time of a transition.
\item We demonstrate the expressiveness of WFFAs through financial examples, including bonds, stocks, European and American options, and limit orders, demonstrating that the framework captures a broad range of financial instruments.
\item We establish closure of WFFAs under standard operations, including sum, Hadamard and Cauchy products, and Kleene star. 
To control the growth of transition expressions, we identify subclasses of bounded size. In particular, we show that monomial
expressions suffice to capture all WFFA behaviors, while general $\FF$-expressions provide a convenient specification formalism.
For certain primitive expressions, the corresponding constructions can be further optimized, ensuring computational tractability in practical applications.
\item We introduce {\em weighted finance regular expressions} and prove a Kleene--Sch\"utzenberger-type theorem~\cite{Sch61,Kle56} establishing expressive equivalence with WFFAs. 
We further show that, when restricted to primitive expressions, weighted finance regular expressions are strictly more expressive than WFFAs, and we identify a restricted fragment that exactly matches WFFA expressiveness.
\item We give an algebraic characterization of WFFAs via matrix representations and derive an efficient evaluation procedure running in time quadratic in the number of states.
\item In the max-plus setting, we prove decidability of the support and threshold problems for WFFAs. These problems are fundamental in the theory of weighted automata (cf.~\cite{KS86,DKV09}). 
In our setting, they determine, respectively, whether any non-trivial outcome is achievable and whether a given payoff level can be exceeded, thereby providing a basis for the formal analysis of financial strategies.
\end{itemize}

{\em Related models:}

\begin{itemize}

\item {\em Binomial trees}~\cite{CRR79,Hul93} are a classical tool in financial engineering for modeling asset price evolution and pricing derivatives. 
They describe a discrete-time branching structure where, at each step, the asset price moves to one of two possible values. 
Such trees can be naturally interpreted as sets of market scenarios, where each path corresponds to a finance word in our setting, while WFFAs provide a mechanism to evaluate quantitative properties along these scenarios.
In contrast to classical binomial models, where price dynamics are fixed in advance, WFFAs allow the use of symbolic transition expressions, enabling more flexible and data-dependent modeling.

\item The {\em discounted cash flow (DCF) model}~\cite{KL05} estimates the value of an investment by aggregating discounted future cash flows.
Several weighted automata models incorporate discounting (cf.~\cite{DK06,DR09}), typically assuming fixed discount factors encoded in the model.
In contrast, WFFAs support dynamically determined discount factors that are provided as part of the input data, for instance derived from observed market rates.

\item WFFAs can be viewed as {\em automata over an infinite alphabet}, since they process symbols paired with values from a potentially infinite domain.
Related models include weighted timed automata~\cite{LBBFHP01} and weighted register automata~\cite{BDP18}, which use clocks or registers to handle data.
While these models achieve high expressiveness, this often comes at a significant computational cost, and classical closure properties may fail.
In contrast, WFFAs remain lightweight and retain key algebraic and algorithmic properties of semiring-weighted automata over finite alphabets.
Moreover, WFFAs are complementary to register-based models: whereas register automata are well-suited for specifying and constraining data-dependent scenarios, WFFAs provide a natural framework for quantitative evaluation of such scenarios.

\end{itemize}

\section{Finance Semirings and Expressions}
\label{SEC:preliminaries}

In this section we introduce the algebraic structures underlying the
definition of WFFAs. In particular, we define \emph{finance semirings} and
\emph{$\FF$-expressions} (financial expressions), which specify quantitative
computations over finance data words and also serve as building blocks
for the weighted finance regular expressions introduced later.

Let $\RR_{\ge 0}$ be the set of non-negative real numbers. For any $k \in \NN$, let ${\NN_{\ge k} = \{n \in \NN  \mid  n \ge k\}}$ and $[k] = \{1, \dots, k\} \subseteq \NN$.

Let $X$ be any non-empty set, which can be infinite. A {\em word} over $X$ is a finite sequence $w = x_1 \dots x_n$ where $n \ge 0$ and $x_1, \dots, x_n \in X$. 
The {\em length} of $w$, denoted by $|w|$, is defined as $|w| = n$. The empty word is denoted by $\eps$, and $|\eps| = 0$. Let $X^*$ denote the set of all words over $X$. If $X$ is finite, then it is called an {\em alphabet}. 

Now, let $\DD$ be a non-empty (possibly infinite) set, referred to as a {\em finance data domain}, and let $\Sigma$ be an alphabet. A {\em finance word} over $\Sigma$ and $\DD$ is a sequence over $\Sigma \times \DD$, i.e., a sequence $w = (a_1, d_1) \dots (a_n, d_n)$ where $n \ge 0$, $a_1, \dots, a_n \in \Sigma$ and $d_1, \dots, d_n \in \DD$. The set of all finance words over $\Sigma$ and $\DD$ is denoted by $\DD \Sigma^* = (\Sigma \times \DD)^*$.

In general, the finance data domain $\DD$ is an arbitrary non-empty set.
In many financial examples considered in this paper, $\DD$ is instantiated
by a subset of $\RR$, since the data values represent quantities
such as prices, rates, or discount factors.

\begin{example}
\label{EX:finance_word}
Let $w = (a_1, d_1) (a_2, d_2) \dots (a_n, d_n) \in \DD \Sigma^*$ be a finance word. The sequence $a_1 \dots a_n$ represents events or actions, while 
$d_1 \dots d_n$ describes the evolution of a financial variable
(e.g., a stock price or exchange rate).

As an example, consider the domain $\DD = \RR_{\ge 0}$ representing a stock price and the alphabet $\Sigma = \{\ll, \mm, \hh\}$, which classifies the stock price as low ($\ll)$, moderate ($\mm$), and high ($\hh$), respectively. For illustration purposes, assume that prices below 1 are classified as low, prices between 1 and 2 as moderate, and prices above 2 as high. Note, however, that in a real-world scenario the classification of a price may be time-dependent.

Then, the finance word $w = (\ll, 0.5) (\hh, 2.5) (\mm, 1.5) (\hh, 2.5)(\ll, 0.5)$ captures the evolution of the stock price over five discrete time units. A graphical representation of $w$ is provided in Figure \ref{FIG:finance_word}.
    
\begin{figure}[ht]
\centering
\scalebox{0.7}{
\begin{tikzpicture}
\node at (0.5, 0.5) [fill,circle,inner sep=0pt,minimum size=3pt, label=below:{\scriptsize $\ll$}] (A) {};
\node at (1.5, 2.5) [fill,circle,inner sep=0pt,minimum size=3pt, label=above:{\scriptsize $\hh$}] (B) {};
\node at (2.5, 1.5) [fill,circle,inner sep=0pt,minimum size=3pt, label=below:{\scriptsize $\mm$}] (C) {};
\node at (3.5, 2.5) [fill,circle,inner sep=0pt,minimum size=3pt, label=above:{\scriptsize $\hh$}] (D) {};
\node at (4.5, 0.5) [fill,circle,inner sep=0pt,minimum size=3pt, label=above:{\scriptsize $\ll$}] (E) {};
    
\node at (-0.2, 0) {\scriptsize{$0$}};
\node at (-0.2, 1) {\scriptsize{$1$}};
\node at (-0.2, 2) {\scriptsize{$2$}};
\node[label=above:{\scriptsize{Price}}] at (0, 2.7) {};
\node[label=right:{\scriptsize{Time}}] at (5, 0) {};
    
\draw (A) -- (B) -- (C) -- (D) -- (E);
\draw[->] (0, 0) -- (5, 0);
\draw[->] (0, 0) -- (0, 2.8);
\draw[dashed] (0, 1) -- (5, 1);
\draw[dashed] (0, 2) -- (5, 2);		
\end{tikzpicture}
}
\caption{Graphical representation of the finance word from Example \ref{EX:finance_word}.}
\label{FIG:finance_word}
\end{figure}
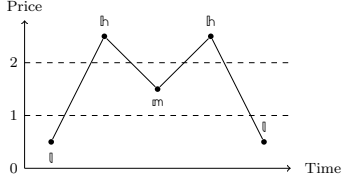
\end{example}

Recall that a {\em monoid} $(M, \cdot, \one)$ consists of a set $M$, a binary operation ${{\cdot}: M \times M \to M}$, which is associative (i.e., $(a \cdot b) \cdot c = a \cdot (b \cdot c)$ for all $a, b, c \in M$), and a unit element $\one \in M$ such that $m \cdot \one = \one \cdot m = m$ for all $m \in M$. The monoid $M$ is called {\em commutative} if $a \cdot b = b \cdot a$ for all $a, b \in M$.
A {\em semiring} is a structure $\SS = (S, \oplus, \otimes, \zero, \one)$ where $(S, \oplus, \zero)$ is a commutative monoid and $(S, \otimes, \one)$ is a monoid, multiplication distributes over addition and $\zero$ is absorbing, i.e.,
$s \otimes \zero = \zero \otimes s = \zero$ for all $s \in S$.
If $(S, \otimes, \one)$ is a commutative monoid, then the semiring $\SS$ is also called {\em commutative}.

Semirings are a widely used computational framework for quantitative systems (see, e.g.,~\cite{DKV09}). Typical examples include:

\begin{itemize}
\item the semiring of real numbers $(\mathbb{R}, +, \cdot, 0.0, 1.0)$ with the usual addition and multiplication;
\item the \emph{arctic (max-plus) semiring}
$
\SS_{\arc} = (\RR \cup \{-\infty\}, \max, +, -\infty, 0.0);
$
\item the \emph{tropical (min-plus)semiring}
$
\SS_{\trop} = (\RR \cup \{\infty\}, \min, +, \infty, 0.0).
$
\end{itemize}

\begin{remark}
We occasionally write $0.0$ and $1.0$ for real constants to emphasize that they are numerical values and to avoid confusion with the abstract semiring elements $\zero$ and $\one$.
\end{remark}

All the semirings introduced above are commutative. 
In the following, the semiring $S$ will serve as the underlying domain
for quantitative values manipulated by the automaton.

Let $\SS = (S, \oplus, \otimes, \zero, \one)$ be a semiring and $\DD$ a finance data domain.
A {\em finance data-binding function} (or simply, a {\em data-binding function}) over $\SS$ and $\DD$ is a mapping $\bb: \DD \times S \to S$.

\begin{example}
\label{EX:data_binding_prod}
Let $\SS = \SS_{\arc}$ and $\DD = \RR_{\ge 0}$. Consider the data-binding function ${\bb_{\times}: \DD \times S \to S}$, defined  by

\begin{equation}
\label{EQ:data_binding_prod}
\bb_{\times}(d, s) = d \times s
\end{equation}
for all $d \in \DD$ and $s \in S$, where $\times$ denotes multiplication extended to $S$ by the convention that $d \times (-\infty) = -\infty$ for all $d \in \RR_{\ge 0}$.

The function $\bb_{\times}$ models standard financial quantities. 
For example, if $d$ represents an annual interest rate and $s$ a principal amount, then $\bb_{\times}(d, s)$ yields the corresponding annual interest. 
If $d$ denotes the price of a stock and $s$ the number of shares, then $\bb_{\times}(d, s)$ gives the total value of the position.
\end{example}

\begin{definition}
\label{DEF:finance_semiring}
A {\em finance semiring} $\FF = (\SS, \DD, \bb)$ consists of a commutative semiring $\SS$, a finance data domain $\DD$, and a data-binding function $\bb: \DD \times S \to S$.
\end{definition}

\begin{remark}
Note that in the definition of a finance semiring we restrict our focus to commutative semirings. This restriction encompasses practically relevant semirings, such as $\SS_{\arc}$ and $\SS_{\trop}$, and is imposed to streamline the presentation.
\end{remark}

\begin{example}
\label{EX:ff_arc_times}
\begin{enumerate}[label=(\alph*)]
\item Let $\bb_{\times}$ be the data-binding function as defined in Example~\ref{EX:data_binding_prod}. Then, $\FF_{\arc, \times} = (\SS_{\arc}, \RR_{\ge 0}, \bb_{\times})$ is a finance semiring. We will use this finance semiring in numerous examples.
\item Symmetrically, we construct the finance semiring
$\FF_{\trop, \times} = (\SS_{\trop}, \RR, \bb_{\times})$,
where $\bb_{\times}$ is defined for all $d \in \RR_{\ge 0}$ and $s \in \RR \cup \{\infty\}$ by Equation~\eqref{EQ:data_binding_prod}, 
and we adopt the convention $d \times \infty = \infty$ for all $d \in \RR_{\ge 0}$.
\item In situation where the financial data may take negative values, the finance semiring $\FF_{\arc, \times}^{\RR} = (\SS_{\arc}, \RR, \bb_{\times}^{\RR})$ is appropriate. The data-binding function ${\bb_{\times}^{\RR}: \RR \times S \to S}$ is defined analogously to Equation~\eqref{EQ:data_binding_prod}, with the data domain extended to $\RR$.
\end{enumerate}
\end{example}

Now we introduce $\FF$-expressions, which are used to specify how
transition weights in WFFAs depend on the data value associated
with a transition.
\begin{definition}
Let $\FF = (\SS, \DD, \bb)$ be a finance semiring with $\SS = (S, \oplus, \otimes, \zero, \one)$.
The set $\EE_\FF$ of {\em $\FF$-expressions} is given by

\[
e \; ::= \; s  \mid  \bind{s}  \mid  e \oplus e  \mid  e \otimes e  \mid  e = e  \mid  e \neq e
\]
where $s \in S$.
\end{definition}
For any $e \in \EE_{\FF}$, the {\em semantics} of $e$ is the mapping $\beh{e}: \DD \to S$ which is defined for every $d \in \DD$ inductively on the structure of $e$ as follows:

{\scriptsize
\[
\begin{aligned}
\beh{s}(d) &= s \\
\beh{\bind{s}}(d) &= \bb(d, s) \\
\beh{e_1 \oplus e_2}(d) &= \beh{e_1}(d) \oplus \beh{e_2}(d) \\
\beh{e_1 \otimes e_2}(d) &= \beh{e_1}(d) \otimes \beh{e_2}(d) \\
\beh{e_1 = e_2}(d) &= \begin{cases}
\one, & \text{if } \beh{e_1}(d) = \beh{e_2}(d), \\
\zero, & \text{otherwise}
\end{cases} \\
\beh{e_1 \neq e_2}(d) &= \begin{cases}
\one, & \text{if } \beh{e_1}(d) \neq \beh{e_2}(d), \\
\zero, & \text{otherwise}
\end{cases}
\end{aligned}
\]
}

Here, $s \in S$ and $e_1, e_2 \in \EE_{\FF}$. To avoid confusion, we will occasionally place $\FF$-expressions in square brackets.

For any $E \subseteq \EE_{\FF}$, let $\beh{E} = \{\beh{e}  \mid  e \in E \}$. The following types of $\FF$-expressions will be important for the developments in Section~\ref{SEC:restricted_wffa}:
\begin{itemize}
\item Let $\bind{S}_{\FF} = \{\bind{s} \mid s \in S \}$. Let $\PP_{\FF} = S \cup \bind{S}_{\FF} \subseteq \EE_{\FF}$ denote the set of all {\em primitive} $\FF$-expressions.
\item We say that $e \in \EE_{\FF}$ is a {\em plain $\FF$-constraint} if $e = [e_1 \bowtie e_2]$ where $e_1, e_2 \in \EE_{\FF}$ are arbitrary (possibly nested) $\FF$-expressions, and ${\bowtie} \in \{{=}, {\neq}\}$. Then, a (composite) {\em $\FF$-constraint} has the form $e = [c_1 \otimes {\dots} \otimes c_k]$ where $k \in \NN_{\ge 1}$ and $c_1, \dots, c_k \in \EE_{\FF}$ are plain $\FF$-constraints. Let $\CC_{\FF}$ denote the set of all $\FF$-constraints.
\item Let $\AA_{\FF} = \{[\llangle s \rrangle \otimes s'] \mid s, s' \in S\} \subseteq \EE_{\FF}$, the collection of {\em affine} $\FF$-expressions.
\end{itemize}
The {\em size} $|e| \in \NN_{\ge 1}$ of an $\FF$-expression $e \in \EE_{\FF}$ is defined inductively as follows.
For any $s \in S$, $|\llangle s \rrangle| = |s| = 1$. For $\op \in \{\oplus, \otimes, =, \neq\}$ and $e_1, e_2 \in \EE_{\FF}$, let  ${|e_1 \op e_2| = |e_1| + |e_2| + 1}$.

\begin{remark}
\label{REM:exp_size}
Assume that the operations $\bb$, $\oplus$, $\otimes$, as well as the comparisons $=$ and $\neq$, can each be evaluated in $\O(1)$ time. Then, for any $e \in \EE_{\FF}$ and $d \in \DD$, the value $\beh{e}(d)$ can be computed in $\O(|e|)$ time. This implies, in particular, that primitive $\FF$-expressions can be evaluated in $\O(1)$ time.
\end{remark}

\begin{example}
\label{EX:ff_arc_times_expressions}
Consider the finance semiring $\FF = \FF_{\arc, \times}$ from Example~\ref{EX:ff_arc_times}(a).
Let $e_1, e_2 \in \EE_{\FF}$. The following derived constructs illustrate the expressive power of $\FF$-expressions.

\begin{enumerate}[label=(\alph*)]

\item \label{EX_ITEM:ff_arc_times_expressions_less_equal}
Define
$
[e_1 \le e_2] := [(e_1 \oplus e_2) = e_2].
$
Then, for all $x \in \DD$, $\beh{e_1 \le e_2}(x) \in \{\zero, \one\}$ and
$
\beh{e_1 \le e_2}(x) = \one \ \text{iff } \beh{e_1}(x) \le \beh{e_2}(x),
$
and $\zero$ otherwise.

\item \label{EX_ITEM:ff_arc_times_expressions_less}
Define
$
[e_1 < e_2] := [(e_1 \oplus e_2) \neq e_1].
$
Then, for all $x \in \DD$, $\beh{e_1 < e_2}(x) \in \{\zero, \one\}$ and
$
\beh{e_1 < e_2}(x) = \one \ \text{iff } \beh{e_1}(x) < \beh{e_2}(x).
$

\item \label{EX_ITEM:ff_arc_times_expressions_min}
Using such guards, we can define
$
\min(e_1, e_2) := [((e_1 \le e_2) \otimes e_1) \oplus ((e_2 \le e_1) \otimes e_2)],
$
and hence
$
\beh{\min(e_1, e_2)}(x) = \min\{\beh{e_1}(x), \beh{e_2}(x)\}.
$
\end{enumerate}

In particular, although the underlying semiring is max-plus, the presence of equality and inequality constraints
 allows us to simulate guards on data values and define operations such as $\min$, which are not available in the pure max-plus setting.
\end{example}

\section{Weighted Finite Finance Automata}
\label{SECTION:WFFA}

Let $\Sigma$ be an alphabet and $\FF = (\SS, \DD, \bb)$ a finance semiring with the underlying commutative semiring $\SS = (S, \oplus, \otimes, \zero, \one)$.

\begin{definition}
\label{DEF:WFFA}
A (non-deterministic) {\em weighted finite finance automaton (WFFA)} over $\Sigma$ and $\FF$ is a tuple $\A = (Q, I, T, F, \wt_I, \wt_T, \wt_F)$ where:
\begin{itemize}
\item $Q$ is a finite set of {\em states},
\item $I, F \subseteq Q$ are sets of {\em initial} resp. {\em final} states,
\item $T \subseteq Q \times \Sigma \times Q$ is a finite set of {\em transitions},
\item $\wt_I: I \to S$ and $\wt_F: F \to S$ are an {\em initial weight function} resp. a {\em final weight function},
\item $\wt_T: T \to \EE_{\FF}$ is a {\em transition weight function} assigning an $\FF$-expression to every transition.
\end{itemize}
\end{definition}

In financial applications, the states $Q$ of a WFFA $\A$ encode the possible stages of a contract or financial process. 
These may represent, for instance, periods between cash flows, phases of a bond or option contract, or different market or credit states.

Let $|\A|_Q = |Q|$, and let
$
|\A|_{\EE_{\FF}} = \max\{|\wt_T(t)| \mid  t \in T\}
$
denote the maximal size of a transition weight occurring in $\A$. If $T = \emptyset$, we set $|\A|_{\EE_{\FF}} = 0$.
We say that $\A$ is {\em purely transition-weighted} if all initial and final weights of $\A$ are equal to $\one$, i.e., $\wt_I(I) = \wt_F(F) = \{\one\}$.
Let $\WFFA_{\Sigma, \FF}$ denote the collection of all WFFAs over $\Sigma$ and $\FF$.
Let ${\WFFA_{\Sigma, \FF}^{\tau} \subseteq \WFFA_{\Sigma, \FF}}$ denote the collection of all purely transition-weighted WFFAs.
Let $t = (q, a, q') \in T$ be any transition with $q, q' \in Q$ and $a \in \Sigma$. Let $\ell(t) = a$, the {\em label} of $t$. We say that $q$ is the {\em source state} of $t$ and $q'$ is the {\em target state} of $t$.
A {\em run of $\A$} is a sequence
\begin{equation}
\label{EQ:wffa_run}
\varrho = \bigl(q_0 \xrightarrow{a_1} q_1 \xrightarrow{a_2} {\dots} \xrightarrow{a_n} q_n \bigr)
\end{equation}
where $q_0 \in I$, $q_n \in F$ and, for every $i \in [n]$, we have $t_i = (q_{i-1}, a_i, q_i) \in T$. The {\em label} of $\varrho$ is the word $\ell(\varrho) = a_1 \dots a_n \in \Sigma^*$.
Let $\Run_{\A}$ denote the set of all runs of $\A$.
A {\em run context} of $\varrho$ is a word of data values $\delta = d_1 \dots d_n \in \DD^*$ describing the evolution of finance data along the run. The {\em weight} of $\varrho$ in context $\delta$ is defined as
\[
\wt_{\A} (\varrho, \delta) = \wt_I(q_0) \otimes \beh{\wt_T(t_1)}(d_1) \otimes {\dots} \otimes \beh{\wt_T(t_n)}(d_n) \otimes \wt_F(q_n).
\]
We say that $\varrho$ is \emph{valid} in run context $\delta$ if $\wt_{\A}(\varrho, \delta) \neq \zero$. Equivalently, $\varrho$ is \emph{invalid} if $\wt_I(q_0) = \zero$ or $\wt_F(q_n) = \zero$ or $\beh{\wt_T(t_i)}(d_i) = \zero$ for some $i \in [n]$. For any finance word $w = (a_1, d_1) \dots (a_n, d_n) \in \DD \Sigma^*$, let $\Run_{\A}(w)$ be the set of all runs $\varrho \in \Run_{\A}$ with label $a_1 \dots a_n$ that are valid in run context $d_1 \dots d_n$. The set is always finite, but it may contain exponentially many runs in $|w|$.

The {\em behavior} of $\A$ is the mapping $\beh{\A}: \DD \Sigma^* \to S$ defined for all finance words $w = (a_1, d_1) \dots (a_n, d_n) \in \DD \Sigma^*$ by 
\[
\beh{\A}(w) = \bigoplus \big(\wt_{\A}(\varrho, d_1 \dots d_n)  \mid  \varrho \in \Run_{\A}(w) \big)
\]

\noindent
where $\bigoplus \emptyset = \zero$. Intuitively, the behavior of $\A$ is the value that the automaton assigns to a given finance word by aggregating the contributions of all valid runs.

For any WFFA collection $\W \subseteq \WFFA_{\Sigma, \FF}$, let $\beh{\W} = \{\beh{\A}  \mid  \A \in \W\}$.

\begin{remark}
Since $S \subseteq \EE_{\FF}$, any WFFA $\A$ with transition weights in $S$ reduces to a standard semiring-weighted automaton. In this case, the behavior of $\A$ is independent of the data values, i.e., $\beh{\A}(w) = \beh{\A}(w')$ for all finance words $w, w'$ with the same underlying event sequence. Hence, $\beh{\A}$ can be viewed as a function $\Sigma^* \to S$, i.e., a formal power series. Thus, WFFAs extend classical weighted automata by allowing data-dependent transition weights.
\end{remark}

The following examples illustrate the expressive power of WFFAs in modelling a range of financial scenarios.

\begin{example}
\label{EX:bond_wffa}
In bond valuation \cite{Fab08}, a sequence of spot rates $s_1, \dots, s_n$ determines discount factors $d_i = \frac{1}{(1+s_i)^i}$. For a bond with coupon $C$ and face value $F$, its present value is
\begin{equation}
\label{EQ:bond_present_value}
PV_{\mathrm{bond}} = C d_1 + \cdots + C d_{n-1} + (C+F)d_n,
\end{equation}
and the profit is $\pi_{\mathrm{bond}} = PV_{\mathrm{bond}} - P_0$, where $P_0$ denotes the purchase price of the bond at time $0$.
Let $\Sigma = \{\texttt{cpn}, \texttt{fin}, \texttt{dfl}\}$. The cash-flow schedule is represented by the finance word
\begin{equation}
\label{EQ:bond_finance_word}
w = (\texttt{cpn}, d_1) \dots (\texttt{cpn}, d_{n-1}) (\texttt{fin}, d_n),
\end{equation}
while a default after $k$ periods yields
$
w' = (\texttt{cpn}, d_1) \dots (\texttt{cpn}, d_k) (\texttt{dfl}, d_{k+1}).
$
We construct a WFFA $\A_{\mathrm{bond}} = (Q, I, T, F, \wt_I, \wt_T, \wt_F) \in \WFFA_{\Sigma, \FF}$ over $\FF = \FF_{\arc, \times}$
with the following components:
\begin{itemize}
\item $Q = \{q_c, q_f, q_d\}$, where $q_c$ represents the state in which an interim coupon payment is received, $q_f$ the state of the final payment, and $q_d$ the default state;

\item $I = \{q_c\}$,\; $F = \{q_f, q_d\}$, \;
$T = \{(q_c, \texttt{cpn}, q_c), (q_c, \texttt{fin}, q_f), (q_c, \texttt{dfl}, q_d)\}$;

\item $\wt_I(q_c) = -P_0$, \; $\wt_F(q_f) = 0.0$, \; $\wt_F(q_d) = 0.0$;

\item $\wt_T(q_c, \texttt{cpn}, q_c) = \llangle C \rrangle$, 
$\wt_T(q_c, \texttt{fin}, q_f) = \llangle C \otimes F \rrangle$,
$\wt_T(q_c, \texttt{dfl}, q_d) = 0.0$.
\end{itemize}
A graphical representation of $\A_{\mathrm{bond}}$ is shown in Figure~\ref{FIG:wffa_bond_ddm_examples} (a).
In the default-free case, $\A_{\mathrm{bond}}$ has a unique run on $w$ with weight
\[
-P_0 + C d_1 + \cdots + C d_{n-1} + (C+F)d_n = \pi_{\mathrm{bond}}.
\]
In the default case,
$
\beh{\A_{\mathrm{bond}}}(w') = -P_0 + C d_1 + \cdots + C d_k.
$
Thus, the WFFA directly evaluates discounted cash flows while remaining independent of any fixed interest-rate model.
\end{example}

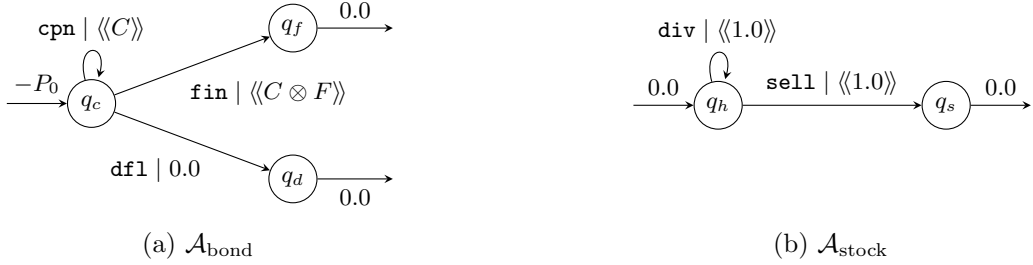
\begin{figure}[ht]
\centering

\begin{minipage}[c]{0.48\textwidth}
\centering
\begin{tikzpicture}[
    >=stealth,
    every node/.style={font=\footnotesize},
    state/.style={circle, draw, minimum size=18pt, inner sep=1pt}
]
\node (q0) at (0,0) {};
\node[state] (q1) at (1.25,0) {$q_c$};
\node[state] (q2) at (3.9,1.0) {$q_f$};
\node (q22) at (5.35,1.0) {};
\node[state] (q3) at (3.9,-1.0) {$q_d$};
\node (q33) at (5.35,-1.0) {};

\draw[->]
    (q0) edge[above] node{$-P_0$} (q1)
    (q1) edge[loop above] node{$\texttt{cpn} \mid \llangle C \rrangle$} (q1)
    (q1) edge node[below right,pos=.42] {$\texttt{fin} \mid \llangle C \otimes F \rrangle$} (q2)
    (q1) edge node[below left,pos=.62] {$\texttt{dfl} \mid 0.0$} (q3)
    (q2) edge[above] node{$0.0$} (q22)
    (q3) edge[below] node{$0.0$} (q33);
\end{tikzpicture}
\end{minipage}
\hfill
\begin{minipage}[c]{0.48\textwidth}
\vspace{-0.9cm}
\centering
\begin{tikzpicture}[
    >=stealth,
    every node/.style={font=\footnotesize},
    state/.style={circle, draw, minimum size=18pt, inner sep=1pt}
]
\node (q0) at (0,0) {};
\node[state] (q1) at (1.25,0) {$q_h$};
\node[state] (q2) at (4.25,0) {$q_s$};
\node (q22) at (5.5,0) {};

\draw[->]
    (q0) edge[above] node{$0.0$} (q1)
    (q1) edge[loop above] node{$\texttt{div} \mid \llangle 1.0 \rrangle$} (q1)
    (q1) edge[above] node{$\texttt{sell} \mid \llangle 1.0 \rrangle$} (q2)
    (q2) edge[above] node{$0.0$} (q22);
\end{tikzpicture}
\end{minipage}
\vfill
\vspace{0.2cm}
\begin{minipage}[c]{0.48\textwidth}
\centering
\small (a) $\A_{\mathrm{bond}}$
\end{minipage}
\hfill
\begin{minipage}[c]{0.48\textwidth}
\centering
\small (b) $\A_{\mathrm{stock}}$
\end{minipage}

\caption{Graphical representations of the WFFAs from Examples~\ref{EX:bond_wffa} and~\ref{EX:ddm_wffa}.}
\label{FIG:wffa_bond_ddm_examples}
\end{figure}

\begin{remark}
\label{REM:bond_wffa_graphical_representation}
We use the following graphical conventions (cf.~Example~\ref{EX:bond_wffa}): states are drawn as circles; initial (resp.~final) states are indicated by incoming (resp.~outgoing) arrows labeled with the corresponding weights; transitions are labeled in the form ``label $|$ weight''. If the alphabet is irrelevant (typically a singleton alphabet $\{\bot\}$), we omit the label and display only the weight.
\end{remark}

\begin{example}
\label{EX:effective_duration}
\emph{Effective duration} \cite[Ch.~15]{Fab08} measures the sensitivity of a bond’s value to small parallel shifts in the spot-rate curve.
Let $\A_{\mathrm{bond}}$ be the WFFA from Example~\ref{EX:bond_wffa} with $P_0 = 0$, so that $\beh{\A_{\mathrm{bond}}}$ computes the bond value. For a finance word $w$, we have $V_0 = \beh{\A_{\mathrm{bond}}}(w)$.
For a small shift $\Delta > 0$, define the perturbed finance words
$
w^{(\pm\Delta)} = (\texttt{cpn}, d_1^{(\pm\Delta)}) \dots (\texttt{fin}, d_n^{(\pm\Delta)}),
$
where $d_i^{(\pm\Delta)} = (1+s_i \pm \Delta)^{-i}$.
Then $V_0^{(\pm\Delta)} = \beh{\A_{\mathrm{bond}}}(w^{(\pm\Delta)})$, and the effective duration is given by
$
\frac{V_0^{(-\Delta)} - V_0^{(+\Delta)}}{2 V_0 \Delta}.
$
This example illustrates that WFFAs naturally support scenario-based sensitivity analysis by evaluating the same automaton on shifted input words.
\end{example}

\begin{example}
\label{EX:ddm_wffa}
The \emph{dividend discount model (DDM)} \cite[Ch.~12]{Fab08} values a stock as the present value of its future cash flows. 
Assume that an investor receives dividends $D_1, \dots, D_n$ and then sells the stock at price $S_{n+1}$. 
With discount rate $r$, the present value is
\[
P_0 = \frac{D_1}{(1+r)} + \dots + \frac{D_n}{(1 + r)^n} + \frac{S_{n+1}}{(1 + r)^{n+1}}.
\]
Let $\Sigma = \{\texttt{div}, \texttt{sell}\}$, where $\texttt{div}$ denotes a dividend payment and $\texttt{sell}$ denotes the sale of the stock. 
We encode discounted cash flows directly in the input by setting
$
d_i = \frac{D_i}{(1+r)^i}
$
($i \in [n]$) and 
$
s_{n+1} = \frac{S_{n+1}}{(1+r)^{n+1}}.
$
Thus, the cash-flow schedule is represented by the finance word
$
w = (\texttt{div}, d_1) \dots (\texttt{div}, d_n)\, (\texttt{sell}, s_{n+1}) \in \DD \Sigma^*.
$
Note that the values $d_i$ and $s_{n+1}$ are not fixed model parameters but part of the input and may vary across scenarios.
Consider the WFFA $\A_{\mathrm{stock}} \in \WFFA_{\Sigma, \FF}$ over $\FF = \FF_{\arc, \times}$ shown in Figure~\ref{FIG:wffa_bond_ddm_examples}~(b), 
where $q_h$ represents holding the stock and $q_s$ the state after selling. 
Then $\beh{\A_{\mathrm{stock}}}(w) = P_0$.
\end{example}

\begin{example}
\label{EX:european_call_option_wffa}
A \emph{European call option} (cf. \cite[Ch.~14]{Fab08}) gives its holder the right to buy an underlying asset 
at \emph{strike price} $K$ at \emph{maturity} $T$. 
The \emph{long-position profit} is
\[
\pi_{\eur}^{\llong}(S_T) = -c_0 + \max\{S_T - K, 0\}.
\]
Here, $S_T$ denotes the underlying price at maturity and $c_0$ the initial option premium. 
We model this payoff using WFFAs over $\FF = \FF_{\arc, \times}$. 
We assume a singleton alphabet $\Sigma = \{\bot\}$ and encode the terminal price $S_T$
as the finance word $\bot_{S_T} = (\bot, S_T)$.

\smallskip
\noindent
\textbf{Two modeling approaches.}
\begin{enumerate}[label=(\alph*)]
\item \emph{Single-expression approach:} encode the payoff using a single $\FF$-expression, e.g.
\begin{equation}
\label{EQ:eur_option_long_expression}
e_{\eur}^{\llong} = \big[[-c_0] \otimes \big((\llangle 1.0 \rrangle \otimes [-K]) \oplus [0.0]\big)\big],
\end{equation}
and use a one-transition automaton.
\item \emph{Automata-theoretic approach:} construct a WFFA with simple primitive expressions from $\PP_{\FF}$ such that the desired behavior is obtained via the aggregation over all runs.
\end{enumerate}

\smallskip
\noindent
We focus on the second approach. 
Consider the WFFA $\A_{\eur}^{\llong}$ shown in Figure~\ref{FIG:european_call_option_wffa_both}~(a). 
The state $q_s$ represents the initial state of the contract, 
$q_e$ corresponds to the exercised case, 
and $q_d$ to the non-exercised (dismissed) case.
For each terminal price $S_T$, the automaton has two runs whose weights are
$-c_0 + S_T - K$ and $-c_0$.
Hence,
$
\beh{\A_{\eur}^{\llong}}(\bot_{S_T})
= \max\{-c_0 + S_T - K,\,-c_0\}
= \pi_{\eur}^{\llong}(S_T).
$

This example illustrates a key advantage of the automata-theoretic approach: 
it avoids large nested $\FF$-expressions and yields a modular and interpretable model.
\end{example}

\begin{figure}[ht]
\centering

\begin{minipage}[c]{0.48\textwidth}
\centering
\begin{tikzpicture}[
    >=stealth,
    every node/.style={font=\footnotesize},
    state/.style={circle, draw, minimum size=18pt, inner sep=1pt}
]
\node (q0) at (0,0) {};
\node[state] (q1) at (1.25,0) {$q_s$};
\node[state] (q2) at (3.25,1.0) {$q_e$};
\node (q22) at (4.5,1.0) {};
\node[state] (q3) at (3.25,-1.0) {$q_d$};
\node (q33) at (4.5,-1.0) {};

\draw[->]
    (q0) edge[above] node{$-c_0$} (q1)
    (q1) edge node[above left,pos=.6] {$\llangle 1.0 \rrangle$} (q2)
    (q1) edge node[below left,pos=.6] {$0.0$} (q3)
    (q2) edge[above] node{$-K$} (q22)
    (q3) edge[below] node{$0.0$} (q33);
\end{tikzpicture}

\vspace{1mm}
\footnotesize (a) $\A^{\llong}_{\eur}$
\end{minipage}
\hfill
\begin{minipage}[c]{0.48\textwidth}
\centering
\begin{tikzpicture}[
    >=stealth,
    every node/.style={font=\footnotesize},
    state/.style={circle, draw, minimum size=18pt, inner sep=1pt}
]
\node (q0) at (0,0) {};
\node[state] (q1) at (1.25,0) {$q_s$};
\node[state] (q2) at (3.25,1.0) {$q_e$};
\node (q22) at (4.5,1.0) {};
\node[state] (q3) at (3.25,-1.0) {$q_d$};
\node (q33) at (4.5,-1.0) {};

\draw[->]
    (q0) edge[above] node{$c_0$} (q1)
    (q1) edge node[below right,pos=.4] {$[K < \llangle 1.0 \rrangle] \otimes \llangle -1.0 \rrangle$} (q2)
    (q1) edge node[below left,pos=.7] {$[\llangle 1.0 \rrangle \le K] \otimes 0.0$} (q3)
    (q2) edge[above] node{$K$} (q22)
    (q3) edge[below] node{$0.0$} (q33);
\end{tikzpicture}

\vspace{1mm}
\footnotesize (b) $\A^{\short}_{\eur}$
\end{minipage}

\caption{WFFAs for the European call option. 
(a) \emph{Long position}: payoff obtained via max-aggregation over runs. 
(b) \emph{Short position}: guarded transitions enforce case distinction $S_T > K$ and $S_T \le K$.}
\label{FIG:european_call_option_wffa_both}
\end{figure}
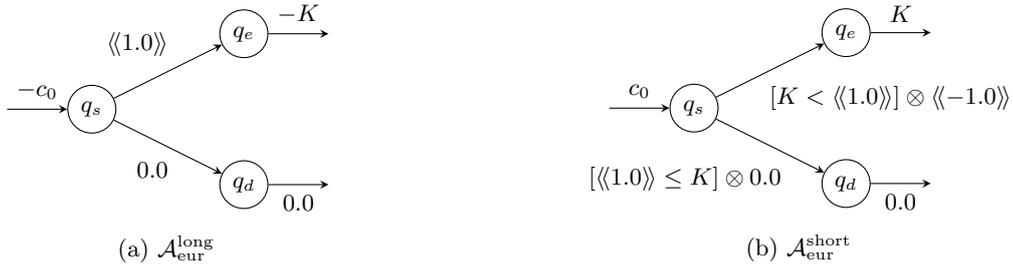

\begin{example}
\label{EX:european_call_option_short_wffa}
We now consider the {\em short-position profit} of a {\em European call option} (cf. \cite[Ch.~14]{Fab08}). 
Recall that
$
\pi_{\eur}^{\short}(S_T) = c_0 - \max\{S_T - K, 0\}.
$
Equivalently, $\pi_{\eur}^{\short}(S_T) = {c_0 + \min\{K - S_T, 0\}}$.
Unlike the long-position case from Example~\ref{EX:european_call_option_wffa}, 
this payoff cannot be obtained directly via max-aggregation over unrestricted runs, 
since the max-plus semiring does not provide a native minimum operation.
There are two ways to address this limitation. 
On the level of $\FF$-expressions, one can encode the $\min$ operator explicitly 
(cf.~Example~\ref{EX:ff_arc_times_expressions}). 
Alternatively, in an automaton-theoretic construction, one can simulate the minimum 
by introducing guards that distinguish the cases $S_T > K$ and $S_T \le K$.

The transition expressions are guarded by the constraints 
$[K < \llangle 1.0 \rrangle]$ and $[\llangle 1.0 \rrangle \le K]$, 
which correspond to the conditions $S_T > K$ and $S_T \le K$, respectively. 
Thus, exactly one of the two runs contributes to the result.
For each $S_T \in \RR_{\ge 0}$, the automaton has two runs whose weights are
$c_0 + (K - S_T)$ and $c_0$.
Due to the guards, only the valid branch is taken, and hence
$
\beh{\A_{\eur}^{\short}}(\bot_{S_T})
= c_0 + \min\{K - S_T, 0\}
= \pi_{\eur}^{\short}(S_T).
$

This example illustrates how $\FF$-constraints compensate for the absence of 
a native minimum operation in the max-plus semiring, enabling precise 
control over branching behavior.
\end{example}

\begin{example}
\label{EX:american_call_option_wffa}
An {\em American call option} gives its holder the right to buy the underlying asset
at strike price $K$ at any time prior to or at maturity.
For $t \in [T]$, where $T$ is the maturity of the option, let $S_t$ denote the underlying price at time $t$, and let $c_0$ be the initial premium.
For a price sequence $S_1,\dots,S_T$, the long-position profit is
\[
\pi_{\amer}^{\llong}(S_1,\dots,S_T)
=
-c_0 + \max\Bigl\{0,\max_{1\le t\le T}(S_t-K)\Bigr\}.
\]
We model this payoff over $\FF=\FF_{\arc,\times}$ using the singleton alphabet $\Sigma=\{\bot\}$.
The price evolution is encoded by the finance word
$
w=(\bot,S_1)\dots(\bot,S_T)\in\DD\Sigma^*.
$
Consider the WFFA $\A_{\amer}^{\llong}\in \WFFA_{\Sigma,\FF}$ shown in Figure~\ref{FIG:american_limit_side_by_side}~(a).
Here, $q_w$ denotes waiting, $q_e$ exercise, and $q_d$ non-exercise (dismissal).
The automaton may remain in $q_w$ for an arbitrary number of time steps and move to $q_e$ to exercise.
Alternatively, if no exercise occurs before maturity, it moves to $q_d$ at the final step.
For each $t\in [T]$, a run that stays in $q_w$ for the first $t-1$ steps and then moves to $q_e$
at time $t$ has weight $-c_0 + S_t - K.$
The run that moves to $q_d$ yields weight $-c_0$.
Hence,
\[
\beh{\A_{\amer}^{\llong}}(w)
=
-c_0 + \max\Bigl\{0,\max_{1\le t\le T}(S_t-K)\Bigr\}
=
\pi_{\amer}^{\llong}(S_1,\dots,S_T).
\]
\end{example}

\begin{figure}[ht]
\centering

\begin{minipage}[c]{0.48\textwidth}
\vspace{0.1cm}
\centering
\begin{tikzpicture}[
    >=stealth,
    every node/.style={font=\footnotesize},
    state/.style={circle, draw, minimum size=18pt, inner sep=1pt}
]
\node (q0) at (0,0) {};
\node[state] (q1) at (1.25,0) {$q_w$};
\node[state] (q2) at (3.9,1.0) {$q_e$};
\node (q22) at (5.2,1.0) {};
\node[state] (q3) at (3.9,-1.0) {$q_d$};
\node (q33) at (5.2,-1.0) {};

\draw[->]
    (q0) edge[above] node{$-c_0$} (q1)
    (q1) edge[loop above] node{$0.0$} (q1)
    (q1) edge node[above left,pos=.65] {$\llangle 1.0 \rrangle$} (q2)
    (q1) edge node[below left,pos=.65] {$0.0$} (q3)
    (q2) edge[loop above] node{$0.0$} (q2)
    (q2) edge[above] node{$-K$} (q22)
    (q3) edge[below] node{$0.0$} (q33);
\end{tikzpicture}

\vspace{1mm}
\footnotesize (a) $\A_{\amer}^{\llong}$
\end{minipage}
\hfill
\begin{minipage}[c]{0.48\textwidth}
\centering
\begin{tikzpicture}[
    >=stealth,
    every node/.style={font=\footnotesize},
    state/.style={circle, draw, minimum size=18pt, inner sep=1pt}
]
\node (q0) at (0,0) {};
\node[state] (q1) at (1.25,0) {$q_w$};
\node[state] (q2) at (3.9,1.0) {$q_e$};
\node (q22) at (5.2,1.0) {};
\node[state] (q3) at (3.9,-1.0) {$q_c$};

\draw[->]
    (q0) edge[above] node{$0.0$} (q1)
    (q1) edge[loop above] node{$\texttt{a}\mid [\llangle 1.0 \rrangle > 50]$} (q1)
    (q1) edge node[below right,pos=.55] {$\texttt{a}\mid [\llangle 1.0 \rrangle \le 50]\otimes \llangle 10\rrangle$} (q2)
    (q1) edge node[below left,pos=.55] {$\texttt{c}\mid 0.0$} (q3)
    (q2) edge[loop above] node{$\texttt{a}\mid 0.0$} (q2)
    (q2) edge[above] node{$0.0$} (q22);
\end{tikzpicture}

\vspace{1mm}
\footnotesize (b) $\A_{\mathrm{limit}}$
\end{minipage}

\caption{WFFAs for (a) an American call option and (b) a buy limit order.}
\label{FIG:american_limit_side_by_side}
\end{figure}
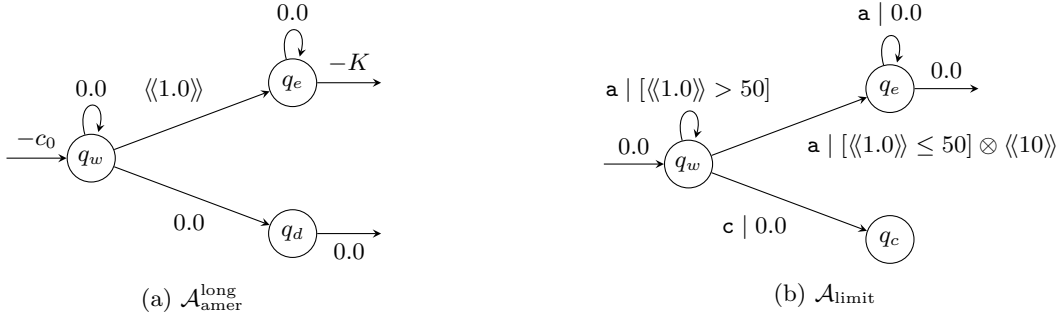

\begin{example}
\label{EX:limit_order_wffa}
A \emph{buy limit order} \cite{FPR13} executes only when the market price is at or below a prescribed limit price. 
We model a simplified buy limit order for $10$ shares with limit price $50$, assuming perfect liquidity and allowing the investor to cancel the order before execution.
Let $\Sigma = \{\texttt{i}, \texttt{c}\}$, where $\texttt{a}$ denotes that the order remains active and $\texttt{c}$ that the investor cancels it. 
We work over the finance semiring $\FF=\FF_{\arc,\times}$ and encode the joint evolution of market price and order status as a finance word in $\DD\Sigma^*$.
Consider the WFFA $\A_{\mathrm{limit}} \in \WFFA_{\Sigma,\FF}$ shown in Figure~\ref{FIG:american_limit_side_by_side}~(b). 
The state $q_w$ denotes waiting for execution, $q_e$ execution, and $q_c$ cancellation. 
If the order is executed, the automaton returns the total acquisition cost; if it is never executed, the value is $-\infty$.

For example, for
$
w_1 = (\texttt{a}, 51)\,(\texttt{a}, 53)\,(\texttt{a}, 48)\,(\texttt{a}, 46),
$
the order is executed at price $48$, and hence
$
\beh{\A_{\mathrm{limit}}}(w_1)=480.
$
For
$
w_2 = (\texttt{a}, 51)\,(\texttt{a}, 53)\,(\texttt{a}, 51)\,(\texttt{a}, 52),
$
the price never reaches the limit, so $\beh{\A_{\mathrm{limit}}}(w_2)=-\infty$. 
For
$
w_3 = (\texttt{a}, 51)\,(\texttt{c}, 50),
$
the order is cancelled before execution, and therefore
$
\beh{\A_{\mathrm{limit}}}(w_3)=-\infty.
$
\end{example}

\section{Weighted Finance Regular Expressions}
\label{SEC:weighted_regexp}

Let $\Sigma$ be an alphabet and $\FF = (\SS, \DD, \bb)$ be a finance semiring, where $\SS = (S, \oplus, \otimes, \zero, \one)$ is a commutative semiring.
A \emph{quantitative finance language (QFL)} over $\Sigma$ and $\FF$ is a mapping $\F: \DD \Sigma^* \to S$. 
We say that $\F$ is {\em proper} if $\F(\eps) = \zero$. 
Let $\QFL_{\Sigma, \FF}$ denote the collection of all QFLs over $\Sigma$ and $\FF$.

From an economic perspective, a QFL $\F$ assigns a quantitative outcome to each possible market scenario. 
A finance word $w \in \DD\Sigma^*$ represents a discrete-time evolution of relevant market variables together with observable events, while the value $\F(w)$ corresponds to a financial quantity of interest, such as a payoff, a cost, or a risk measure. 
Thus, QFLs provide a formal, scenario-based view of financial systems, capturing how the outcome of a financial instrument or strategy depends on the underlying market trajectory. 
This perspective naturally encompasses key tasks in quantitative finance, including scenario analysis, stress testing, and the evaluation of trading strategies under varying market conditions.

Given $\A \in \WFFA_{\Sigma, \FF}$, the behavior $\beh{\A}$ is a QFL over $\Sigma$ and $\FF$. Hence,
$
\beh{\WFFA_{\Sigma, \FF}} \subseteq \QFL_{\Sigma, \FF}.
$
That is, WFFAs define a subclass of QFLs that admit a finite-state, operational representation.

While QFLs provide a semantic view of financial systems, such mappings are typically too complex to be specified explicitly in a compact form.
Instead, one seeks structured and compositional formalisms for defining them. 
WFFAs provide one such formalism. 
In this section, we introduce an alternative, expression-based approach for defining QFLs, namely {\em weighted finance regular expressions}.

To define their semantics, we first lift the standard operations on formal power series (see, e.g., \cite{DKV09}) to the setting of quantitative finance languages.

\begin{definition}
\label{DEF:qfl_operations}
Let $\F_1, \F_2 \in \QFL_{\Sigma, \FF}$.
\begin{enumerate}[label=(\alph*)]
\item The {\em sum} $(\F_1 \oplus \F_2) \in \QFL_{\Sigma, \FF}$ is defined for all $w \in \DD \Sigma^*$ by
$
(\F_1 \oplus \F_2)(w) = \F_1(w) \oplus \F_2(w).
$
\item The {\em Hadamard product} $(\F_1 \otimes \F_2) \in \QFL_{\Sigma, \FF}$ is defined for all $w \in \DD \Sigma^*$ by $(\F_1 \otimes \F_2)(w) = \F_1(w) \otimes \F_2(w)$.
\item The {\em Cauchy product} $(\F_1 \cdot \F_2) \in \QFL_{\Sigma, \FF}$ is defined for all $w \in \DD \Sigma^*$ by
\[
{(\F_1 \cdot \F_2)(w) = \bigoplus \big(\F_1(w_1) \otimes \F_2(w_2)  \mid  w_1, w_2 \in \DD \Sigma^* \text{ and } w_1 \cdot w_2 = w \big)}
\]
where ${w_1 \cdot w_2}$ denotes the concatenation of finance words.
\item Let $\F \in \QFL_{\Sigma, \FF}$ be proper. The {\em Kleene star} $\F^* \in \QFL_{\Sigma, \FF}$ is defined by
\[
{\F^*(w) = \bigoplus_{k \in \NN} \left(\bigotimes_{i = 1}^k \F(w_i)  \mid  w_1, \dots, w_k \in \DD \Sigma^* \text{ and } w_1 \cdot {\dots} \cdot w_k = w \right)}
\]
for all $w \in \DD \Sigma^*$, where the empty product (for $k = 0$) is defined to be $\one$.
\end{enumerate}
\end{definition}

\begin{remark}
Note that the infinite sum in the equation for $\F^*(w)$ has only finitely many non-zero terms, since $\F$ is proper. Hence, $\F^*(w)$ is well defined.
\end{remark}

\begin{remark}
The operations from Definition~\ref{DEF:qfl_operations} admit a natural interpretation in quantitative finance.
The sum $\F_1 \oplus \F_2$ models the aggregation of alternative outcomes, for instance selecting the better
payoff in max-plus settings or combining competing strategies. 
The Hadamard product $\F_1 \otimes \F_2$ combines the values of $\F_1$ and $\F_2$ for the same market scenario.
In the max-plus case, this operation 
reduces to additive composition, so that combining with a constant corresponds 
to shifting the resulting value, e.g., modeling fixed costs or premiums.
The Cauchy product $\F_1 \cdot \F_2$ captures the sequential composition of financial mechanisms: it aggregates all possible decompositions of a scenario into two consecutive phases, reflecting, for example, multi-period contracts or staged investment decisions. Finally, the Kleene star $\F^*$ models iteration and repetition, allowing for an unbounded number of consecutive applications of a financial mechanism, as in repeated trading strategies or rolling contracts.

Together, these operations provide a compositional toolkit for constructing complex financial models from simpler building blocks.
\end{remark}

\begin{definition}
\label{DEF:regexp}
A {\em weighted finance regular expression} $\R$ over $\Sigma$ and $\FF$ is derived from the grammar
\[
\R \; ::= \; s_{\eps}  \mid  e_a  \mid  \R \oplus \R  \mid  \R \cdot \R  \mid  \R^*
\]
where $s \in S$, $e \in \EE_{\FF}$ and $a \in \Sigma$.
\end{definition}

The {\em semantics} of $\R$ is the QFL $\beh{\R}: \DD \Sigma^* \to S$ defined for all $w \in \DD \Sigma^*$ inductively on the structure of $\R$ as follows:
{\scriptsize
\[
\begin{aligned}
\beh{s_{\varepsilon}}(w) &= \begin{cases}
s, & \text{if } w = \varepsilon, \\
\zero, & \text{otherwise}
\end{cases} \\
\beh{e_{a}}(w) &= \begin{cases}
\beh{e}(d), & \text{if } w = (a, d) \text{ with } a \in \Sigma \text{ and } d \in \DD \\
\zero, & \text{otherwise.}
\end{cases} \\
\beh{\R_1 \oplus \R_2}(w) &= (\beh{\R_1} \oplus \beh{\R_2})(w) \\
\beh{\R_1 \cdot \R_2}(w) &= (\beh{\R_1} \cdot \beh{\R_2})(w) \\
\beh{(\R')^*}(w) &= \beh{\R'}^*(w) 
\end{aligned}
\]
}

\noindent where $s \in S$, $e \in \EE_{\FF}$, $a \in \Sigma$ and $\R_1, \R_2, \R'$ are subexpressions of $\R$. Here, if $\beh{\R'}$ is not proper, then the semantics $\beh{\R'}^*$ may be undefined, since the definition of the star involves infinite sums, which need not exist in general semirings. A regular expression $\R \in \Reg_{\Sigma, \FF}$ is called {\em valid} if, for every subexpression $(\R')^*$ of $\R$, $\beh{\R'}$ is a proper QFL. Let $\Reg_{\Sigma, \FF}$ denote the set of all valid weighted finance regular expressions over $\Sigma$ and $\FF$. In the rest of this section we consider only valid expressions, so that the semantics is well defined. 
For any collection of weighted finance regular expressions $\W \subseteq \Reg_{\Sigma, \FF}$, let $\beh{\W} = \{\beh{\R} \mid \R \in \W\}$.

\begin{remark}
In contrast to the classical theory of weighted regular expressions (see, e.g., \cite{DKV09}), 
the Hadamard product is not included as a primitive syntactic operation in the grammar above, since it can be expressed using the available constructs.
Nevertheless, it is convenient to use the Hadamard product at the semantic level, 
as it provides a concise way to express pointwise combinations of quantitative values.
\end{remark}

\begin{example}
If $E = S$, then $\Reg_{\Sigma, \FF}$ corresponds to the standard semiring-weighted regular expressions of Sch\"utzenberger \cite{Sch61} (cf. also \cite{DKV09,DK21}).
\end{example}

\begin{example}
\label{EX:bond_regexp}
Consider the bond valuation model from Example~\ref{EX:bond_wffa}. 
We recall the alphabet 
$
\Sigma = \{\texttt{cpn}, \texttt{fin}, \texttt{dfl}\},
$
where \texttt{cpn}, \texttt{fin}, and \texttt{dfl} denote coupon, final, and default events, respectively.
We also recall the finance semiring $\FF = \FF_{\arc, \times}$.
The bond cash-flow structure can be described by the weighted finance regular expression
\[
\R_{\mathrm{bond}} = [-P_0]_{\varepsilon} \cdot (\llangle C \rrangle_{\texttt{cpn}})^* 
\cdot \big(\llangle C \otimes F \rrangle_{\texttt{fin}} \oplus [0.0]_{\texttt{dfl}}\big).
\]
For a finance word $w \in \DD \Sigma^*$ of the form~\eqref{EQ:bond_finance_word}, 
we obtain $\beh{\R_{\mathrm{bond}}}(w) = \pi_{\mathrm{bond}}$.
\end{example}

\begin{example}
\label{EX:eur_call_option_regexp}
In this example, we construct a regular expression modeling the profit of a long position 
in the European call option from Example~\ref{EX:european_call_option_wffa}. 
Let $\Sigma = \{\bot\}$ and $\FF = \FF_{\arc, \times}$. 
Let $e^{\llong}_{\eur} \in \EE_{\FF}$ be defined as in Equation~\eqref{EQ:eur_option_long_expression}. 
The corresponding regular expression
$
\R^{\llong}_{\eur} = [e^{\llong}_{\eur}]_{\bot} \in \Reg_{\Sigma, \FF}
$
directly encodes the payoff function.
An equivalent expression using only primitive $\FF$-expressions can be constructed as
\[
\tilde{\R}^{\llong}_{\eur} = [-c_0]_{\eps} \cdot 
  \bigl((\bind{1.0}_{\bot} \cdot [-K]_{\eps}) \oplus [0.0]_{\bot}\bigr),
\]
where $\tilde{\R}^{\llong}_{\eur} \in \Reg_{\Sigma, \FF}(\PP_{\FF})$. 
For all $S_T \in \RR_{\ge 0}$, we have
$
\beh{\tilde{\R}^{\llong}_{\eur}}(\bot_{S_T}) = \pi^{\llong}_{\eur}(S_T).
$
\end{example}

\begin{example}
\label{EX:amer_call_option_regexp}
Consider the American call option from Example~\ref{EX:american_call_option_wffa}. 
We construct a weighted finance regular expression over 
$\Sigma = \{\bot\}$ and $\FF = \FF_{\arc, \times}$ describing the payoff of a long position.
The expression is given by
\[
\R_{\amer} = [-c_0]_{\varepsilon} \cdot 
\bigl([0.0]_{\bot})^* \cdot 
\bigl([0.0]_{\varepsilon} \oplus (\llangle 1.0 \rrangle_{\bot} \cdot ([0.0]_{\bot})^*)\bigr)\bigr).
\]
Intuitively, the factor $([0.0]_{\bot})^*$ models waiting, allowing the option holder 
to postpone the exercise arbitrarily. The choice 
$[0.0]_{\varepsilon} \oplus (\llangle 1.0 \rrangle_{\bot} \cdot ([0.0]_{\bot})^*)$ 
captures the decision either not to exercise the option or to exercise it at a given time step. 
The term $[-c_0]_{\varepsilon}$ accounts for the initial premium.
For all input words $w = (\bot, S_1)\dots(\bot, S_T) \in \DD\Sigma^*$, 
the expression $\R_{\amer}$ yields the payoff of the American call option.
\end{example}
These examples demonstrate that weighted finance regular expressions 
provide a compositional framework for modeling financial payoffs.

Our first main result is the following theorem.
\begin{theorem}[Kleene–Sch\"utzenberger Theorem for QFLs]
\label{THM:kleene_schuetzenberger}
For every QFL ${\F : \DD \Sigma^* \to S}$, the following are equivalent:
\begin{enumerate}[label=(\alph*)]
    \item $\F$ is definable by a weighted finance regular expression, i.e., $\F \in \beh{\Reg_{\Sigma,\FF}}$;
    \item $\F$ is recognizable by a WFFA, i.e., $\F \in \beh{\WFFA_{\Sigma,\FF}}$.
\end{enumerate}
Moreover, the equivalence is effective: from a weighted finance regular expression one can effectively construct a WFFA defining the
same QFL, and conversely, from a WFFA one can effectively construct a weighted finance regular expression defining the same QFL.
\end{theorem}

\begin{remark}	
The proof of Theorem~\ref{THM:kleene_schuetzenberger} is deferred to Section~\ref{SEC:proof_kleene_schuetzenberger}, where we present a stronger version of the theorem. This stronger result provides additional technical insights into the constructions, and Theorem~\ref{THM:kleene_schuetzenberger} can be seen as a special case of that more general statement.
\end{remark}

\begin{remark}[Relation to classical Kleene--Sch\"utzenberger results]
\label{REM:ks_relation}
The construction in Theorem~\ref{THM:kleene_schuetzenberger} follows the classical proof scheme of the Kleene--Sch\"utzenberger theorem for weighted automata. One may ask whether this result can be derived directly from classical equivalence results between weighted automata and rational expressions over semirings.

Such a reduction can be formulated by encoding WFFAs as weighted automata over suitable semirings of functions. For instance, transition weights may be viewed as local weight functions
$
f : D \times \Sigma \to S,
$
describing the contribution of a transition to the weight of a valued symbol. Runs of the automaton then correspond to finite sequences of such functions, whose evaluation along a finance word induces a quantitative finance language
$
\F: (D \times \Sigma)^* \to S .
$

More generally, one may consider semirings of multisets of finite sequences of such local weight functions. This yields a multiset-based semantics in the spirit of unifying frameworks for weighted automata and weighted logics (see, e.g., \cite{GM18,DP16}). In this way, classical equivalence results between weighted automata and rational expressions over suitable semirings (see, e.g., \cite[Theorem~2.12]{DKV09}) can be related to the present setting.

However, such encodings lead to rather technically involved semirings whose elements no longer admit a natural
financial interpretation. For this reason, we prefer to give direct constructions tailored to WFFAs.
These constructions are based on standard automata-theoretic techniques, such as product and state-elimination constructions, but are carried out directly within the WFFA model.
Theorem~\ref{THM:kleene_schuetzenberger}
therefore provides the conceptual backbone of our framework and serves as a basis for the later study of
large but syntactically restricted classes of finance expressions.
\end{remark}

\section{Closure Properties of WFFAs}
\label{SEC:closure_wffa}

In this section, we study the closure properties of QFLs captured by WFFAs. 
These results will be used in Section~\ref{SEC:proof_kleene_schuetzenberger} 
to prove our Kleene--Sch\"utzenberger theorem. 
They are also of independent interest, as they provide a compositional view 
of financial systems in terms of simpler building blocks.

Thanks to our definition of $\FF$-expressions, the automata constructions 
establishing these closure properties closely resemble those for standard 
semiring-weighted automata. 
We present these constructions in detail, as they will also be used in 
Section~\ref{SEC:restricted_wffa} to analyze WFFAs with restricted expressions.

Let $\Sigma$ be an alphabet and $\FF = (\SS, \DD, \bb)$ a finance semiring where ${\SS = (S, \oplus, \otimes, \zero, \one)}$ is a commutative semiring.

The following theorem establishes that WFFAs share the same closure properties as weighted automata over commutative semirings (see, e.g., \cite{DKV09}) with respect to the operations introduced in Definition~\ref{DEF:qfl_operations}.

\begin{theorem}
\label{THM:closure_wffa}
Let $\F_1, \F_2, \F \in \beh{\WFFA_{\Sigma, \FF}}$ such that $\F$ is proper. Then, the QFLs $\F_1 \oplus \F_2$, $\F_1 \otimes \F_2$, $\F_1 \cdot \F_2$ and $\F^*$ are also in $\beh{\WFFA_{\Sigma, \FF}}$.
\end{theorem}

\begin{example}
\label{EX:bull_spread}
A \emph{bull spread} \cite{Nat14} is constructed by taking a long position in a call option 
with strike $K^{\mathrm{low}}$ and a short position in a call option 
with strike $K^{\mathrm{high}} > K^{\mathrm{low}}$.
Let $\pi_{\mathrm{low}}^{\llong}$ and $\pi_{\mathrm{high}}^{\short}$ denote the corresponding payoffs, 
and let $\F(\pi)$ be their encoding as QFLs over $\Sigma = \{\bot\}$ and $\FF = \FF_{\arc,\times}$. 
Then the payoff of the bull spread satisfies
\[
\F(\pi_{\mathrm{bull}}) 
= \F(\pi_{\mathrm{low}}^{\llong}) \otimes \F(\pi_{\mathrm{high}}^{\short}).
\]
Since both components are recognizable by WFFAs (Example~\ref{EX:european_call_option_wffa}), 
Theorem~\ref{THM:closure_wffa} implies that the bull spread is also recognizable.
This illustrates the compositional modeling of financial strategies via closure properties.
\end{example}

In the remainder of this subsection, we present a proof of Theorem~\ref{THM:closure_wffa}.

\begin{lemma}
\label{LEMMA:closure_sum_wffa}
Let $\A_1, \A_2 \in \WFFA_{\Sigma, \FF}$. Then there exists $\A_{\oplus} \in \WFFA_{\Sigma, \FF}$ such that 
$\beh{\A_{\oplus}} = \beh{\A_1} \oplus \beh{\A_2}$.
\end{lemma}
\begin{proof}
Let $\A_1$ and $\A_2$ be WFFAs with state sets $Q_1$ and $Q_2$, respectively.
We use the standard disjoint union construction. 
Assume $Q_1 \cap Q_2 = \emptyset$ and define $\A_{\oplus}$ by taking the union of the components of $\A_1$ and $\A_2$. 
Then every run of $\A_{\oplus}$ is a run of either $\A_1$ or $\A_2$, with weights preserved. 
Hence, for all $w \in \DD\Sigma^*$,
$
\beh{\A_{\oplus}}(w) = \beh{\A_1}(w) \oplus \beh{\A_2}(w).
$
\end{proof}

\begin{lemma}
\label{LEMMA:closure_hadamard_wffa}
Let $\A_1, \A_2 \in \WFFA_{\Sigma, \FF}$. Then there exists 
$\A_{\otimes} \in \WFFA_{\Sigma, \FF}$ such that 
$\beh{\A_{\otimes}} = \beh{\A_1} \otimes \beh{\A_2}$.
\end{lemma}
\begin{proof}
We use the standard synchronous product construction for weighted automata
over commutative semirings (see, e.g., \cite{DKV09}). Let
$
\A_i = (Q_i, I_i, T_i, F_i, \wt_{I_i}, \wt_{T_i}, \wt_{F_i}) \; (i \in \{1,2\}).
$
We define
$
Q = Q_1 \times Q_2,
$
$
I = I_1 \times I_2,
$
$
F = F_1 \times F_2
$
and let $T$ consist of all synchronized transitions
$
{t = ((q_1, q_2), a, (q_1', q_2'))}
$
such that $(q_1, a, q_1') \in T_1$ and ${(q_2, a, q_2') \in T_2}$.
The initial and final weights are defined componentwise by
$
{\wt_I(i_1,i_2)=\wt_{I_1}(i_1)\otimes \wt_{I_2}(i_2)},
$
for $i_1 \in I_1$ and ${i_2 \in I_2}$, and
$
\wt_F(f_1,f_2)=\wt_{F_1}(f_1)\otimes \wt_{F_2}(f_2),
$
for ${f_1 \in F_1}$, $f_2 \in F_2$.
The weight of a synchronized transition
${t = ((q_1, q_2), a, (q_1', q_2')) \in T}$ is
$
{\wt_T(t)=\wt_{T_1}(q_1,a,q_1') \otimes \wt_{T_2}(q_2,a,q_2')},
$
where $\otimes$ denotes the syntactic combination of $\FF$-expressions.
This is well defined since $\EE_{\FF}$ is closed under $\otimes$.
Then runs of $\A_{\otimes}$ are in one-to-one correspondence with pairs of runs
of $\A_1$ and $\A_2$ on the same word, and their weights evaluate to the product
of the corresponding run weights. Hence, for every $w \in \DD\Sigma^*$,
$
\beh{\A_{\otimes}}(w)=\beh{\A_1}(w)\otimes \beh{\A_2}(w).
$
\end{proof}

\begin{definition}
\label{DEF:normalized_wffa}
A WFFA $\A = (Q, I, T, F, \wt_I, \wt_T, \wt_F)$ is called:
\begin{enumerate}[label=(\alph*)]
\item {\em initially normalized} if $I = \{q_I\}$, $\wt_I(q_I)=\one$, and no transition enters $q_I$;
\item {\em finally normalized} if $F = \{q_F\}$, $\wt_F(q_F)=\one$, and no transition leaves $q_F$;
\item {\em completely normalized} if it is both initially and finally normalized.
\end{enumerate}
\end{definition}

\begin{lemma}
\label{LEMMA:normalize_wffa}
Let $\A \in \WFFA_{\Sigma, \FF}$. Then:
\begin{enumerate}[label=(\alph*)]
\item there exists an initially normalized $\A_I$ with $\beh{\A_I}=\beh{\A}$;
\item there exists a finally normalized $\A_F$ with $\beh{\A_F}=\beh{\A}$;
\item if $\beh{\A}$ is proper, there exists a completely normalized $\A_{I,F}$ with $\beh{\A_{I,F}}=\beh{\A}$.
\end{enumerate}
\end{lemma}

\begin{proof}
We use standard normalization constructions (cf., e.g.,~\cite{DKV09}). 
The only non-trivial aspect is the redistribution of initial and final weights to transitions. 
In our setting, this is achieved by lifting weights to $\FF$-expressions: an initial weight $s$ can be incorporated into outgoing transitions as $s \otimes w$, and dually for final weights. 
Moreover, closure of $\EE_{\FF}$ under $\oplus$ allows us to combine contributions from multiple initial (or final) states.
For initial normalization, introduce a fresh state $q_I$ and redirect all transitions from initial states.  If necessary, adjust the final weight of $q_I$ to preserve the value on $\eps$. Final normalization is dual.
For complete normalization, assuming that the automaton is proper, introduce fresh states $q_I, q_F$ and combine the two constructions.
\end{proof}

\begin{lemma}
\label{LEMMA:cauchy_product_kleene_star_normalized_wffa}
\begin{enumerate}[label=(\alph*)]
\item 
Let $\A_F \in \WFFA_{\Sigma, \FF}$ be finally normalized and $\A_I \in \WFFA_{\Sigma, \FF}$ be initially normalized. 
Then there exists $\A_{\mathrm{Cauchy}} \in \WFFA_{\Sigma, \FF}$ such that
$\beh{\A_{\mathrm{Cauchy}}}=\beh{\A_F}\cdot\beh{\A_I}$.

\item 
Let $\A \in \WFFA_{\Sigma, \FF}$ be completely normalized. 
Then there exists $\A^* \in \WFFA_{\Sigma, \FF}$ with $\beh{\A^*}=\beh{\A}^*$.
\end{enumerate}
\end{lemma}

\begin{proof}
We use standard constructions for weighted automata (see, e.g.,~\cite{DKV09}), adapted to $\FF$-expressions.

\begin{enumerate}[label=(\alph*)]
\item 
Identify the final state of $\A_F$ with the initial state of $\A_I$ and take the union of the remaining components. 
Then every run of the resulting automaton is obtained by concatenating a run of $\A_F$ with a run of $\A_I$. 
The weight of such a run is the product of the corresponding run weights, since the connecting state carries no additional weight. 
Hence, the behavior of the constructed automaton coincides with the Cauchy product.

\item 
Assume that $\A$ is completely normalized, i.e., it has a unique initial state and a unique final state. 
Identify the final state with the initial state, and redirect all transitions entering the final state to the initial state, thereby creating loops. 
Then every run corresponds to a finite concatenation of runs of $\A$. 
The weight of such a run is the product of the weights of the corresponding subruns. 
Hence, the behavior of the constructed automaton coincides with the Kleene star of $\beh{\A}$. \qedhere
\qedhere
\end{enumerate}
\end{proof}

As a corollary from Lemmas~\ref{LEMMA:normalize_wffa} and \ref{LEMMA:cauchy_product_kleene_star_normalized_wffa}, we obtain:

\begin{corollary}
\label{COR:cauchy_product_kleene_star_full_wffa}
\begin{enumerate}[label=(\alph*)]
\item For all $\A_1,\A_2 \in \WFFA_{\Sigma,\FF}$, there exists $\A' \in \WFFA_{\Sigma,\FF}$ such that 
$\beh{\A'}=\beh{\A_1}\cdot\beh{\A_2}$.
\item For all $\A \in \WFFA_{\Sigma, \FF}$ such that $\beh{\A}$ is proper, there exists $\A^* \in \WFFA_{\Sigma, \FF}$ such that 
$\beh{\A^*}=\beh{\A}^*$.
\end{enumerate}
\end{corollary}

\begin{proof}[Proof of Theorem~\ref{THM:closure_wffa}]
The theorem follows by combining Lemmas~\ref{LEMMA:closure_sum_wffa} and \ref{LEMMA:closure_hadamard_wffa} with Corollary~\ref{COR:cauchy_product_kleene_star_full_wffa}. 
\end{proof}

Recall that $\WFFA_{\Sigma, \FF}^{\tau}$ denotes the class of purely transition-weighted WFFAs. As an immediate corollary of Lemma~\ref{LEMMA:normalize_wffa}, we obtain:
\begin{corollary}
$\beh{\WFFA_{\Sigma, \FF}} = \beh{\WFFA_{\Sigma, \FF}^{\tau}}$.
\end{corollary}
Intuitively, this means that any WFFA can be transformed into an equivalent purely transition-weighted form without changing its behavior.

\begin{remark}
\label{REM:computational_complexity_of_expressions}
While the closure constructions for WFFAs follow the same principles as for weighted automata over semirings, they may lead to a significant increase in the size of $\FF$-expressions.
In particular, normalization (Lemma~\ref{LEMMA:normalize_wffa}) can yield expressions of size $\O((|\A|_Q)^2 \cdot |\A|_{\EE_{\FF}})$, which may be costly to evaluate.
\end{remark}

\section{Restricted WFFAs}
\label{SEC:restricted_wffa}

As discussed in Remark~\ref{REM:computational_complexity_of_expressions}, 
allowing unrestricted $\FF$-expressions in WFFAs can lead to a significant increase in the size and evaluation cost of transition-weight expressions. 
This becomes particularly problematic when the automaton's behavior must be evaluated repeatedly or when extremal values are computed.

These observations motivate the study of more disciplined subclasses of $\FF$-expressions that control this growth. 
In this section, we identify several such subclasses that remain computationally tractable while still preserving useful closure properties of WFFAs.

Throughout all of this section, let $\Sigma$ be an alphabet and $\FF = (\SS, \DD, \bb)$ a finance semiring such that $\SS = (S, \oplus, \otimes, \zero, \one)$ is a commutative semiring. Let $E \subseteq \EE_{\FF}$ be an arbitrary collection of $\FF$-expressions. We call an automaton
$
{\A = (Q, I, T, F, \wt_I, \wt_T, \wt_F) \in \WFFA_{\Sigma, \FF}}
$
{\em $E$-weighted} if $\wt_T(T) \subseteq E$, i.e., all transition weights are drawn from the restricted class $E$. 
Intuitively, this means that the quantitative behavior of the automaton is built using only a prescribed set of simple or ``admissible'' expressions, typically chosen to control the size and computational complexity of the resulting automata.

Let $\WFFA_{\Sigma, \FF}(E) \subseteq \WFFA_{\Sigma, \FF}$ denote the collection of all $E$-weighted WFFAs over $\Sigma$ and $\FF$. 
Let $\WFFA_{\Sigma, \FF}^{\tau}(E) = \WFFA^{\tau}_{\Sigma, \FF} \cap \WFFA_{\Sigma, \FF}(E)$ denote the collection of all $E$-weighted WFFAs over $\Sigma$ and $\FF$, which are purely transition-weighted.

\subsection{Closure Properties of Restricted WFFAs}

In this subsection, we generalize Theorem~\ref{THM:closure_wffa} to $E$-weighted WFFA.

We say that $E$ is {\em $\otimes$-closed} if, for every $e_1, e_2 \in E$, $\beh{e_1 \otimes e_2} \in \beh{E}$. We say that $E$ is {\em $(S, \otimes)$-closed} if, for every $s \in S$ and $e \in E$, $\beh{s \otimes e} \in \beh{E}$.

\begin{example}
\label{EX:closure_expressions}
\begin{enumerate}[label=(\alph*)]
\item $\EE_{\FF}$ and $S$ are $\otimes$- and $(S, \otimes)$-closed.
\item If $\FF = \FF_{\arc, \times}$, then $\bind{S}_{\FF}$ is $\otimes$-closed but not $(S, \otimes)$-closed. This is due to the fact that expressions in $\bind{\SS}_{\FF}$ can define only functions $d \mapsto s \cdot d$ for some constant $s \in S$. On the other side, expressions of the form $s \otimes \bind{s'}$ define functions of the form $d \mapsto s \cdot d + s'$.
\item
\label{EX:closure_expressions_primitive}
For $\FF = \FF_{\arc, \times}$, the collection of primitive expressions $\PP_{\FF}$ is neither $\otimes$- nor $(S, \otimes)$-closed. The reason for this is that primitive expressions define either constant functions $d \mapsto s$ or functions of the form $d \mapsto s \cdot d$ whereas expressions $[\pi_1 \otimes \pi_2]$ (with $\pi_1, \pi_2 \in \PP_{\FF}$) and $[s \otimes \pi]$ (with $s \in S$ and $\pi \in \PP_{\FF}$) can define, e.g., the function $d \mapsto d + 1$.
\item Recall that $\AA_{\FF}$ denotes the collection of affine $\FF$-expressions. If $\FF = \FF_{\arc, \times}$, then $\AA_{\FF}$ is $\otimes$- and $(S, \otimes)$-closed.
\end{enumerate}
\end{example}

\begin{theorem}
\label{THM:closure_properties_e_weighted_wffa}
Let $E \subseteq \EE_{\FF}$ and $\F_1, \F_2 \in \beh{\WFFA_{\Sigma, \FF}(E)}$. Then, the following hold.
\begin{enumerate}[label=(\alph*)]
\item $(\F_1 \oplus \F_2) \in \beh{\WFFA_{\Sigma, \FF}(E)}$.
\item
If $E$ is $\otimes$-closed, then $(\F_1 \otimes \F_2) \in \beh{\WFFA_{\Sigma, \FF}(E)}$.
\item
$(\F_1 \cdot \F_2) \in \beh{\WFFA_{\Sigma, \FF}(E)}$.
\item 
If $E$ is $(S, \otimes)$-closed and $\F \in \beh{\WFFA_{\Sigma, \FF}(E)}$ is proper, then its Kleene star $\F^*$ belongs to $\beh{\WFFA_{\Sigma, \FF}(E)}$.
\end{enumerate}
\end{theorem}

\begin{remark}
Note that the constructions for $\oplus$ and the Cauchy product $\cdot$ do not require any closure assumptions on $E$, in contrast to $\otimes$ and the star operation.
\end{remark}

In the rest of this subsection, we provide a proof of Theorem~\ref{THM:closure_properties_e_weighted_wffa}.

\begin{lemma}
\label{LEMMA:closure_cauchy_for_arbitrary_expressions}
Let $E \subseteq \EE_{\FF}$ be any collection of $\FF$-expressions and ${\F_1, \F_2 \in \beh{\WFFA_{\Sigma, \FF}(E)}}$. Then their Cauchy product $\F_1 \cdot \F_2$ belongs to $\beh{\WFFA_{\Sigma, \FF}(E)}$.
\end{lemma}

The proof of this theorem follows from Lemmas~\ref{LEMMA:cauchy_product_kleene_star_normalized_wffa} and \ref{LEMMA:init_and_final_normalization_for_arbitrary_expressions}.

\begin{lemma}
\label{LEMMA:init_and_final_normalization_for_arbitrary_expressions}
Let $E \subseteq \EE_{\FF}$ be any collection of $\FF$-expressions and $\A \in \WFFA_{\Sigma, \FF}(E)$.
\begin{enumerate}[label=(\alph*)]
\item There exists an initially normalized $\A_{I} \in \WFFA_{\Sigma, \FF}(E)$ such that $\beh{\A_{I}} = \beh{\A}$.
\item There exists a finally normalized $\A_{F} \in \WFFA_{\Sigma, \FF}(E)$ such that $\beh{\A_{F}} = \beh{\A}$.
\end{enumerate}
\end{lemma}

\begin{proof}
Let $\A = (Q, I, T, F, \wt_I, \wt_T, \wt_F)$.

\begin{enumerate}[label=(\alph*)]
\item 
In contrast to Lemma~\ref{LEMMA:normalize_wffa}, we cannot absorb initial weights into transitions, as $E$ is not assumed to be $\oplus$- or $(S,\otimes)$-closed.
Instead, we record the initial state along each run and defer the contribution of $\wt_I$ to the final step.
Formally, we introduce a fresh state $q_I$ and extend states $q \in Q$ to pairs $(i,q)$ with $i \in I$, simulating each run of $\A$ while remembering its origin. 
Transitions are lifted componentwise, and final weights are defined by
\[
\wt_F(i,f) = \wt_I(i)\otimes \wt_F(f).
\]
for $i \in I$ and $f \in F$.
If necessary, the value on $\eps$ is handled by making $q_I$ final.
Then each run of $\A_{I}$ corresponds to a run of $\A$ with the same accumulated weight, using commutativity of $\SS$, hence $\beh{\A_{I}}=\beh{\A}$.

\item 
Final normalization is dual. \qedhere
\end{enumerate}
\end{proof}

\begin{remark}
Given an arbitrary collection $E \subseteq \EE_{\FF}$ and $\A_1, \A_2 \in \WFFA_{\Sigma, \FF}(E)$, using the normalization constructions of Lemma~\ref{LEMMA:init_and_final_normalization_for_arbitrary_expressions} we can construct an $E$-weighted WFFA for $\beh{\A_1} \cdot \beh{\A_2}$ with $\O((|\A_1|_Q)^2 + (|\A_2|_Q)^2)$ states.
\end{remark}

For any $E \subseteq \EE_{\FF}$, let the {\em $\oplus$-closure} of $E$ be defined as
\[
\clp{E} = \Bigl\{[e_1 \oplus {\dots} \oplus e_k]  \mid  k \in \NN_{\ge 1} \text{ and } e_i \in E \text{ for all } i \in [k]\Bigr\} \subseteq \EE_{\FF}.
\]

Here, $[e_1 \oplus {\dots} \oplus e_k]$ denotes an expression obtained by choosing some valid way of parenthesizing the summands $e_1, \dots, e_k$. While different parenthesizations yield syntactically distinct expressions, their evaluation in the semiring is independent of the grouping, due to the associativity of $\oplus$.

The following lemma shows that closing $E$ under $\oplus$ does not increase the expressive power of WFFAs. 
The key idea is that sums of expressions can be simulated by nondeterministic branching at the automaton level, as formalized in the proof below.

\begin{lemma}
\label{LEMMA:plus_closed_expressiveness}
Let $E \subseteq \EE_{\FF}$ be any collection of $\FF$-expressions. Then
$
\beh{\WFFA_{\Sigma, \FF}(\clp{E})} = \beh{\WFFA_{\Sigma, \FF}(E)}.
$
\end{lemma}

\begin{proof}
Let $\A_{\clp{E}} = (Q, I, T, F, \wt_I, \wt_T, \wt_F) \in \WFFA_{\Sigma, \FF}(\clp{E})$. 
We construct ${\A_E \in \WFFA_{\Sigma, \FF}(E)}$ with the same behavior.
For each transition $t = (p,a,q) \in T$ with weight
\[
\wt_T(t) = [e_1 \oplus \dots \oplus e_k], \qquad e_i \in E,
\]
we replace $t$ by $k$ transitions from copies of $p$ to copies of $q$, all labeled by $a$, whose weights are the summands $e_1,\dots,e_k$.
Accordingly, the automaton $\A_E$ is obtained by expanding the state space, replacing each state $q$ by copies $(q,i)$, where $i$ is an index indicating which summand in the $\oplus$-decomposition of a transition weight has been selected.
Since WFFA semantics aggregates runs using $\oplus$, we obtain
$
\beh{\A_E} = \beh{\A_{\clp{E}}}.
$
The converse inclusion is immediate, since $E \subseteq \clp{E}$.
\end{proof}

\begin{lemma}
\label{LEMMA:closure_kleene_star_for_e}
\begin{enumerate}[label=(\alph*)]
\item Let $E \subseteq \EE_{\FF}$ be $(S, \otimes)$-closed and let $\F \in \beh{\WFFA_{\Sigma, \FF}(E)}$ be proper. Then $\F^* \in \beh{\WFFA_{\Sigma, \FF}(E)}$.
\item  Let $E \subseteq \EE_{\FF}$ be arbitrary and let $\F \in \beh{\WFFA^{\tau}_{\Sigma, \FF}(E)}$ be proper. Then ${\F^* \in \beh{\WFFA^{\tau}_{\Sigma, \FF}(E)}}$.
\end{enumerate}
\end{lemma}

\begin{proof}
\begin{enumerate}[label=(\alph*)]
\item 
Let $\A \in \WFFA_{\Sigma, \FF}(E)$ with $\beh{\A} = \F$. 
Using the Kleene-star construction from Lemma~\ref{LEMMA:cauchy_product_kleene_star_normalized_wffa}, 
we obtain $\A^* \in \WFFA_{\Sigma, \FF}(\EE_{\FF})$ with $\beh{\A^*} = \F^*$.

The transition weights of $\A^*$ are built from expressions in $E$ by combining them with constants from $S$ via $\otimes$ and $\oplus$. 
Since $E$ is $(S,\otimes)$-closed and $\otimes$ is commutative, every such expression is equivalent to an element of $\clp{E}$. 
Hence, by replacing transition weights with equivalent expressions, we obtain $\A^*_{\clp{E}} \in \WFFA_{\Sigma, \FF}(\clp{E})$ with the same behavior.
Finally, Lemma~\ref{LEMMA:plus_closed_expressiveness} yields an equivalent automaton $\A^*_E \in \WFFA_{\Sigma, \FF}(E)$, so $\F^* \in \beh{\WFFA_{\Sigma, \FF}(E)}$.

\item 
The argument is analogous. In the purely transition-weighted case, the construction only introduces expressions of the form
$e$, $\one \otimes e$, $e \otimes \one$, or $\one \otimes e \otimes \one$ with $e \in E$, all of which are equivalent to $e$. 
Thus all transition weights lie in $\clp{E}$, and applying Lemma~\ref{LEMMA:plus_closed_expressiveness} yields the result. \qedhere
\end{enumerate}
\end{proof}

\begin{proof}[Proof of Theorem~\ref{THM:closure_properties_e_weighted_wffa}]
\begin{enumerate}[label=(\alph*)]
\item By Lemma~\ref{LEMMA:closure_sum_wffa}, the disjoint union construction preserves $E$-weightedness.
\item By Lemma~\ref{LEMMA:closure_hadamard_wffa}, transition weights have the form $[e_1 \otimes e_2]$ with $e_1,e_2 \in E$. 
Since $E$ is $\otimes$-closed, these can be replaced by equivalent expressions in $E$.
\item Follows from Lemma~\ref{LEMMA:closure_cauchy_for_arbitrary_expressions}.
\item Follows from Lemma~\ref{LEMMA:closure_kleene_star_for_e}~(a). \qedhere
\end{enumerate}
\end{proof}

As a corollary from the proof of Theorem~\ref{THM:closure_properties_e_weighted_wffa}, we obtain the following closure properties for purely transition weighted WFFAs.

\begin{corollary}
\label{COR:closure_properties_purely_transition_weighted}
Let $E \subseteq \EE_{\FF}$ be an arbitrary collection of $\FF$-expressions. Then, for all QFLs $\F_1, \F_2 \in \beh{\WFFA^{\tau}_{\Sigma, \FF}(E)}$, the following hold.
\begin{enumerate}[label=(\alph*)]
\item If $\F_1(\eps), \F_2(\eps) \in \{\zero, \one\}$, then $(\F_1 \oplus \F_2) \in \beh{\WFFA^{\tau}_{\Sigma, \FF}(E)}$.
\item
$(\F_1 \cdot \F_2) \in \beh{\WFFA^{\tau}_{\Sigma, \FF}(E)}$.
\item If $\F \in \beh{\WFFA^{\tau}_{\Sigma, \FF}(E)}$ is proper, then $\F^* \in \beh{\WFFA^{\tau}_{\Sigma, \FF}(E)}$ 
\end{enumerate}
\end{corollary}

\begin{proof}
\begin{enumerate}[label=(\alph*)]
\item This follows directly from the construction given in the proof of Lemma~\ref{LEMMA:closure_sum_wffa}. 
\item This follows from the construction in the proof of Lemma~\ref{LEMMA:closure_cauchy_for_arbitrary_expressions}. Specifically, since $\F_1(\eps), \F_2(\eps) \in \{\zero, \one\}$, all initial and final weights in the constructed WFFA for $\F_1 \cdot \F_2$ are equal to $\one$.
\item This follows immediately from Lemma~\ref{LEMMA:closure_kleene_star_for_e}~(b). \qedhere
\end{enumerate}
\end{proof}

\subsection{Primitive Expressions}
\label{SUBSEC:primitive_expressions}

In Section \ref{SECTION:WFFA} we saw that WFFAs with primitive expressions are capable of describing a wide range of practically relevant scenarios. 
However, as already illustrated by simple constructions, composing WFFAs (e.g., via sequential or Hadamard composition) quickly leads to increasingly complex $\FF$-expressions. 
Such expressions tend to grow rapidly in size, making their analysis and interpretation cumbersome. 
This raises the question of whether one can remain within a restricted class of simple expressions while still preserving useful closure properties.

In this subsection, we investigate this question and study the closure properties of WFFAs with primitive expressions. 
In particular, we demonstrate that the WFFA constructions for the Cauchy product and the Kleene star can be optimized in this restricted setting.

In Example \ref{EX:closure_expressions}~\ref{EX:closure_expressions_primitive}, we showed that, in general, primitive expressions are neither $\otimes$-closed nor $(S, \otimes)$-closed.
We further show that $\PP_{\FF}$-weighted WFFAs do not, in general, preserve closure under the Hadamard product and the Kleene star. 
This justifies the conditions placed on $E$ in Theorem~\ref{THM:closure_properties_e_weighted_wffa}(b) and (d).

\begin{theorem}
\label{THM:closure_counterexample_primitive}
Let $\Sigma = \{\bot\}$ and $\FF = \FF_{\arc, \times}$. Then, the following hold:
\begin{enumerate}[label=(\alph*)]
\item There exist $\F_1, \F_2 \in \beh{\WFFA_{\Sigma, \FF}(\PP_{\FF})}$ such that $(\F_1 \otimes \F_2) \notin \beh{\WFFA_{\Sigma, \FF}(\PP_{\FF})}$.
\item There exists $\F \in \beh{\WFFA_{\Sigma, \FF}(\PP_{\FF})}$ such that $\F$ is proper and ${\F^* \notin \beh{\WFFA_{\Sigma, \FF}(\PP_{\FF})}}$.
\end{enumerate}
\end{theorem}

For the proof of Theorem~\ref{THM:closure_counterexample_primitive}, we require Lemmas~\ref{LEMMA:primitive_wffa_cube_property} and~\ref{LEMMA:primitive_counterexample}.

For $\delta = (d_1, \dots, d_n) \in \DD^n$ and $R \in \RR_{>0}$, we define the hypercube
\[
{\mathcal C(\delta, R) = \{(d_1', \dots, d_n') \in \DD^n  \mid  d_i \le d_i' \le d_i + R \text{ for all } i \in [n]\}}.
\]
For $\delta = d_1 \dots d_n \in \DD^*$, let $\bot_{\delta}$ denote the finance data word
$
(\bot, d_1) \dots (\bot, d_n).
$
Let $\Run_{\A}^n \subseteq \Run_{\A}$ be the set of all runs of $\A$ of length $n$, which is a finite set.

\begin{lemma}
\label{LEMMA:primitive_wffa_cube_property}
Let $\Sigma = \{\bot\}$, let $\FF = \FF_{\arc, \times}$, and let $\A \in \WFFA_{\Sigma, \FF}(\PP_{\FF})$ be such that $\beh{\A}(w) \neq \zero$ for all $w \in \DD \Sigma^*$. 
Then for every $n \in \NN_{\ge 1}$ there exist $\varrho \in \Run_{\A}^n$, $\delta \in \DD^n$, and $R > 0$ such that
$
\beh{\A}(\bot_{\delta'}) = \wt_{\A}(\varrho, \delta')
$
for all $\delta' \in \mathcal C(\delta, R)$.
\end{lemma}

\begin{proof}
Let $\A = (Q, I, T, F, \wt_I, \wt_T, \wt_F)$ and fix $n \ge 1$. 
Enumerate $\Run_{\A}^n = (\varrho_j)_{j \in [k]}$.
For each run $\varrho_j$, the weight $\wt_{\A}(\varrho_j, \delta)$ is an affine function of $\delta \in \DD^n$. 
Hence,
$
{\beh{\A}(\bot_{\delta}) = \max_{j \in [k]} \wt_{\A}(\varrho_j, \delta)}
$
is the pointwise maximum of finitely many affine functions.
For each $j \in [k]$, let
\[
\Delta_j = \{\delta \in \DD^n \mid \wt_{\A}(\varrho_j, \delta) \ge \wt_{\A}(\varrho_{j'}, \delta) \text{ for all } j' \in [k]\}.
\]
Then each $\Delta_j$ is a polyhedron defined by linear inequalities, and by assumption $\beh{\A}(\bot_{\delta}) \neq \zero$ for all $\delta$, we have
$
\DD^n = \bigcup_{j \in [k]} \Delta_j.
$
Since this is a finite covering of $\DD^n$ by polyhedra, at least one $\Delta_j$ contains a cube $\mathcal C(\delta, R)$ for some $\delta \in \DD^n$ and $R > 0$. 
For this $j$, we have
$
\beh{\A}(\bot_{\delta'}) = \wt_{\A}(\varrho_j, \delta')
$
for all $\delta' \in \mathcal C(\delta, R)$.
\end{proof}

Recall that $\AA_{\FF}$ denotes the collection of affine $\FF$-expressions (cf. Section~\ref{SEC:preliminaries}).

\begin{lemma}
\label{LEMMA:primitive_counterexample}
Let $\Sigma = \{\bot\}$ and $\FF = \FF_{\arc, \times}$. Let $\F: \DD \Sigma^* \to S$ be the QFL defined by  
$
\F(\bot_{d_1, \dots, d_n}) = n + \sum_{i=1}^n d_i
$
for all $n \in \NN$ and $d_1, \dots, d_n \in \DD$.
Then, $\F \in \beh{\WFFA_{\Sigma, \FF}(\AA_{\FF})} \setminus \beh{\WFFA_{\Sigma, \FF}(\PP_{\FF})}$.
\end{lemma}

\begin{proof}
The QFL $\F$ is recognized by the one-state WFFA with a single loop transition 
${t = (q, \bot, q)}$ with weight
$
\wt_T(t) = [\bind{1.0} \otimes 1.0],
$
where both the initial and final weights of $q$ are $0.0$. 
Thus each input symbol contributes $d_i + 1$, and therefore
$
{\beh{\A_{\F}}(\bot_{d_1,\dots,d_n}) = n + \sum_{i = 1}^n d_i}.
$

Suppose, for contradiction, that $\F \in \beh{\WFFA_{\Sigma,\FF}(\PP_{\FF})}$, and let $\A$ be such an automaton. 
By Lemma~\ref{LEMMA:primitive_wffa_cube_property}, for every $n$ there exist a run $\varrho \in \Run_{\A}^n$, a vector $\delta \in \DD^n$, and $R>0$ such that
$
\beh{\A}(\bot_{\delta'}) = \wt_{\A}(\varrho,\delta')
$
for all $\delta' \in \mathcal C(\delta,R)$.
Since $\A$ has primitive weights, $\wt_{\A}(\varrho,\delta')$ is an affine function
$
c + \sum_{i=1}^n \alpha_i d_i'
$
with $\alpha_i \in \{0,s_i\}$ depending on the transition type. 
On the other hand,
$
\beh{\A}(\bot_{\delta'}) = n + \sum_{i=1}^n d_i'
$
on $\mathcal C(\delta,R)$. 
Hence these affine functions coincide, implying $\alpha_i=1$ for all $i$ and $c=n$.
Thus all transitions along $\varrho$ must be $\bind{1.0}$, and the constant term $c$ equals the sum of initial and final weights, which is bounded independently of $n$. 
This contradicts $c=n$ for arbitrarily large $n$. Therefore, $\F \notin \beh{\WFFA_{\Sigma,\FF}(\PP_{\FF})}$.
\end{proof}

\begin{proof}[Proof of Theorem~\ref{THM:closure_counterexample_primitive}]
\begin{enumerate}[label=(\alph*)]
\item 
Let $\F_1, \F_2 : \DD \Sigma^* \to S$ be defined as follows. 
For ${w = (\bot,d_1)\dots(\bot,d_n) \in \DD \Sigma^*}$, let
$
\F_1(w) = \sum_{i=1}^n d_i
$ 
and
$
\F_2(w) = n,
$
and let ${\F_1(\varepsilon) = \F_2(\varepsilon) = 0.0}$.
Both QFLs are recognized by one-state WFFAs: $\F_1$ via a loop with weight $\llangle 1.0 \rrangle$, and $\F_2$ via a loop with constant weight $1.0$.
Then, for $\F = \F_1 \otimes \F_2$, we have
$
\F(\bot_{d_1,\dots,d_n}) = n + \sum_{i=1}^n d_i.
$
By Lemma~\ref{LEMMA:primitive_counterexample}, $\F \notin \beh{\WFFA_{\Sigma,\FF}(\PP_{\FF})}$.

\item 
Let $\F : \DD \Sigma^* \to S$ be defined by
$
\F((\bot,d)) = d+1
$
and
$
\F(w)=\zero
$
for all $w \in \DD \Sigma^*$ with $|w|\neq 1$. 
Then $\F$ is proper and is recognized by a WFFA with a single transition $(q_0,\bot,q_1)$ of primitive weight $\bind{1.0}$, initial weight $1.0$ at $q_0$, and final weight $0.0$ at $q_1$. 
For its Kleene star, we obtain
$
\F^*(\bot_{d_1,\dots,d_n}) = n + \sum_{i = 1}^n d_i.
$
Again, by Lemma~\ref{LEMMA:primitive_counterexample}, $\F^* \notin \beh{\WFFA_{\Sigma,\FF}(\PP_{\FF})}$.
\qedhere
\end{enumerate}
\end{proof}

Despite the negative results of Theorem~\ref{THM:closure_counterexample_primitive}, 
Corollary~\ref{COR:closure_properties_purely_transition_weighted} shows that 
purely transition-weighted WFFAs with primitive expressions remain closed under the Kleene star. 
Moreover, as we illustrate below, this restricted class admits more efficient constructions for both the Cauchy product and the Kleene star.

To formalize these improvements, we introduce two structural properties.
We say that $E \subseteq \EE_{\FF}$ satisfies the \emph{bounded-sum property} 
if there exists a constant $B \in \NN_{\ge 1}$ such that for all $k \in \NN_{\ge 1}$ 
and all $e_1, \dots, e_k \in E$, there exist $b \le B$ and $e_1', \dots, e_b' \in E$ satisfying
$
\beh{e_1 \oplus \dots \oplus e_k}
=
\beh{e_1' \oplus \dots \oplus e_b'}.
$
Let $B(\PP_{\FF})$ denote the minimal such constant.
We say that a data-binding function $\bb: \DD \times S \to S$ is \emph{additive} if
$
{\bb(d, s \oplus s') = \bb(d, s) \oplus \bb(d, s')}
$
for all $d \in \DD$ and $s, s' \in S$.

\begin{lemma}
\label{LEMMA:bounded_sum_property_ff_arc_times_positive}
Let $\FF$ be a finance semiring with an additive data-binding function. Then, $\PP_{\FF}$ satisfies the bounded sum property and $B(\PP_{\FF}) \le 2$.
\end{lemma}

\begin{proof}
Let $e = [\pi_1 \oplus \dots \oplus \pi_k]$ with $\pi_i \in \PP_{\FF}$. 
Since $\oplus$ is commutative, we may group all constant and all binding terms and write
$
e = [s' \oplus \bind{s''}]
$
for some $s', s'' \in S$, where one of the terms may be absent.
By additivity of $\bb$, we have
$
\beh{e} = \beh{s' \oplus \bind{s''}}.
$
Thus every sum of elements of $\PP_{\FF}$ is equivalent to a sum of at most two elements of $\PP_{\FF}$, and hence $B(\PP_{\FF}) \le 2$.
\end{proof}

As an illustration, we can demonstrate that the bounded-sum property holds not only for finance semirings with an additive data-binding function. For example, consider $\FF = \FF^{\RR}_{\arc, \times}$ as defined in Example \ref{EX:ff_arc_times}(b). Clearly, the data-binding function of $\FF$ is not additive. For instance,
$
\bb_{\times}^{\RR}(-1, 1 \oplus 2) \neq \bb_{\times}^{\RR}(-1, 1) \oplus \bb_{\times}^{\RR}(-1, 2)
$.

\begin{lemma}
\label{LEMMA:bounded_sum_property_ff_arc_times_rr}
Let $\FF = \FF_{\arc, \times}^{\RR}$. Then $\PP_{\FF}$ satisfies the bounded-sum property with $B(\PP_{\FF}) \le 5$.
\end{lemma}

\begin{proof}[Proof sketch]
Let $e = [\pi_1 \oplus \dots \oplus \pi_k]$ with $\pi_i \in \PP_{\FF}$. 
By separating constant and binding terms, and further splitting binding terms according to the sign of their coefficients, 
one can reduce $e$ to an equivalent sum of at most five primitive expressions.
The key observation is that, although the data-binding function is not additive, 
it is still possible to control the contribution of positive and negative coefficients separately. 
A detailed case analysis yields the bound $5$.
\end{proof}

\begin{lemma}
\label{LEMMA:wffa_constructions_with_bounded_sum}
Assume that $\PP_{\FF}$ satisfies the bounded-sum property.
\begin{enumerate}[label=(\alph*)]
\item Let $\A_1, \A_2 \in \WFFA^{\tau}_{\Sigma, \FF}(\PP_{\FF})$. Then, there exists $\A' \in \WFFA_{\Sigma, \FF}^{\tau}(\PP_{\FF})$ such that ${\beh{\A'} = \beh{\A_1} \cdot \beh{\A_2}}$, and $|\A'|_Q = \O(|\A_1|_Q + |\A_2|_Q)$.
\item Let $\A \in \WFFA^{\tau}_{\Sigma, \FF}(\PP_{\FF})$ such that $\beh{\A}$ is proper. Then, there exists ${\A^* \in \WFFA_{\Sigma, \FF}^{\tau}(\PP_{\FF})}$ such that $\beh{\A^*} = \beh{\A}^*$ and $|\A^*|_Q = \O(|\A|_Q)$.
\end{enumerate}
\end{lemma}

\begin{proof}
\begin{enumerate}[label=(\alph*)]
\item 
Using the standard construction for the Cauchy product 
(Corollary~\ref{COR:cauchy_product_kleene_star_full_wffa}), 
we obtain an automaton $\hat{\A}$ with 
$|\hat{\A}|_Q = |\A_1|_Q + |\A_2|_Q + \O(1)$ and 
$\beh{\hat{\A}} = \beh{\A_1} \cdot \beh{\A_2}$.

Since $\A_1$ and $\A_2$ are purely transition-weighted, 
all transition weights of $\hat{\A}$ lie in $\clp{\PP_{\FF}}$. 
By the bounded-sum property, each such weight can be replaced by an equivalent expression of bounded size. 
Applying Lemma~\ref{LEMMA:plus_closed_expressiveness} then yields 
$\A' \in \WFFA^{\tau}_{\Sigma,\FF}(\PP_{\FF})$ with the same behavior.
Since $B(\PP_{\FF})$ is constant, the number of states remains 
$\O(|\A_1|_Q + |\A_2|_Q)$.

\item 
The proof is analogous, using the Kleene-star construction from 
Corollary~\ref{COR:cauchy_product_kleene_star_full_wffa}.
\qedhere
\end{enumerate}
\end{proof}

\begin{remark}
The proof of Lemma~\ref{LEMMA:wffa_constructions_with_bounded_sum} is non-constructive, 
as it relies on the existence of equivalent expressions with a bounded number of summands. However, for the concrete semirings $\FF = \FF_{\arc, \times}$ and $\FF = \FF_{\arc, \times}^{\RR}$, 
Lemmas~\ref{LEMMA:bounded_sum_property_ff_arc_times_positive} and~\ref{LEMMA:bounded_sum_property_ff_arc_times_rr} 
provide explicit constructions, and thus yield effective procedures.
\end{remark}

\subsection{Affine Expressions}
\label{SUBSEC:affine_expressions}

Recall that an $\FF$-expression $e \in \EE_{\FF}$ is \emph{affine} if either $e \in \PP_{\FF}$ or $e = [\bind{s} \otimes s']$ for some $s,s' \in S$. 
Let $\AA_{\FF}$ denote the set of all affine $\FF$-expressions. Clearly,
$
\PP_{\FF} \subseteq \AA_{\FF} \subseteq \EE_{\FF}.
$

We say that $\FF$ is \emph{multiplicative} if
$
\bb(d, s \otimes s') = \bb(d, s) \otimes \bb(d, s')
$
for all $d \in \DD$ and $s,s' \in S$. For instance, $\FF_{\arc,\times}$ is multiplicative. It is immediate that $\AA_{\FF}$ is $(S,\otimes)$-closed, and if $\FF$ is multiplicative, then $\AA_{\FF}$ is also $\otimes$-closed.

As an immediate consequence of Theorem~\ref{THM:closure_properties_e_weighted_wffa} , we obtain the following closure properties of affine $\FF$-expressions $\AA_{\FF}$.

\begin{corollary}
Let $\F_1, \F_2, \F \in \beh{\WFFA_{\Sigma, \FF}(\AA_{\FF})}$ such that $\F$ is proper. Then:
\begin{enumerate}[label=(\alph*)]
\item The QFLs $\F_1 \oplus \F_2$, $\F_1 \cdot \F_2$ and $\F^*$ are also in $\beh{\WFFA_{\Sigma, \FF}(\AA_{\FF})}$.
\item If $\FF$ is multiplicative, then the QFL $\F_1 \otimes \F_2$ is also in $\beh{\WFFA_{\Sigma, \FF}(\AA_{\FF})}$.
\end{enumerate}
\end{corollary}

\begin{lemma}
\label{LEMMA:wffa_with_affine_expressions_qfls_inclusion}
Let $\FF = \FF_{\arc, \times}$ and $\Sigma = \{\bot\}$. Then,
\[
\beh{\WFFA_{\Sigma, \FF}(\PP_{\FF})} \subsetneq \beh{\WFFA_{\Sigma, \FF}(\AA_{\FF})} \subsetneq \beh{\WFFA_{\Sigma, \FF}(\EE_{\FF})}.
\]
\end{lemma}

\begin{proof}
The inclusions 
$
\beh{\WFFA_{\Sigma, \FF}(\PP_{\FF})} \subseteq \beh{\WFFA_{\Sigma, \FF}(\AA_{\FF})} \subseteq \beh{\WFFA_{\Sigma, \FF}(\EE_{\FF})}
$
are immediate.
To show strictness, we argue as follows.
First, define $\F: \DD \Sigma^* \to S$ on words of the form $\bot_{d_1,\dots,d_n}$, where $n \ge 0$ and $d_1,\dots,d_n \in D$, by
$
\F(\bot_{d_1,\dots,d_n}) = n + \sum_{i=1}^n d_i.
$ 
By Lemma~\ref{LEMMA:primitive_counterexample}, $\F \notin \beh{\WFFA_{\Sigma,\FF}(\PP_{\FF})}$, 
while $\F$ is recognized by a WFFA with affine expressions. Hence the first inclusion is strict.
Second, define $\F': \DD \Sigma^* \to S$ such that $\F'(\bot_{0.5}) = 1$ and $\F'(w) = \zero$ for all other words $w$.
Then $\F' \in \beh{\WFFA_{\Sigma,\FF}(\EE_{\FF})}$.
Suppose, for contradiction, that $\F' \in \beh{\WFFA_{\Sigma,\FF}(\AA_{\FF})}$. 
Then there exists a WFFA $\A$ over $\AA_{\FF}$ such that
$
{\beh{\A}(\bot_d) = \max_{i \in [k]} (a_i d + b_i)}
$
for some $k \in \NN$ and $a_i, b_i \in \RR$.
However, a function of this form cannot be equal to $\zero$ at all points except at a single input. 
This contradicts the definition of $\F'$. Hence ${\F' \notin \beh{\WFFA_{\Sigma,\FF}(\AA_{\FF})}}$, 
and the second inclusion is strict.
\end{proof}

\begin{lemma}
\label{LEMMA:bounded_sum_for_affine_expressions}
Let $\FF = \FF_{\arc, \times}$. Then $\AA_{\FF}$ does not satisfy the bounded-sum property.
\end{lemma}

\begin{proof}
For $i \in \NN_{\ge 1}$, let $\alpha_i = [\bind{2i} \otimes (-i^2)]$. Then
$
\beh{\alpha_i}(x) = 2ix - i^2.
$
For every $N \in \NN_{\ge 1}$, the function
$
\beh{\alpha_1 \oplus \dots \oplus \alpha_N}
$
is the maximum of $N$ distinct affine functions, each dominating on a non-empty interval. 
Hence it cannot be represented as a maximum of fewer than $N$ affine expressions. 
Therefore no uniform bound exists.
\end{proof}

The failure of the bounded-sum property has concrete consequences for the size of automata constructions.

\begin{lemma}
\label{LEMMA:affine_expressions_blowup}
Let $\FF = \FF_{\arc, \times}$ and $\Sigma = \{\bot\}$. 
Then there exists a family $(\F_N)_{N \ge 1}$ such that:
\begin{enumerate}[label=(\alph*)]
\item $\F_N \in \beh{\WFFA^{\tau}_{\Sigma,\FF}(\AA_{\FF})}$ is recognized by an automaton with $\O(N)$ states;
\item every $\A \in \WFFA^{\tau}_{\Sigma,\FF}(\AA_{\FF})$ with $\beh{\A} = \F_N \cdot \F_N$ satisfies
$
|\A|_Q \ge 3 \cdot N^{\frac{4}{3}}.
$
\end{enumerate}
\end{lemma}

\begin{proof}[Proof sketch]
For $i \in \NN_{\ge 1}$, define the function $\varphi_i: x \mapsto 2ix - i^2$, and for $N \in \NN_{\ge 1}$, define $f_N: x \mapsto \max_{i \in [N^2]} \varphi_i(x)$. 
Let $\F_N : \DD \Sigma^* \to S$ be given by $\F_N(\bot_d) = f_N(d)$ and $\F_N(w)=\zero$ for $|w| \neq 1$.
The function $f_N$ is realized by a purely transition-weighted WFFA over $\AA_{\FF}$ with $2N$ states: one distributes the $N^2$ affine expressions $\varphi_i$ over $N$ initial and $N$ final states.

For the Cauchy product, we obtain
$
(\F_N \cdot \F_N)(\bot_{x,y})
=
\max_{i,j \in [N^2]} \bigl(\varphi_i(x) + \varphi_j(y)\bigr).
$
Each pair $(i,j)$ defines a distinct affine function in two variables, and for every $(i,j)$ there is a non-empty region on which this function uniquely attains the maximum. Hence any WFFA recognizing $\F_N \cdot \F_N$ must realize at least $N^4$ distinct runs of length $2$.

Now let $\A=(Q,I,T,F,\wt_I, \wt_T, \wt_F)$ be a purely transition-weighted WFFA over $\AA_{\FF}$ with $\beh{\A} = \F_N \cdot \F_N$, and let
$
H = Q \setminus (I \cup F).
$
Since $\A$ accepts only words of length $2$, $I \cap F = \emptyset$ and every accepting run has the form
$
q_I \to q_H \to q_F,
$
where $q_I \in I$, $q_H \in H$, and $q_F \in F$. Therefore the number of runs is at most
$
|I| \cdot |H| \cdot |F|.
$
Thus
$
N^4 \le |I| \cdot |H| \cdot |F|.
$
By the arithmetic--geometric mean inequality,
$
|\A|_Q = |I| + |H| + |F| \ge 3 (|I|\,|H|\,|F|)^{1/3} \ge 3 N^{4/3}.
$
This proves the claim.
\end{proof}

\subsection{Monomials}
\label{SUBSEC:monomials}

In this subsection, we show that for $\FF = \FF_{\arc, \times}$, WFFAs with arbitrary $\FF$-expressions can be simulated by WFFAs using expressions of constant size.

An $\FF$-constraint $c \in \EE_{\FF}$ is called {\em basic} if $c = [(\pi_1 \oplus \pi_2) \bowtie \pi_3]$ where $\pi_1, \pi_2, \pi_3 \in \PP_{\FF}$ and ${\bowtie} \in \{=, \neq\}$. An $\FF$-constraint $c \in \EE_{\FF}$ is called {\em simple} if it is either basic or of the form $c = [\beta_1 \otimes \beta_2]$ where $\beta_1, \beta_2 \in \CC_{\FF}$ are basic $\FF$-constraints. An expression $e \in \EE_{\FF}$ is called an {\em $\FF$-monomial} if it is of the form $e = [c \otimes \alpha]$ where $c \in \CC_{\FF}$ is a simple $\FF$-constraint and $\pi \in \AA_{\FF}$ is a affine $\FF$-expression. Let $\MM_{\FF} \subseteq \CC_{\FF}$ denote the collection of all $\FF$-monomials.

\begin{remark}
\label{REM:monomials_expressiveness}
Clearly, $\beh{\AA_{\FF}} \subsetneq \beh{\MM_{\FF}}$. 
Moreover, all expressions in $\MM_{\FF}$ have bounded size, and thus can be evaluated in constant time (cf.~Remark~\ref{REM:exp_size}).
\end{remark}

\begin{theorem}
\label{THM:ff_arc_times_reduction_to_monomials}
Let $\FF = \FF_{\arc, \times}$. Then
$
\beh{\WFFA_{\Sigma, \FF}} = \beh{\WFFA_{\Sigma, \FF}(\MM_{\FF})}.
$
\end{theorem}

\begin{proof}
The statement follows from Lemmas~\ref{LEMMA:from_arbitrary_expressions_to_monomials} and~\ref{LEMMA:plus_closed_expressiveness}.
\end{proof}
We briefly recall that a function $f : \DD \to S$ is \emph{piecewise affine} if there exists a finite partition of $\DD$ into intervals such that, on each interval, either $f$ is affine or $f$ is constantly $\zero$.

\begin{lemma}
\label{LEMMA:finarc_exp_piecewise_affine}
For every $e \in \EE_{\FF}$, the function $\beh{e} : \DD \to S$ is piecewise affine.
\end{lemma}

\begin{proof}[Proof sketch]
The statement follows by structural induction on $e$. 
Constants and bindings are affine, $\otimes$ preserves affinity, and the operations $\oplus$, equality, and inequality introduce only finitely many breakpoints, yielding a finite partition of $\DD$.
\end{proof}

\begin{lemma}
\label{LEMMA:polynomial_for_piecewise}
Every piecewise-affine function $f : \DD \to S$ can be represented by an expression in $\clp{\MM_{\FF}}$.
\end{lemma}

\begin{proof}[Proof sketch]
Let $f : \DD \to S$ be piecewise affine. Then there exists a finite partition of $\DD$ into intervals such that $f$ is affine on each interval.
Each affine component can be encoded by a monomial together with constraints restricting its domain to the corresponding interval; see Example~\ref{EX:ff_arc_times_expressions}.
Taking the $\oplus$-sum over all components yields an expression in $\clp{\MM_{\FF}}$ representing $f$.
\end{proof}

\begin{lemma}
\label{LEMMA:from_arbitrary_expressions_to_monomials}
Let $\FF = \FF_{\arc, \times}$. Then
$
\beh{\EE_{\FF}} = \beh{\clp{\MM_{\FF}}}.
$
\end{lemma}

\begin{proof}[Proof sketch]
By Lemma~\ref{LEMMA:finarc_exp_piecewise_affine}, every expression $e \in \EE_{\FF}$ defines a piecewise-affine function $f : \DD \to S$.
By Lemma~\ref{LEMMA:polynomial_for_piecewise}, such a function can be represented by an expression in $\clp{\MM_{\FF}}$.
\end{proof}

As a consequence, we obtain the following strict hierarchy:

\begin{corollary}
Let $\FF = \FF_{\arc, \times}$. Then:
\[
\beh{\WFFA_{\Sigma, \FF}(\PP_{\FF})} 
\subsetneq 
\beh{\WFFA_{\Sigma, \FF}(\AA_{\FF})} 
\subsetneq 
\beh{\WFFA_{\Sigma, \FF}(\MM_{\FF})} 
= 
\beh{\WFFA_{\Sigma, \FF}}.
\]
\end{corollary}

\begin{remark}
\label{REM:negation}
The reduction result of Theorem~\ref{THM:ff_arc_times_reduction_to_monomials} can be used to derive expressiveness limitations of WFFAs. 
In particular, one can show that the class of WFFA-recognizable QFLs is not closed under pointwise negation over the arctic semiring $\SS_{\arc}$. 
However, after an appropriate transformation, such negated functions can be represented in the tropical semiring. 
We omit the technical details.
\end{remark}

\section{Proof of the Kleene--Sch\"utzenberger Theorem}
\label{SEC:proof_kleene_schuetzenberger}

In this section, we prove a parametric version of the Kleene--Sch\"utzenberger theorem for WFFAs. In Section~\ref{SEC:weighted_regexp}, we stated the theorem for the full class of WFFAs with arbitrary $\FF$-expressions as transition weights. Here, we generalize this result to WFFAs whose transitions are restricted to an arbitrary collection $E$ of $\FF$-expressions. This formulation covers the full class as a special case and provides a flexible framework to study WFFAs with different types of transition expressions.

Let $\Sigma$ be an alphabet and $\FF = (\SS, \DD, \bb)$ a finance semiring with $\SS = (S, \oplus, \otimes, \zero, \one)$.

Given a collection of $\FF$-expressions $E \subseteq \EE_{\FF}$, we say that $\R \in \Reg_{\Sigma, \FF}$ is {\em $E$-weighted} if, for every subexpression of $\R$ of the form $e_a$ with $e \in \EE_{\FF}$ and $a \in \Sigma$, we have $e \in E$. Let
$
\Reg_{\Sigma, \FF}(E) = \{\R \in \Reg_{\Sigma, \FF}  \mid  \R \text{ is } E\text{-weighted}\}.
$
For any collection of weighted finance regular expressions $\W \subseteq \Reg_{\Sigma, \FF}$, let 
$
\beh{\W} = \{\beh{\R}  \mid  \R \in \W\}.
$

Unfortunately, under a restriction to an arbitrary collection $E$ of $\FF$-expressions, the expressive correspondence between WFFA and weighted finance regular expressions may fail. For example, the following lemma demonstrates that, already for primitive expressions, weighted finance regular expressions are strictly more expressive than WFFAs. 

\begin{lemma}
Let $\Sigma = \{\bot\}$ and $\FF = \FF_{\arc, \times}$. Then
$
\beh{\Reg_{\Sigma, \FF}(\PP_{\FF})} \supsetneq \beh{\WFFA_{\Sigma, \FF}(\PP_{\FF})}.
$
\end{lemma}

\begin{proof}
Let $\F: \DD \Sigma^* \to S$ be defined by $\F(\bot_{d_1, \dots, d_n}) = n + \sum_{i = 1}^n d_i$ for all $d_1, \dots, d_n \in \DD$. Let $\R_{\F} \in \Reg_{\Sigma, \FF}(E)$ be the regular expression $([1.0]_{\eps} \cdot \bind{1.0}_{\bot})^*$. Then, $\beh{\R_{\F}} = \F$, so $\F \in \beh{\Reg_{\Sigma, \FF}(\PP_{\FF})}$.

On the other hand, by Lemma~\ref{LEMMA:primitive_counterexample}, $\R \notin \beh{\WFFA_{\Sigma, \FF}(E)}$.
\end{proof}

\begin{definition}
A regular expression $\R \in \Reg_{\Sigma, \FF}$ is called {\em $\eps$-free} if it can be derived from the grammar
\[
\R \; ::= \; e_a  \mid  \R \oplus \R  \mid  \R \cdot \R  \mid  \R \cdot \R^*  \mid  \R^* \cdot \R
\]
where $e \in \EE_{\FF}$ and $a \in \Sigma$.
\end{definition}

For $E \subseteq \EE_{\FF}$, let $\Reg_{\Sigma, \FF}^{\epsfree}(E) = \{\R \in \Reg_{\Sigma, \FF}(E)  \mid  \R \text{ is } \epsfree\}$.

\begin{remark}
\label{REM:epsfree_proper}
Note that, for any $\R \in \Reg_{\Sigma, \FF}^{\epsfree}(E)$, $\beh{\R}$ is proper, i.e., $\beh{\R}(\eps) = \zero$.
\end{remark}

\begin{definition}
A regular expression $\R \in \Reg_{\Sigma, \FF}$ is called {\em restricted} if it can be derived from the grammar
\[
\R \; ::= \; s_{\eps}  \mid  e_a  \mid  \R \oplus \R  \mid  \R \cdot \R  \mid  (\R_{\epsfree})^*
\]
where $s \in S$, $e \in \EE_{\FF}$ and $\R_{\epsfree} \in \Reg_{\Sigma, \FF}^{\epsfree}(\EE_{\FF})$.
\end{definition}

Clearly,
$
\Reg_{\Sigma, \FF}^{\epsfree}(E) \subseteq \Reg_{\Sigma, \FF}^{\res}(E) \subseteq \Reg_{\Sigma, \FF}(E).
$

We now present the main result of this section: a parametric version of the Kleene--Sch\"utzenberger theorem for WFFAs, which relates weighted finance regular expressions and WFFAs relative to an arbitrary collection of $\FF$-expressions. This formulation captures both the case where the collection is $(S,\otimes)$-closed and the general case for an arbitrary collection $E \subseteq \EE_{\FF}$.

Let $\WFFA_{\Sigma, \FF}^{\tau, \mathrm{proper}}(E)$ denote the class of all purely transition-weighted automata $\A \in \WFFA_{\Sigma, \FF}^{\tau}(E)$ such that $\beh{\A}$ is a proper QFL.

\begin{theorem}[Parametric Kleene--Sch\"utzenberger theorem]
\label{THM:parametric_kleene_schuetzenberger}
Let $E \subseteq \EE_{\FF}$ be a collection of $\FF$-expressions.
\begin{enumerate}[label=(\alph*)]
\item If $E$ is $(S, \otimes)$-closed, then $\beh{\Reg_{\Sigma, \FF}(E)} = \beh{\WFFA_{\Sigma, \FF}(E)}$.
\item For an arbitrary $E$, $\beh{\Reg_{\Sigma, \FF}^{\res}(E)} = \beh{\WFFA_{\Sigma, \FF}(E)}$.
\item For an arbitrary $E$, then $\beh{\Reg_{\Sigma, \FF}^{\epsfree}(E)} = \beh{\WFFA_{\Sigma, \FF}^{\tau, \mathrm{proper}}(E)}$.
\end{enumerate}
\end{theorem}

\begin{remark}
As mentioned in Example~\ref{EX:closure_expressions}(a), $E = \EE_{\FF}$ is $(S, \otimes)$-closed. Therefore, Theorem~\ref{THM:kleene_schuetzenberger} is an immediate consequence of Theorem~\ref{THM:parametric_kleene_schuetzenberger}.
\end{remark}

The proof of Theorem~\ref{THM:parametric_kleene_schuetzenberger} is given below.

The next lemma establishes recognizability of the elementary weighted finance expressions.

\begin{lemma}
\label{LEMMA:primitive_regexp_recognizable}
Let $E \subseteq \EE_{\FF}$.
\begin{enumerate}[label=(\alph*)]
\item For every $e \in E$ and $a \in \Sigma$, $\beh{e_a} \in \beh{\WFFA_{\Sigma, \FF}^{\tau}(E)}$.
\item For every $s \in S$, $\beh{s_{\eps}} \in \beh{\WFFA_{\Sigma, \FF}(E)}$.
\end{enumerate}
\end{lemma}

\begin{proof}
Both claims are immediate from the one-transition and one-state constructions: 
$e_a$ is realized by a purely transition-weighted automaton with a single transition labeled by $a$ and weighted by $e$, and $s_{\eps}$ is realized by a one-state automaton whose unique state is both initial and final, with weight $s$ assigned to either the initial or the final weight function.
\end{proof}

\begin{lemma}
\label{LEMMA:epsfree_regexp_recognizable}
Let $E \subseteq \EE_{\FF}$ and $\R \in \Reg_{\Sigma, \FF}^{\epsfree}(E)$. Then there exists $\A \in \WFFA_{\Sigma, \FF}^{\tau}(E)$ such that $\beh{\A} = \beh{\R}$.
\end{lemma}

\begin{proof}
By induction on the structure of $\R$. The base case $\R=e_a$ follows from Lemma~\ref{LEMMA:primitive_regexp_recognizable}(a). The cases $\oplus$ and concatenation follow from Corollary~\ref{COR:closure_properties_purely_transition_weighted}(a),(b). For expressions of the form $\R_1 \cdot \R_2^*$ and $\R_1^* \cdot \R_2$, we use Remark~\ref{REM:epsfree_proper} to ensure that the subexpression under the Kleene star is proper, and then apply Corollary~\ref{COR:closure_properties_purely_transition_weighted}(c).
\end{proof}

\begin{lemma}
\label{LEMMA:restricted_regexp_recognizable}
Let $E \subseteq \EE_{\FF}$ and $\R \in \Reg_{\Sigma, \FF}^{\res}(E)$. Then there exists $\A \in \WFFA_{\Sigma, \FF}(E)$ such that $\beh{\A} = \beh{\R}$.
\end{lemma}

\begin{proof}
Again by induction on the structure of $\R$. The cases $e_a$ and $s_{\eps}$ follow from Lemma~\ref{LEMMA:primitive_regexp_recognizable}. The cases $\oplus$ and concatenation follow from Theorem~\ref{THM:closure_properties_e_weighted_wffa}(a),(c). If $\R = (\R_{\epsfree})^*$, then Lemma~\ref{LEMMA:epsfree_regexp_recognizable} yields a purely transition-weighted WFFA for $\R_{\epsfree}$, and Remark~\ref{REM:epsfree_proper} ensures that its behavior is proper; hence Corollary~\ref{COR:closure_properties_purely_transition_weighted}(c) applies.
\end{proof}

\begin{lemma}
\label{LEMMA:full_regexp_recognizable_if_s_otimes_closed}
Let $E \subseteq \EE_{\FF}$ be $(S,\otimes)$-closed and let $\R \in \Reg_{\Sigma, \FF}(E)$. Then there exists $\A \in \WFFA_{\Sigma, \FF}(E)$ such that $\beh{\A} = \beh{\R}$.
\end{lemma}

\begin{proof}
Immediate from Lemma~\ref{LEMMA:primitive_regexp_recognizable} and Theorem~\ref{THM:closure_properties_e_weighted_wffa}.
\end{proof}

\begin{lemma}
\label{LEMMA:from_wffa_to_restricted_regexp}
Let $E \subseteq \EE_{\FF}$ and $\A \in \WFFA_{\Sigma, \FF}(E)$. Then there exists $\R \in \Reg_{\Sigma, \FF}^{\res}(E)$ such that $\beh{\R} = \beh{\A}$.
\end{lemma}

\begin{proof}[Proof sketch]
Let $\A = (Q, I, T, F, \wt_I, \wt_T, \wt_F) \in \WFFA_{\Sigma, \FF}(E)$ with $Q = \{1,\dots,m\}$.
For each transition $t = (p,a,r) \in T$, define a restricted expression
$
\R_t := [\wt_T(t)]_a \in \Reg_{\Sigma, \FF}^{\res}(E).
$
For all $p,r \in Q$ and $k \in \{0,\dots,m\}$, we define expressions
\(
\R_{p,r}^{(k)} \in \Reg_{\Sigma, \FF}^{\res}(E)
\)
recursively by
$
\R_{p,r}^{(0)} = \bigoplus_{t=(p,a,r)\in T} \R_t,
$
and
$
\R_{p,r}^{(k+1)}
=
\R_{p,r}^{(k)}
\oplus
\R_{p,k+1}^{(k)} \cdot \bigl(\R_{k+1,k+1}^{(k)}\bigr)^* \cdot \R_{k+1,r}^{(k)}.
$
Since $\R_{k+1,k+1}^{(k)}$ is $\eps$-free, its behavior is proper by Remark~\ref{REM:epsfree_proper}, and thus the star is admissible in the restricted syntax.
As in the usual state-elimination construction for weighted automata, one shows by induction on $k$ that $\R_{p,r}^{(k)}$ denotes the sum of the weights of all paths from $p$ to $r$ whose intermediate states lie in $\{1,\dots,k\}$. Hence $\R_{i,f}^{(m)}$ captures all paths from $i$ to $f$.
Finally, define
$
\R
=
\bigoplus_{i \in I,\, f \in F}
\bigl((\wt_I(i))_{\eps} \cdot \R_{i,f}^{(m)} \cdot (\wt_F(f))_{\eps}\bigr)
\in \Reg_{\Sigma, \FF}^{\res}(E).
$
Then $\beh{\R}$ coincides with $\beh{\A}$ on all non-empty words. If $\beh{\A}(\eps)\neq \zero$, we add the term $[\beh{\A}(\eps)]_{\eps}$.
\end{proof}

As a corollary of Theorem~\ref{THM:parametric_kleene_schuetzenberger}, we obtain:

\begin{corollary}
$
\beh{\Reg_{\Sigma, \FF}} = \beh{\WFFA_{\Sigma, \FF}}.
$
\end{corollary}

\section{Matrix Representation of WFFAs}
\label{SEC:matrix_representation}

In this section, we give an algebraic characterization of WFFAs in terms of matrix products. 
As a consequence, we obtain a polynomial-time procedure for evaluating their behavior. 
More precisely, the evaluation algorithm is quadratic in the number of states and linear in the length of the input word. 
This is noteworthy since a WFFA may have exponentially many runs of a given length, whereas the matrix-based evaluation avoids explicit run enumeration.

Throughout this section, let $\Sigma$ be an alphabet, let $\SS = (S, \oplus, \otimes, \zero, \one)$ be a semiring, and let $\FF = (\SS, \DD, \bb)$ be a finance semiring.

For $k,l \in \NN_{\ge 1}$, a \emph{$k \times l$ matrix} over a set $X$ is a mapping $A : [k] \times [l] \to X$. 
For matrices $A \in S^{k \times l}$ and $B \in S^{l \times m}$, their product $A \boxtimes B \in S^{k \times m}$ is defined by
$
(A \boxtimes B)_{i,j} = \bigoplus_{x \in [l]} A_{i,x} \otimes B_{x,j}.
$
Let $\UU^{(k)} \in S^{k \times k}$ denote the unit matrix.

It is well known that $(S^{k \times k}, \boxtimes, \UU^{(k)})$ is a monoid and that recognizable weighted languages admit a matrix characterization via monoid morphisms (see, e.g.,~\cite{DKV09}). 
We now adapt this viewpoint to WFFAs.

\begin{remark}
\label{REM:finance_homo_counterexample}
Not every morphism $\mu : \DD\Sigma_{\free}^* \to S^{k \times k}_{\boxtimes}$ corresponds to a WFFA. 
For instance, over $\FF_{\arc,\times}$, the QFL $\F: \DD \{\bot\}^* \to S$ defined by
$
\F(\bot_{d_1,\dots,d_n}) = \sum_{i = 1}^n d_i^2
$
admits such a matrix representation, but is not WFFA-recognizable, since every $\FF$-expression over $\FF_{\arc,\times}$ defines a piecewise-affine function.
\end{remark}

For an $\FF$-expression matrix $X \in (\EE_{\FF})^{k \times k}$, let $\beh{X}(d) \in S^{k \times k}$ be defined entrywise by
$
(\beh{X}(d))_{i,j} = \beh{X_{i,j}}(d).
$
Given a mapping $\xi : \Sigma \to (\EE_{\FF})^{k \times k}$, define the morphism
$
\mu_{\xi} : \DD\Sigma_{\free}^* \to S^{k \times k}_{\boxtimes}
$
by $\mu_{\xi}(\eps)=\UU^{(k)}$ and
\[
\mu_{\xi}((a_1,d_1)\dots(a_n,d_n))
=
\beh{\xi(a_1)}(d_1) \boxtimes \dots \boxtimes \beh{\xi(a_n)}(d_n).
\]

Let $E \subseteq \EE_{\FF}$. We say that a morphism $\mu : \DD\Sigma_{\free}^* \to S^{k \times k}_{\boxtimes}$ is \emph{$E$-generated} if there exists $\xi : \Sigma \to E^{k \times k}$ such that $\mu = \mu_{\xi}$.

\begin{theorem}
\label{THM:morphism_characterization_wffa}
Let $E \subseteq \EE_{\FF}$ with $\beh{\zero} \in \beh{E}$, and let $\F : \DD\Sigma^* \to S$. Then
$
{\F \in \beh{\WFFA_{\Sigma,\FF}(E)}}
$
if and only if there exist $k \in \NN_{\ge 1}$, vectors $\lambda \in S^{1 \times k}$ and $\nu \in S^{k \times 1}$, and an $E$-generated morphism
$
\mu : \DD\Sigma_{\free}^* \to S^{k \times k}_{\boxtimes}
$
such that
$
\F(w) = \lambda \boxtimes \mu(w) \boxtimes \nu
$
for all $w \in \DD\Sigma^*$.
\end{theorem}

\begin{proof}
The statement follows from Lemmas~\ref{LEMMA:wffa_to_morphism} and~\ref{LEMMA:morphism_to_wffa}.
\end{proof}

\begin{lemma}
\label{LEMMA:wffa_to_morphism}
Let $E \subseteq \EE_{\FF}$ with $\beh{\zero} \in \beh{E}$ and let $\A \in \WFFA_{\Sigma,\FF}(E)$. Then there exist $k \in \NN_{\ge 1}$, vectors $\lambda \in S^{1 \times k}$ and $\nu \in S^{k \times 1}$, and an $E$-generated morphism $\mu$ such that
$
\beh{\A}(w) = \lambda \boxtimes \mu(w) \boxtimes \nu
$
for all $w \in \DD\Sigma^*$.
\end{lemma}

\begin{proof}
Enumerate the states of $\A$ as $q_1,\dots,q_k$. 
Define $\lambda \in S^{1 \times k}$ and $\nu \in S^{k \times 1}$ by ${\lambda_i = \wt_I(q_i)}$ and ${\nu_i = \wt_F(q_i)}$ for all $i$, where missing weights are interpreted as $\zero$.
For each $a \in \Sigma$ define a matrix $\xi(a) \in E^{k \times k}$ by placing the transition weight of $(q_i,a,q_j)$ in position $(i,j)$, and an expression equivalent to $\zero$ otherwise. 
Then the induced morphism $\mu_{\xi}$ satisfies
$
\beh{\A}(w) = \lambda \boxtimes \mu_{\xi}(w) \boxtimes \nu,
$
as follows directly from the definition of WFFA behavior.
\end{proof}

\begin{lemma}
\label{LEMMA:morphism_to_wffa}
Let $k \in \NN_{\ge 1}$, $E \subseteq \EE_{\FF}$, $\lambda \in S^{1 \times k}$, $\nu \in S^{k \times 1}$, and $\xi : \Sigma \to E^{k \times k}$. 
If
$
\F(w)=\lambda \boxtimes \mu_{\xi}(w) \boxtimes \nu
$
for all $w \in \DD\Sigma^*$, then $\F \in \beh{\WFFA_{\Sigma,\FF}(E)}$.
\end{lemma}

\begin{proof}
Construct a WFFA with state set $[k]$, all states both initial and final, initial and final weights given by $\lambda$ and $\nu$, and transition weight of $(i,a,j)$ equal to $(\xi(a))_{i,j}$. 
By construction, its behavior coincides with $\F$.
\end{proof}

As a direct consequence, we obtain the following complexity bound.

\begin{corollary}
\label{COR:wffa_behavior_evaluation}
Assume that $\oplus$, $\otimes$, and $\bb$ can each be computed in $\O(1)$ time. 
Let ${\A \in \WFFA_{\Sigma,\FF}}$ and let $w \in \DD\Sigma^*$. Then $\beh{\A}(w)$ can be computed in
$
{\O\bigl((|\A|_Q)^2 \cdot |\A|_{\EE_{\FF}} \cdot |w|\bigr)}
$
time.
\end{corollary}

\begin{proof}
Let $k = |\A|_Q$. By Lemma~\ref{LEMMA:wffa_to_morphism}, write
$
\beh{\A}(w)=\lambda \boxtimes \mu_{\xi}(w) \boxtimes \nu.
$
Assume that ${w=(a_1,d_1)\dots(a_n,d_n)}$. We compute iteratively
\[
v_0=\lambda,
\qquad
v_i = v_{i-1} \boxtimes \beh{\xi(a_i)}(d_i)
\quad (i \in [n]).
\]
Each matrix evaluation $\beh{\xi(a_i)}(d_i)$ requires $\O(k^2 |\A|_{\EE_{\FF}})$ time, and each matrix-vector product has the same asymptotic cost. 
Thus the total running time is
$
\O(k^2 |\A|_{\EE_{\FF}} |w|).
$
\end{proof}

\begin{remark}
If $E = S$, $E = \PP_{\FF}$, or $E = \MM_{\FF}$, then $|\A|_{\EE_{\FF}} = \O(1)$. Hence the behavior can be evaluated in
$
\O((|\A|_Q)^2 \cdot |w|).
$
In particular, this applies to $\FF_{\arc,\times}$ when restricted to primitive, affine, or monomial transition weights.
\end{remark}

\section{Decision Problems for WFFAs}
\label{SEC:decision_problems}

In this section, we do not attempt an exhaustive study of decision problems for WFFAs. 
Instead, we illustrate how the reduction to monomials can be used to derive decidability results for two representative verification tasks: the support problem and the $>$-threshold problem.

\subsection{Support Problem}

The support problem asks whether a given WFFA assigns a value different from $\zero$ to some finance word.
Equivalently, it asks whether the quantitative finance language recognized by the automaton has non-empty support.

Let $X$ be a set and $f : X \to S$ a mapping. The \emph{support} of $f$ is defined as
$
{\supp(f) = \{x \in X \mid f(x) \neq \zero\}}.
$
For any finance word $w = (a_1,d_1)\dots(a_n,d_n) \in \DD\Sigma^*$, let
$
\pi_{\Sigma}(w) = a_1\dots a_n \in \Sigma^*.
$
For $L \subseteq \DD\Sigma^*$, let
$
\pi_{\Sigma}(L)=\{\pi_{\Sigma}(w)\mid w\in L\}\subseteq \Sigma^*.
$

\begin{lemma}
\label{LEMMA:expression_support_decidable}
Let $\FF = \FF_{\arc, \times}$ and let $e \in \EE_{\FF}$. Then it is decidable whether ${\supp(\beh{e}) \neq \emptyset}$.
\end{lemma}

\begin{proof}
By Lemma~\ref{LEMMA:finarc_exp_piecewise_affine}, $\beh{e}$ is piecewise affine. 
Hence there exists a finite partition of $\DD$ into points and intervals such that $\beh{e}$ is affine on each piece. 
It is therefore enough to inspect finitely many pieces. On each piece, one can decide whether the corresponding affine function is non-zero somewhere. Thus it is decidable whether $\supp(\beh{e}) \neq \emptyset$.
\end{proof}

\begin{lemma}
\label{LEMMA:rec_qfl_support_regular}
Let $\FF = \FF_{\arc, \times}$ and let $\A \in \WFFA_{\Sigma, \FF}$. Then $\pi_{\Sigma}(\supp(\beh{\A})) \subseteq \Sigma^*$ is a regular language.
\end{lemma}

\begin{proof}
Let $\A = (Q, I, T, F, \wt_I, \wt_T, \wt_F)$. Define the non-deterministic finite automaton
$
\A' = (Q, I', T', F')
$
by
$I' = \{i \in I \mid \wt_I(i) \neq \zero\},$
$F' = \{f \in F \mid \wt_F(f) \neq \zero\},$
and
$
T' = \{(p,a,q) \in T \mid \supp(\beh{\wt_T(p,a,q)}) \neq \emptyset\}.
$
Then a word $u = a_1 \dots a_n \in \Sigma^*$ is accepted by $\A'$ if and only if there exists a finance word
$
w=(a_1,d_1)\dots(a_n,d_n)\in \DD\Sigma^*
$
with $\pi_{\Sigma}(w)=u$ and $\beh{\A}(w)\neq \zero$. Hence
$
L(\A')=\pi_{\Sigma}(\supp(\beh{\A})),
$
so $\pi_{\Sigma}(\supp(\beh{\A}))$ is regular.
\end{proof}

\begin{lemma}
\label{LEMMA:wffa_support_decidable}
Let $\FF = \FF_{\arc, \times}$ and let $\A \in \WFFA_{\Sigma, \FF}$. Then it is decidable whether $\supp(\beh{\A}) \neq \emptyset$.
\end{lemma}

\begin{proof}
By Lemma~\ref{LEMMA:rec_qfl_support_regular}, the language $\pi_{\Sigma}(\supp(\beh{\A}))$ is regular and effectively constructible. 
Moreover,
$
\supp(\beh{\A}) \neq \emptyset
$
if and only if
$
\pi_{\Sigma}(\supp(\beh{\A})) \neq \emptyset.
$
Since emptiness of regular languages is decidable (see, e.g., \cite{HU79}), the claim follows.
\end{proof}

\subsection{Threshold Problems}

In this subsection, we focus on the finance semiring $\FF = \FF_{\arc, \times}$. 
In financial terms, the $>$-threshold problem asks whether a modeled instrument or strategy can exceed a prescribed payoff level under some market scenario.

Let $\D$ be a scenario language, i.e., a set of finance words representing possible scenarios.

\begin{definition}
\label{DEF:threshold_problem}
Let ${\bowtie} \in \{<,>\}$. 
The \emph{$\bowtie$-threshold problem over $\D$} asks, given ${\A \in \WFFA_{\Sigma, \FF}}$ and $\theta \in \RR$, whether there exists a finance word $w \in \D$ such that
$
\beh{\A}(w) \bowtie \theta.
$
If $\D = \DD\Sigma^*$, we simply speak of the \emph{$\bowtie$-threshold problem}.
\end{definition}

\begin{remark}
For WFFAs with constant transition weights, the threshold problem reduces to the corresponding problem for classical max-plus weighted automata; the relevant decidability results are well known (see, e.g.,~\cite{ABK22}).
\end{remark}

Let $\I \subseteq \DD$. We denote by $\DD_{\I}\Sigma^*$ the set of all finance words (scenarios)
$
w=(a_1,d_1)\dots(a_n,d_n)\in \DD\Sigma^*
$
such that $d_i \in \I$ for all $i \in [n]$. 
This defines a basic class of scenario languages with a global constraint on the data values, typically given by an interval $\I = [\ell, h]$.

\begin{theorem}
\label{THM:threshold_problem_wffa}
\begin{enumerate}[label=(\alph*)]
\item The $>$-threshold problem for $\WFFA_{\Sigma, \FF}$ is decidable.
\item Let $\I = [\ell,h] \subseteq \DD$ be an interval. Then the $>$-threshold problem over $\D = \DD_{\I}\Sigma^*$ for $\WFFA_{\Sigma, \FF}$ is decidable.
\end{enumerate}
\end{theorem}

For the proof of Theorem~\ref{THM:threshold_problem_wffa}, we require the following lemma.

\begin{lemma}
\label{LEMMA:expression_sup_computable}
Let $e \in \EE_{\FF}$ and let either $\I = \DD$ or $\I = [l,h] \subseteq \DD$. Then $\sup_{x \in \I}\beh{e}(x)$ is computable.
\end{lemma}

\begin{proof}
By Lemma~\ref{LEMMA:finarc_exp_piecewise_affine}, $\beh{e}$ is piecewise affine. 
Hence there exists a finite partition of $\DD$ into points and intervals such that $\beh{e}$ is affine on each piece. 
The supremum over $\DD$ or over $[l,h]$ is therefore attained either at a breakpoint or at an endpoint of one of the finitely many relevant intervals, and is thus computable.
\end{proof}

\begin{remark}
\label{REM:complexity_sup}
Let $\I=\DD$ or $\I=[l,h]\subseteq \DD$.
\begin{enumerate}[label=(\alph*)]
\item If, for a class $E \subseteq \EE_{\FF}$, every $e \in E$ defines a piecewise-affine function with $\O(|e|)$ pieces, then $\sup_{x \in I} \beh{e}(x)$ can be computed in $\O(|e|)$ time.
\item In particular, for $E = \MM_{\FF}$, all expressions have constant-size piecewise-affine structure, so $\sup_{x \in \I}\beh{e}(x)$ is computable in $\O(|e|)$ time.
\end{enumerate}
\end{remark}

\begin{proof}[Proof of Theorem~\ref{THM:threshold_problem_wffa}]
Let $\A = (Q, I, T, F, \wt_I, \wt_T, \wt_F) \in \WFFA_{\Sigma, \FF}$, let $\I = \DD$ or ${\I=[l,h]}$, and let $\theta \in \RR$. 
For each transition $t \in T$, compute
$
M_t = \sup_{x \in \I}\left(\beh{\wt_T(t)}(x)\right)
$
using Lemma~\ref{LEMMA:expression_sup_computable}. If some transition $t$ with $M_t = \infty$ occurs on a path from an initial state to a final state using only transitions with supremum different from $-\infty$, then the answer is positive.
Otherwise, transitions with $M_t=-\infty$ can be removed, and every remaining transition can be replaced by the constant weight $M_t$. 
This yields a classical max-plus weighted automaton $\A'$ over $\Sigma$ such that
\[
\exists w \in \DD_\I\Sigma^*:\ \beh{\A}(w)>\theta
\quad\Longleftrightarrow\quad
\exists u \in \Sigma^*:\ \beh{\A'}(u)>\theta.
\]
The claim now follows from the decidability of the $>$-threshold problem for max-plus weighted automata.
\end{proof}

\begin{remark}
By Theorem~\ref{THM:threshold_problem_wffa} and the constant-size representation of monomial expressions, the $>$-threshold problem for $\WFFA_{\Sigma,\FF}(\MM_{\FF})$ is decidable in polynomial time.
\end{remark}

Since Theorem~\ref{THM:kleene_schuetzenberger} yields an effective translation from weighted finance regular expressions to WFFAs, the preceding decidability results immediately transfer to weighted finance regular expressions.

\begin{corollary}
Let $\FF = \FF_{\arc, \times}$ and let $\R \in \Reg_{\Sigma, \FF}$. Then the following hold:
\begin{enumerate}[label=(\alph*)]
\item It is decidable whether $\supp(\beh{\R}) \neq \emptyset$.
\item Given $\theta \in \RR$, it is decidable whether there exists $w \in \DD\Sigma^*$ such that $\beh{\R}(w) > \theta$.
\item Given $\theta \in \RR$ and an interval $\I = [l,h] \subseteq \DD$, it is decidable whether there exists $w \in \DD_\I\Sigma^*$ such that $\beh{\R}(w) > \theta$.
\end{enumerate}
\end{corollary}

\begin{remark}
The $<$-threshold problem is closely related to risk analysis, where one asks whether there exists a scenario whose payoff stays below a given risk level. 
In the deterministic case (i.e., when each word induces at most one run), the $<$- and $>$-threshold problems are inter-reducible. 
In contrast, for nondeterministic WFFA over $\FF_{\arc, \times}$, threshold problems become undecidable, as follows from results on max-plus automata~\cite{ABK22}.
\end{remark}

\section{Conclusion}

We introduced weighted finite finance automata (WFFAs), a novel automata-theoretic framework for modeling quantitative aspects of financial systems.
We established closure properties of the model and proved a Kleene--Sch\"utzenberger-type characterization via weighted finance regular expressions, yielding an effective translation between the two formalisms.
Moreover, we provided an algebraic matrix representation of WFFAs, which leads to a polynomial-time evaluation procedure.
Finally, we applied these results to derive decidability results for fundamental verification tasks, including the support problem and the threshold problem.

Several directions for future research remain open.
On the algorithmic side, a systematic complexity analysis of the constructions and decision procedures developed in this paper is required.
In particular, the decidability landscape of WFFAs is not yet fully understood.
Beyond the threshold problem considered here, an important direction is to identify classes of quantitative finance languages that capture realistic market scenarios while preserving decidability.

From a modeling perspective, WFFAs can be extended to incorporate multiple interacting financial variables, enabling the analysis of systems driven by correlated factors such as multiple asset prices, exchange rates, and portfolio dynamics, thereby opening the way to applications in portfolio management.
Another promising direction is the extension to continuous time, in the spirit of weighted timed automata~\cite{LBBFHP01}, which would allow modeling of time-sensitive financial processes.

Another direction for future work concerns the algebraic encoding discussed in Remark~\ref{REM:ks_relation}. While the present paper focuses on direct constructions tailored to WFFAs, it would be interesting to investigate a fully algebraic treatment of the model. Such an approach may provide further structural insight and could be useful for studying extensions of WFFAs with additional features.

Altogether, WFFAs provide a unified automata-theoretic framework for modeling and analyzing quantitative aspects of financial systems, combining expressiveness with decidability and efficient evaluation.

\bibliographystyle{amsplain}
\bibliography{WFFA}

\newpage

\begin{appendices}

\section{Detailed Proofs}

This appendix provides detailed proofs of the main results as well as additional technical arguments omitted from the main text for clarity. 
All results stated in the paper are proved in full detail here.

\subsection{Closure Properties of WFFAs (Section~\ref{SEC:closure_wffa})}

\begin{lemma}[Lemma~\ref{LEMMA:closure_sum_wffa}]
Let $\A_1, \A_2 \in \WFFA_{\Sigma, \FF}$. Then there exists $\A_{\oplus} \in \WFFA_{\Sigma, \FF}$ such that 
$\beh{\A_{\oplus}} = \beh{\A_1} \oplus \beh{\A_2}$.
\end{lemma}

\begin{proof}
We use the standard disjoint union construction for non-deterministic automata. 
For $i \in \{1,2\}$, let $\A_i = (Q_i, I_i, T_i, F_i, \wt_{I_i}, \wt_{T_i}, \wt_{F_i})$. We may assume that $Q_1 \cap Q_2 = \emptyset$. Then, we define $\A_{\oplus}$ by
\[
\A_{\oplus} = (Q_1 \cup Q_2, I_1 \cup I_2, T_1 \cup T_2, F_1 \cup F_2, \wt_{I_1} \cup \wt_{I_2}, \wt_{T_1} \cup \wt_{T_2}, \wt_{F_1} \cup \wt_{F_2}).
\]
Here, the union of functions is defined as follows. For $i \in \{1, 2\}$, let $f_i: X_i \to Y_i$ such that $X_1 \cap X_2 = \emptyset$. Then, $(f_1 \cup f_2): (X_1 \cup X_2) \to (Y_1 \cup Y_2)$ is defined  by $(f_1 \cup f_2)(x_1) = f_1(x_1)$ for all $x_1 \in X_1$ and by $(f_1 \cup f_2)(x_2) = f_2(x_2)$ for all $x_2 \in X_2$.

Since the state sets $Q_1$ and $Q_2$ are disjoint, every run of $\A_{\oplus}$ lies entirely within either $\A_1$ or $\A_2$. Thus, the set of runs of $\A_{\oplus}$ on any word $w$ is the disjoint union of the sets of runs of $\A_1$ and $\A_2$ on $w$. Moreover, all weights (initial, transition, final) are preserved for the respective components. Therefore,
we obtain
\[
\beh{\A_{\oplus}}(w) = \beh{\A_1}(w) \oplus \beh{\A_2}(w).
\]
\qedhere
\end{proof}

\begin{lemma}[Lemma~\ref{LEMMA:closure_hadamard_wffa}]
Let $\A_1, \A_2 \in \WFFA_{\Sigma, \FF}$. Then there exists 
$\A_{\otimes} \in \WFFA_{\Sigma, \FF}$ such that 
$\beh{\A_{\otimes}} = \beh{\A_1} \otimes \beh{\A_2}$.
\end{lemma}

\begin{proof}
In this case, we can adopt the standard product construction for weighted automata over commutative semirings (see, e.g., \cite{DKV09}).
Let
$
\A_i = (Q_i, I_i, T_i, F_i, \wt_{I_i}, \wt_{T_i}, \wt_{F_i})
$ ($i = 1, 2$).
We construct the product automaton
\[
\A_{\otimes} = (Q, I, T, F, \wt_I, \wt_T, \wt_F)
\]
where:
\begin{itemize}
\item $Q = Q_1 \times Q_2$, $I = I_1 \times I_2$, $F = F_1 \times F_2$;
\item $T$ contains all transitions $((q_1, q_2), a, (q_1', q_2'))$ such that $(q_1, a, q_1') \in T_1$ and $(q_2, a, q_2') \in T_2$;
\item for $(i_1, i_2) \in I$, $\wt_I(i_1, i_2) = \wt_{I_1}(i_1) \otimes \wt_{I_2}(i_2)$;
\item for $(f_1, f_2) \in F$, $\wt_F(f_1, f_2) = \wt_{F_1}(f_1) \otimes \wt_{F_2}(i_2)$;
\item for $t = ((q_1, q_2), a, (q_1', q_2')) \in T$, $\wt_T(t) = [\wt_{T_1}(q_1, a, q_1') \otimes \wt_{T_2}(q_2, a, q_2')]$. Note that here $\otimes$ is applied to $\FF$-expressions.
\end{itemize}
It can be shown that $\beh{\A_{\otimes}} = \beh{\A_1} \otimes \beh{\A_2}$. The idea is that, for each word $w$, there is a one-to-one correspondence between the runs of $\A_{\otimes}$ and the pairs of runs of $\A_1$ and $\A_2$ on the same word $w$. By pairing states and synchronizing transitions on the same input symbol, each accepting run in the product automaton carries the weight of the first run multiplied by the weight of the second run (using the commutativity of the semiring $\SS$). Summing over all runs on $w$ then yields the Hadamard product of the behaviors of $\A_1$ and $\A_2$.
\end{proof}

\begin{lemma}[Lemma~\ref{LEMMA:normalize_wffa}]
Let $\A \in \WFFA_{\Sigma, \FF}$. Then:
\begin{enumerate}[label=(\alph*)]
\item there exists an initially normalized $\A_I$ with $\beh{\A_I}=\beh{\A}$;
\item there exists a finally normalized $\A_F$ with $\beh{\A_F}=\beh{\A}$;
\item if $\beh{\A}$ is proper, there exists a completely normalized $\A_{I,F}$ with $\beh{\A_{I,F}}=\beh{\A}$.
\end{enumerate}
\end{lemma}

\begin{proof}
Let $\A = (Q, I, T, F, \wt_I, \wt_T, \wt_F)$. 
\begin{enumerate}[label=(\alph*)]
\item Let $q_I$ be an element with $q_I \notin Q$. We define the initially normalized WFFA
\[
\A_I = (Q \cup \{q_I\}, \{q_I\}, T', F', [q_I \mapsto \one], \wt_{T'}, \wt_{F'})
\]
as follows. 
\begin{itemize}
\item We define the new transition set
\[
T' = T \cup \{(q_I, a, q)  \mid  (i, a, q) \in T \text{ and } i \in I\}
\]
and the transition weights $\wt_{T'}$ by 
\[
\wt_{T'}(t) = \wt_T(t)
\] for all $t \in T$
and
\[
\wt_{T'}(q_I, a, q) = \left[\bigoplus (\wt_I(i) \otimes \wt_T(i, a, q)  \mid  (i, a, q) \in T \text{ and } i \in I) \right].
\]
for all $(q_I, a, q) \in T' \setminus T$. Here, formally, we choose any order for the summands and any parenthesizing. Note that the order of the summands  and their parenthesizing in $\FF$-expression above are irrelevant, as its evaluation is performed using the commutative $\oplus$-operation in the semiring $\SS$.
\item If $\beh{\A_{I}}(\eps) = \zero$, then we let $F' = F$ and $\wt_{F'} = \wt_{F}$.
\item If $\beh{\A_{I}}(\eps) = s_{\eps} \neq \zero$, then we let $F' = F \cup \{q_I\}$ and $\wt_{F'} = \wt_{F} \cup [q_I \mapsto s_{\epsilon}]$.
\end{itemize}
Similarly to the corresponding construction for semiring-weighted automata (see, e.g., \cite{DKV09}), it can be shown that $\beh{\A_{I}} = \beh{\A}$. Every run of $\A$ starting from an original initial state $i \in I$ is simulated in $\A_I$ by a run starting from the new state $q_I$, whose outgoing transitions encode the initial weights $\wt_I(i)$. Likewise, the treatment of $q_I$ as a final state for $\eps$ ensures that the empty-word behavior is preserved, so overall $\beh{\A_{I}} = \beh{\A}$.
\item We proceed dually to (a).
\item Choose $q_I, q_F \notin Q$ such that $q_I \neq q_F$ and define the completely normalized WFFA
\[
\A_{I, F} = (Q \cup \{q_I, q_F\}, \{q_I\}, T', \{q_F\}, [q_I \mapsto \one], \wt_{T'}, [q_F \mapsto \one])
\]
where:
\begin{itemize}
\item $T' = T \cup T_{I} \cup T_{I,F} \cup T_{F}$ where:
\begin{itemize}
\item $T_{I} = \{(q_I, a, q) \; | \; (i, a, q) \in T \text{ for some } i \in I\}$;
\item $T_{I,F} = \{(q_I, a, q_F)  \mid  (i, a, f) \in T \text{ for some } i \in I \text{ and } f \in F\}$;
\item $T_{F} = \{(q, a, q_F)  \mid  (q, a, f) \in T \text{ for some } f \in F\}$;
\end{itemize}
\item $\wt_{T'}$ is defined as follows:
\begin{itemize}
\item for all $t \in T$, $\wt_{T'}(t) = \wt_T(t)$;
\item for all $(q_I, a, q) \in T_{I}$,
\[
\wt_{T'}(q_I, a, q) = \left[ \bigoplus(\wt_I(i) \otimes \wt_T(i, a, q)  \mid  (i, a, q) \in T \text{ and } i \in I) \right]
\]
\item for all $(q, a, q_F) \in T_{F}$
\[
\wt_{T'}(q, a, q_F) = \left[\bigoplus (\wt_T(q, a, f) \otimes \wt_F(f)  \mid  (q, a, f) \in T \text{ and } f \in F) \right].
\]
\item for all $(q_I, a, q_F) \in T_{I, F}$,
\[
\begin{aligned}
\wt_{T'}(q_I, a, q_F) = \left[
    \bigoplus (\wt_I(i) \otimes \wt_T(i, a, f) \otimes \wt_F(f) \mid (i, a, f) \in T_{I, a, F}) \right]
\end{aligned}
\]
where $T_{I, a, F} = T \cap (I \times \{a\} \times F)$.
\end{itemize}
Since $\beh{\A}(\eps) = \zero$, we may assume that $I \cap F = \emptyset$. Then, $\A_{I, F}$ is completely normalized. Moreover, it can be shown that the construction correctly simulates all runs of the original automaton $\A$. 
Multiple runs of $\A$ may correspond to the same run in $\A_{I,F}$, but by definition of the initial, transition, and final weights, the weight of each run in $\A_{I,F}$ equals the sum of the weights of the corresponding runs in $\A$. In particular, the contributions of $\eps$-runs are preserved. Hence, we have
$
\beh{\A_{I, F}} = \beh{\A}.
$
\qedhere
\end{itemize}
\end{enumerate}
\end{proof}

\begin{lemma}[Lemma~\ref{LEMMA:cauchy_product_kleene_star_normalized_wffa}]
\begin{enumerate}[label=(\alph*)]
\item 
Let $\A_F \in \WFFA_{\Sigma, \FF}$ be finally normalized and $\A_I \in \WFFA_{\Sigma, \FF}$ be initially normalized. 
Then there exists $\A_{\mathrm{Cauchy}} \in \WFFA_{\Sigma, \FF}$ such that
$\beh{\A_{\mathrm{Cauchy}}}=\beh{\A_F}\cdot\beh{\A_I}$.

\item 
Let $\A \in \WFFA_{\Sigma, \FF}$ be completely normalized. 
Then there exists $\A^* \in \WFFA_{\Sigma, \FF}$ with $\beh{\A^*}=\beh{\A}^*$.
\end{enumerate}
\end{lemma}

\begin{proof}
\begin{enumerate}[label=(\alph*)]
\item Assume that
\[
\A_1 = (Q_1, I_1, T_1, \{q_{F_1}\}, \wt_{I_1}, \wt_{T_1}, [q_{F_1} \mapsto \one])
\]
and 
\[
\A_2 = (Q_2, \{q_{I_2}\}, T_2, F_2, [q_{I_2} \mapsto \one], \wt_{T_2}, \wt_{F_2}).
\]
We can choose $Q_1$ and $Q_2$ such that $q_{F_1} = q_{I_2}$ and $Q_1 \cap Q_2 = \{q_{F_1}\}$. Note that $T_1 \cap T_2 = \emptyset$ in this case. Then, we let
\[
\A_{\mathrm{Cauchy}} = (Q_1 \cup Q_2, I_1, T_1 \cup T_2, F_2, \wt_{I_1}, \wt_{T_1} \cup \wt_{T_2}, \wt_{F_2}).
\]
By construction, every run of $\A_{\mathrm{Cauchy}}$ consists of a run of $\A_1$ from some initial state to $q_{F_1} = q_{I_2}$, immediately followed by a run of $\A_2$ from $q_{I_2}$ to some final state in $F_2$. The weight of such a run is the product of the weights of the $\A_1$ and $\A_2$ components, because the final weight of $q_{F_1}$ in $\A_1$ and the initial weight of $q_{I_2}$ in $\A_2$ are both set to $\one$. Summing over all runs then yields the Cauchy product of $\beh{\A_1}$ and $\beh{\A_2}$
\item Let $\A = (Q, \{q_I\}, T, \{q_F\}, [q_I \mapsto \one], \wt_T, [q_F \mapsto \one])$. Since $\beh{\A}(\varepsilon) = \zero$, we have $q_I \neq q_F$. We define $\A^*$ from $\A$ by merging the states $q_I$ and $q_F$, i.e.
\[
\A^* = (Q \setminus \{q_F\}, \{q_I\}, T', \{q_I\}, [q_I \mapsto \one], \wt'_T, [q_I \mapsto \one]).
\]
where:
\begin{itemize}
\item $T' = T'_1 \cup T'_2$ such that
\begin{itemize}
\item $T_1' = \{(q, a, q') \in T \mid q \neq q_F \neq q'\} \subseteq T$,
\item $T_2' =  \{(q, a, q_I) \mid (q, a, q_F) \in T\}$;
\end{itemize}
\item $\wt_{T'}$ is defined for all $t_1' \in T_1'$ by $\wt_{T'}(t_1') = \wt_{T}(t_1')$ and, for all $(q, a, q_I) \in T_2'$, by $\wt_{T'}(q, a, q_I) = \wt_{T}(q, a, q_F)$.
\end{itemize}
By merging the initial state $q_I$ and the final state $q_F$, every run of $\A^*$ corresponds to a sequence of one or more runs of $\A$ concatenated end-to-end. Each transition in $\A^*$ that originally ended at $q_F$ now loops back to $q_I$, so the weight of a run in $\A^*$ is exactly the product of the weights of the corresponding runs in $\A$. Summing over all such sequences, then gives the Kleene star of $\beh{\A}$, i.e., $\beh{\A^*} = \beh{\A}^*$.
\qedhere
\end{enumerate}
\end{proof}

\subsection{Restricted WFFAs (Section~\ref{SEC:restricted_wffa})}

\begin{lemma}[Lemma~\ref{LEMMA:init_and_final_normalization_for_arbitrary_expressions}]
Let $E \subseteq \EE_{\FF}$ be any collection of $\FF$-expressions and $\A \in \WFFA_{\Sigma, \FF}(E)$.
\begin{enumerate}[label=(\alph*)]
\item There exists an initially normalized $\A_{I} \in \WFFA_{\Sigma, \FF}(E)$ such that $\beh{\A_{I}} = \beh{\A}$.
\item There exists a finally normalized $\A_{F} \in \WFFA_{\Sigma, \FF}(E)$ such that $\beh{\A_{F}} = \beh{\A}$.
\end{enumerate}
\end{lemma}

\begin{proof}
Let $\A = (Q, I, T, F, \wt_I, \wt_T, \wt_F)$.
\begin{enumerate}[label=(\alph*)]
\item
\label{LEMMA:init_normalization_for_arbitrary_expressions}
Note that in the proof of Lemma~\ref{LEMMA:normalize_wffa}, the initial weights were eliminated by incorporating them into the transition weights. However, this approach cannot be applied if the collection of expressions is neither $\oplus$- nor $(S, \otimes)$-closed.

To overcome this limitation, we introduce an alternative normalization procedure. The key idea is that, when computing the weight of a WFFA run over a commutative semiring, the contribution of the initial weight can be deferred until a final state is reached. To implement this, we extend the states of $\A$ with an additional component that records the initial state of each run. Using this information, we can eliminate initial weights by transferring their contribution to the final weights. 

Formally, we choose a new state $q_I \notin Q$ and let
\[
\A_{I} = (\{q_I\} \cup (I \times Q), \{q_I\}, T_{I}, F_{I}, [q_I \mapsto \one], \wt_{T_{I}}, \wt_{F_{I}})
\]
where:
\begin{itemize}
\item $T_{I}$ and $\wt_{T_{I}}$ are defined as follows. For all $t = (q, a, q') \in T$ and all $i \in I$, we let $t' = ((i, q), a, (i, q')) \in T_{I}$ and $\wt_{T_{I}}(t') = \wt_T(t)$. Furthermore, for all $t = (i, a, q) \in T$ with $i \in I$, we let $t'' = (q_I, a, (i, q)) \in T_{I}$ and $\wt_{T_{I}}(t'') = \wt_T(t)$.
\item If $\beh{\A}(\eps) = \zero$, then $F_{I} = I \times F$. Otherwise, $F_{I} = \{q_I\} \cup (I \times F)$.
\item For all $i \in I$ and $f \in F$, $\wt_{F_{I}}(i, f) = \wt_I(i) \otimes \wt_F(f)$. If $\beh{\A}(\eps) \neq \zero$, then we let $\wt_{F_{I}}(q_I) = \beh{\A}(\eps)$.
\end{itemize}
Then, using commutativity of $\SS$, it can be shown that $\beh{\A_{I}} = \beh{\A}$.
\item This case is symmetric to \ref{LEMMA:init_normalization_for_arbitrary_expressions}.
\qedhere
\end{enumerate}
\end{proof}

\begin{lemma}[Lemma~\ref{LEMMA:plus_closed_expressiveness}]
Let $E \subseteq \EE_{\FF}$ be any collection of $\FF$-expressions. Then
$
\beh{\WFFA_{\Sigma, \FF}(\clp{E})} = \beh{\WFFA_{\Sigma, \FF}(E)}.
$
\end{lemma}

\begin{proof}
For any $e = [e_1 \oplus {\dots} \oplus e_k] \in \langle E \rangle_{\oplus}$ with $e_i \in E$, let $\dist{e}_E = k$ and, for any $j \in [k]$, let $e_E^{(j)} = e_j$.
Let
\[
\A_{\clp{E}} = (Q, I, T, F, \wt_I, \wt_T, \wt_F) \in \WFFA_{\Sigma, \FF}(\clp{E})
\]. Our goal is to construct $\A_E \in \WFFA_{\Sigma, \FF}(E)$ such that $\beh{\A_E} = \beh{\A_{\clp{E}}}$. The idea of our construction is to split each transition of $\A_{\clp{E}}$ with weight $e \in \clp{E}$ into $\dist{e}_E$ many transitions with weights $e^{(1)}_E, \dots, e^{(\dist{e}_E)}_E$. For this, we enlarge the state set.

For each state $q \in Q$, let $\prev(q) = \{t = (p, a, q) \in T  \mid  p \in Q \text{ and } a \in \Sigma\}$ and let
\[
\dist{q}_E = \begin{cases}
\max \{\dist{\wt_T(t)}_E  \mid  t \in \prev(q) \}, & \text{if } \prev(q) \neq \emptyset, \\
1, & \text{otherwise}.
\end{cases}
\]
Then, we construct 
\[
\A_E = (Q', I', T', F', \wt_{I'}, \wt_{T'}, \wt_{F'}) \in \WFFA_{\Sigma, \FF}(E)
\]
where:
\begin{itemize}
\item $Q' = \bigl\{(q, i)  \mid  q \in Q \text{ and } i \in \{1, \dots, \dist{q}_E\}\bigr\}$;
\item $I' = Q \times \{1\}$ and, for every $q_0' = (q_0, 1) \in I'$, we let $\wt_{I'}(q_0') = \wt_{I}(q_0)$;
\item $F' = \{(q_f, i)  \mid  q_f \in F \text{ and } i \in \{1, \dots, \dist{q_f}_E\}\}$ and, for every $q_f' = (q_f, i) \in F'$, we let $\wt_{F'}(q_f') = \wt_{F}(q_f)$;
\item for every $t = (p, a, q) \in T$, every $i_p \in \{1, \dots, \dist{p}_E\}$ and every ${i_q \in \{1, \dots, \dist{\wt_T(t)}_E\}}$, we let $t' = ((p, i_p), a, (q, i_q)) \in T'$ and $\wt_{T'}(t') = (\wt_{T}(t))_E^{(i_q)}$.
\end{itemize}
Then, $\beh{\A_E} = \beh{\A_{\clp{E}}}$ and hence $\beh{\WFFA_{\Sigma, \FF}(\clp{E})} \subseteq \beh{\WFFA_{\Sigma, \FF}(E)}$. The opposite inclusion is obvious, as $E \subseteq \clp{E}$.
\end{proof}

\begin{lemma}[Lemma~\ref{LEMMA:closure_kleene_star_for_e}]
\begin{enumerate}[label=(\alph*)]
\item Let $E \subseteq \EE_{\FF}$ be $(S, \otimes)$-closed and let $\F \in \beh{\WFFA_{\Sigma, \FF}(E)}$ be proper. Then $\F^* \in \beh{\WFFA_{\Sigma, \FF}(E)}$.
\item  Let $E \subseteq \EE_{\FF}$ be arbitrary and let $\F \in \beh{\WFFA^{\tau}_{\Sigma, \FF}(E)}$ be proper. Then ${\F^* \in \beh{\WFFA^{\tau}_{\Sigma, \FF}(E)}}$.
\end{enumerate}
\end{lemma}

\begin{proof}
\begin{enumerate}[label=(\alph*)]
\item Let $\A \in \WFFA_{\Sigma, \FF}(E)$ such that $\beh{\A} = \F$. First, we construct a WFFA $\A^* \in \WFFA_{\Sigma, \FF}(\EE_{\FF})$ with $\beh{\A} = \F^*$ as in the proof of Corollary \ref{COR:cauchy_product_kleene_star_full_wffa}. We can observe that the transition weights of $\A^*$ are the following expressions:
\begin{itemize}
\item $e \in E \subseteq \clp{E}$.
\item $e = [(s_1 \otimes e_1) \oplus {\dots} \oplus (s_k \otimes e_k)]$ where $s_1, \dots, s_k \in S$ and $e_1, \dots, e_k \in E$. Since $E$ is $(S, \otimes)$-closed, for every $i \in [k]$ there exists $\hat{e}_i \in E$ with $\beh{\hat{e}_i} = \beh{s_i \otimes e_i}$. Let $\hat{e} = [\hat{e}_1 \oplus {\dots} \oplus \hat{e}_k] \in \clp{E}$. Then, clearly, $\beh{\hat{e}} = \beh{e}$. Then, $e$ can be replaced with an equivalent expression from $\clp{E}$ without altering the behavior of $\A_{\iota, \phi}$.
\item $e = [(e_1 \otimes s_1) \oplus {\dots} \oplus (e_k \otimes s_k)]$ where $s_1, \dots, s_k \in S$ and $e_1, \dots, e_k \in E$. Since $\otimes$ is commutative, $\beh{e} = \beh{(s_1 \otimes e_1) \oplus {\dots} \oplus (s_1 \otimes e_k)}$. Then, by (b), $e$ can be also replaced with an equivalent expression from $\clp{E}$.
\item $e = [(s_1 \otimes e_1 \otimes s_1') \oplus {\dots} \oplus (s_k \otimes e_k \otimes s_k')]$ where $s_1, s_1', \dots, s_k, s_k' \in S$ and $e_1, \dots, e_k \in E$.
Since $\otimes$ is commutative, $\beh{e} = \beh{(\hat{s}_1 \otimes e_1) \oplus {\dots} \oplus (\hat{s}_k \otimes e_k)}$ where $\hat{s}_i = (s_1 \otimes s_1') \in S$ for all $i \in [k]$. Then, again, by (b), $e$ can also be replaced with an equivalent expression from $\clp{E}$.
\end{itemize}
These observations imply that, by modifying only transition weights, we may transform $\A^*$ into an automaton $\A^*_{\clp{E}} \in \WFFA_{\Sigma, \FF}(\clp{E})$ without changing its behavior, i.e., $ \beh{\A^*_{\clp{E}}} = \F^*$. By Lemma~\ref{LEMMA:plus_closed_expressiveness}, there exists an automaton $\A^*_{E} \in \WFFA_{\Sigma, \FF}(E)$ with the same behavior, $\beh{\A^*_{E}} = \F^*$. Consequently, $\F^* \in \beh{\WFFA_{\Sigma, \FF}(E)}$.
\item The proof follows the same construction as the proof of (a). We only need to observe that the transition weights of $\A^*$ are the expressions of one of the following forms:
\begin{itemize}
\item $e \in E \subseteq \clp{E}$;
\item $e = [(\one \otimes e_1) \oplus {\dots} \oplus (\one \otimes e_k)]$ with $e_1, \dots, e_k \in E$;
\item $e = [(e_1 \otimes \one) \oplus {\dots} \oplus (e_k \otimes \one)]$ with $e_1, \dots, e_k \in E$;
\item $e = [(\one \otimes e_1 \otimes \one) \oplus {\dots} \oplus (\one \otimes e_k \otimes \one)]$ with $e_1, \dots, e_k \in E$.
\end{itemize}
In each of the last three cases, the expression $e$ is equivalent to ${[e_1 \oplus {\dots} \oplus e_n]}$, which lies in $\clp{E}$. Replacing the transition weights accordingly and then applying the construction from the proof of part (a) yields an automaton $\A^*_E \in \WFFA_{\Sigma, \FF}(E)$ with $\beh{\A_E^*} = \F^*$. Finally, inspecting the construction in the proof of Lemma~\ref{LEMMA:plus_closed_expressiveness} shows that $\A^*_E$ in fact belongs to $\WFFA_{\Sigma, \FF}^{\tau}(E)$.
\qedhere
\end{enumerate}
\end{proof}

\subsubsection{Primitive Expressions (Subsection~\ref{SUBSEC:primitive_expressions})}

\begin{lemma}[Lemma~\ref{LEMMA:primitive_wffa_cube_property}]
Let $\Sigma = \{\bot\}$, let $\FF = \FF_{\arc, \times}$, and let $\A \in \WFFA_{\Sigma, \FF}(\PP_{\FF})$ be such that $\beh{\A}(w) \neq \zero$ for all $w \in \DD \Sigma^*$. 
Then for every $n \in \NN_{\ge 1}$ there exist $\varrho \in \Run_{\A}^n$, $\delta \in \DD^n$, and $R > 0$ such that
$
\beh{\A}(\bot_{\delta'}) = \wt_{\A}(\varrho, \delta')
$
for all $\delta' \in \mathcal C(\delta, R)$.
\end{lemma}

\begin{proof}
Let $\A = (Q, I, T, F, \wt_I, \wt_T, \wt_F)$. We may assume that $\wt_T(t) \notin \{\zero, \bind{\zero}\}$ for all $t \in T$.
We fix a enumeration $\Run_{\A}^n = (\varrho_i)_{i \in [k]}$ where $k =|\Run_{\A}^n|$. 

For each $j \in [k]$ and $\delta = (d_1, \dots, d_n) \in \DD^n$, the weight $\wt_{\A}(\varrho_j, \delta)$ can be written in the form
\[
\wt_{\A}(\varrho_j, \delta) = a_{j,1} \cdot d_1 + {\dots} + a_{j, n} \cdot d_n + b_j
\]
for some $a_{j,1}, \dots, a_{j, n}, b_j \in \RR$. Thus, $\beh{\A}(\bot_{\delta}) = \max\{\wt_{\A}(\varrho_j, \delta) \mid j \in [k]\}$.
For each $j \in [k]$, let 
\[
\label{EQ:delta_j}
\Delta_j = \{\delta \in \DD^n  \mid  \beh{\A}(\bot_{\delta}) = \wt_{\A}(\varrho_j, \delta)\}.
\]
Since $\beh{\A}(\bot_{\delta}) \neq \zero$ for all $\delta \in \DD^n$, we have the coverage equation:
\begin{equation}
\label{EQ:coverage_delta_j}
\DD^n = \bigcup_{j \in [k]} \Delta_j.
\end{equation}
Furthermore, for every $\delta \in \DD^n$ and $j \in [k]$, we have $\delta \in \Delta_j$ if and only if 
\begin{equation}
\label{EQ:primitive_wffa_cube_constraint}
\wt_{\A}(\varrho_{j'}, \delta) \le \wt_{\A}(\varrho_j, \delta)
\end{equation}
for all $j' \in [k] \setminus \{j\}$. Note that each constraint (\ref{EQ:primitive_wffa_cube_constraint}) is a linear inequality. Therefore, $\Delta_j$ is defined by a system of linear inequalities, making it an $n$-dimensional polyhedron. By (\ref{EQ:coverage_delta_j}), there exist $j \in [k]$ and a cube $\mathcal C(\delta, R)$, for some $\delta \in \DD^n$ and $R \in \RR_{> 0}$, such that $\mathcal C(\delta, R) \subseteq \Delta_j$. Therefore, $\beh{\A}(\bot_{\delta'}) = \wt_{\A}(\varrho_j, \delta')$ for all $\delta' \in \mathcal C(\delta, R)$.
\end{proof}

\begin{lemma}[Lemma~\ref{LEMMA:primitive_counterexample}]
Let $\Sigma = \{\bot\}$ and $\FF = \FF_{\arc, \times}$. Let $\F: \DD \Sigma^* \to S$ be the QFL defined by  
$
\F(\bot_{d_1, \dots, d_n}) = n + \sum_{i=1}^n d_i
$
for all $n \in \NN$ and $d_1, \dots, d_n \in \DD$.
Then, $\F \in \beh{\WFFA_{\Sigma, \FF}(\PP_{\FF})} \setminus \beh{\WFFA_{\Sigma, \FF}(\PP_{\FF})}$.
\end{lemma}

\begin{proof}
Let $\A_{\F} \in \WFFA_{\Sigma, \FF}(\AA_{\FF})$ be the automaton depicted in Figure \ref{FIG:primitive_counterexample}, where the transition labels $\bot$ are omitted. It is straightforward to verify that $\beh{\A_{\F}} = \F$.

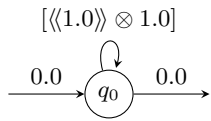
\begin{figure}[ht]
\centering
\begin{tikzpicture}[
    >=stealth,
    every node/.style={font=\footnotesize},
    state/.style={circle, draw, minimum size=18pt, inner sep=1pt}
]
\node (q0) {};
\node[state, right=1cm of q0] (q1) {$q_0$};
\node[right=1cm of q1] (q11) {};

\draw[->] 
	(q0) edge[above] node{$0.0$} (q1)
	(q1) edge[loop above] node{$[\llangle 1.0 \rrangle \otimes 1.0]$} (q1)
	(q1) edge[above] node{$0.0$} (q11);
\end{tikzpicture}
\caption{WFFA $\A_{\F}$ from the proof of Lemma \ref{LEMMA:primitive_counterexample}.}
\label{FIG:primitive_counterexample}
\end{figure}

Next, we show that $\F \notin \beh{\WFFA_{\Sigma, \FF}(\PP_{\FF})}$. Suppose, for the sake of contradiction, that there exists an automaton
\[
\A =  (Q, I, T, F, \wt_I, \wt_T, \wt_F) \in \WFFA_{\Sigma, \FF}(\PP_{\FF})
\]
with atomic transition weights such that $\beh{\A} = \F$. By Lemma~\ref{LEMMA:primitive_wffa_cube_property}, for every $n \in \NN_{\ge 1}$, there exist a run $\varrho \in \Run_{\A}^n$, a run context $\delta \in \DD^n$, and a positive constant $R \in \RR$ such that 
\begin{equation}
\label{EQ:primitive_counterexample_run_eq_behavior}
\beh{\A}(\bot_{\delta'}) = \wt_{\A}(\varrho, \delta')
\end{equation}
for all $\delta' \in \mathcal C(\delta, R)$. Let $t_1 \dots t_n \in T^+$ denote the transition sequence of $\varrho$. For each $i \in [n]$, let $e_i^n = \wt_T(t_i)$ and $t_i = (q_i, \bot, q_i')$. Furthermore, let $w_I^n = \wt_I(q_1)$ and $\wt_F^n = \wt_F(q_n')$. We distinguish between the following subcases:
\begin{itemize}
\item {\em All transition weights are in $\llangle{S}\rrangle$}. Suppose that, for all $i \in [n]$, $e_i^n = \bind{s_i^n}$ with $s_i^n \in S$. Let $\delta' = (d_1, \dots, d_n) \in \C(\delta, R)$. Then, by Equation~\eqref{EQ:primitive_counterexample_run_eq_behavior},
\begin{equation}
\label{EQ:primitive_counterexample_all_binding_cube_start}
n + \sum_{i \in [n]}d_i = w_I^n + \sum_{i \in [n]} s_i^n \cdot d_i + w_F^n
\end{equation}
For any $r_1, \dots, r_n \in [0, R]$, define $\delta' = (d_1 + r_1, \dots, d_n + r_n) \in \C(\delta, R)$. By the same equation, we have
\begin{equation}
\label{EQ:primitive_counterexample_all_binding_cube_end}
n + \sum_{i \in [n]} (d_i + r_i) = w_I^n + \sum_{i \in [n]} s_i^n \cdot (d_i + r_i) + w_F^n.
\end{equation}
Subtracting \eqref{EQ:primitive_counterexample_all_binding_cube_start} from \eqref{EQ:primitive_counterexample_all_binding_cube_end} yields
\[
\sum_{i \in [n]} (1 - s_i^n) \cdot r_i = 0.
\]
Since this must hold for all choices of $r_1, \dots, r_n \in [0, R]$, we conclude that $s_i^n = 1$ for all $i \in [n]$. Substituting back into \eqref{EQ:primitive_counterexample_all_binding_cube_start} gives
\[
n = w_{I}^n + w_{F}^n.
\]
However, by definition, $w_I^n + w_F^n \le \max \wt_I(I) + \max \wt_F(F)$,
while $n \in \NN_{\ge 1}$ can be chosen arbitrarily large. This leads to a contradiction.
\item {\em Mixed transition weights}. Suppose that $\mathcal I = \{i \in [n] \mid e_i^n \in S \} \neq \emptyset$. For all $i \in \mathcal I$, let $e_i^n = s_i^n$, and for all $i \in [n] \setminus \mathcal I$, let $e_i^n = \bind{s_i^n}$. Then, by Equation~\eqref{EQ:primitive_counterexample_run_eq_behavior}, for $\delta' = (d_1, \dots, d_n) \in \C(\delta, \R)$, we have
\begin{equation}
\label{EQ:primitive_counterexample_mixed_cube_start}
n + \sum_{i \in [n]} d_i = w_I^n + \sum_{i \in \mathcal I} s_i^n + \sum_{i \in [n] \setminus \mathcal I} s_i^n \cdot d_i + w_F^n
\end{equation}
For any $r_1, \dots, r_n \in [0, R]$, define $\delta' = (d_1 + r_1, \dots, d_n + r_n) \in \C(\delta, R)$. By the same equation,
\begin{equation}
\label{EQ:primitive_counterexample_mixed_cube_end}
n + \sum_{i \in [n]} (d_i + r_i) = w_I^n + \sum_{i \in \mathcal I} s_i^n + \sum_{i \in [n] \setminus \mathcal I} s_i^n \cdot (d_i + r_i) + w_F^n.
\end{equation}
Subtracting \eqref{EQ:primitive_counterexample_mixed_cube_start} from \eqref{EQ:primitive_counterexample_mixed_cube_end} gives
\begin{equation}
\label{EX:primitive_counterexample_non_trivial_solution}
\sum_{i \in \mathcal I} r_i = \sum_{i \in [n] \setminus \mathcal I} s_i^n \cdot r_i
\end{equation}
This equality must hold for all $r_1, \dots, r_n \in [0, R]$. But if we set $r_i = 0$ for all $i \in [n] \setminus \mathcal I$ and $r_i = R$ for all $i \in \mathcal I$, we obtain $R \cdot |\mathcal I| = 0$, which implies $R = 0$. This contradicts the assumption that $R$ is positive.
\end{itemize}
Since all possible subcases result in contradictions, it follows that no WFFA with atomic transition weights can realize $\F$. Therefore, $\F \notin \beh{\WFFA_{\Sigma, \FF}}(\PP_{\FF})$.
\end{proof}

\begin{theorem}[Theorem~\ref{THM:closure_counterexample_primitive}]
Let $\Sigma = \{\bot\}$ and $\FF = \FF_{\arc, \times}$. Then, the following hold:
\begin{enumerate}[label=(\alph*)]
\item There exist $\F_1, \F_2 \in \beh{\WFFA_{\Sigma, \FF}(\PP_{\FF})}$ such that $(\F_1 \otimes \F_2) \notin \beh{\WFFA_{\Sigma, \FF}(\PP_{\FF})}$.
\item There exists $\F \in \beh{\WFFA_{\Sigma, \FF}(\PP_{\FF})}$ such that $\F$ is proper and ${\F^* \notin \beh{\WFFA_{\Sigma, \FF}(\PP_{\FF})}}$.
\end{enumerate}
\end{theorem}

\begin{proof}
\begin{enumerate}[label=(\alph*)]
\item Let $\F_1: \DD \Sigma^* \to S$ be defined for all $d_1, \dots, d_n \in \DD$ by 
\[
\F_1(\bot_{d_1, \dots, d_n}) = d_1 + {\dots} + d_n.
\]
and let $\F_2: \DD \Sigma^* \to S$ be defined for all $w \in \DD \Sigma^*$ by $\F_2(w) = |w|$.
Consider the WFFAs $\A_{\F_1} \in \WFFA_{\Sigma, \FF}(\PP_{\FF})$ and $\A_{\F_2} \in \WFFA_{\Sigma, \FF}(S)$ depicted in Figure \ref{FIG:thm_closure_counterexample_atomic_proof_a}. It is straightforward to check that $\beh{\A_{\F_1}} = \F_1$ and $\beh{\A_{\F_1}} = \F_2$. Let $\F = \F_1 \otimes \F_2$. Then, by Lemma~\ref{LEMMA:primitive_counterexample}, $\F \notin \beh{\WFFA_{\Sigma, \FF}(\PP_{\FF})}$.

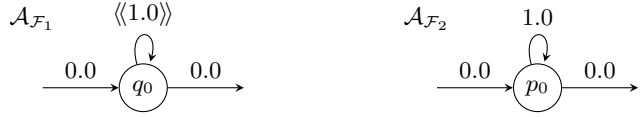
\begin{figure}[ht]
\centering
\begin{tikzpicture}[
    >=stealth,
    every node/.style={font=\footnotesize},
    state/.style={circle, draw, minimum size=18pt, inner sep=1pt}
]
\node (q0) {};
\node[state, right=1cm of q0] (q1) {$q_0$};
\node[right=1cm of q1] (q11) {};
\node[above=0.5cm of q0] {$\A_{\F_1}$};

\node[right=2cm of q11] (p0) {};
\node[state, right=1cm of p0] (p1) {$p_0$};
\node[right=1cm of p1] (p11) {};
\node[above=0.5cm of p0] {$\A_{\F_2}$};

\draw[->] 
	(q0) edge[above] node{$0.0$} (q1)
	(q1) edge[loop above] node{$\llangle 1.0 \rrangle$} (q1)
	(q1) edge[above] node{$0.0$} (q11)
	(p0) edge[above] node{$0.0$} (p1)
	(p1) edge[loop above] node{$1.0$} (p1)
	(p1) edge[above] node{$0.0$} (p11);
\end{tikzpicture}
\caption{WFFAs $\A_{\F_1}$ and $\A_{\F_2}$ from the proof of Theorem~\ref{THM:closure_counterexample_primitive}(a).}
\label{FIG:thm_closure_counterexample_atomic_proof_a}
\end{figure}
\item Let $\F: \DD \Sigma^* \to S$ be defined by $\F(\bot_{d}) = d + 1$ for all $d \in \DD$ and $\F(w) = \zero$ for all $w \in \DD \Sigma^*$ with $|w| \neq 1$. Clearly, $\F$ is proper. Consider the automaton $\A_{\F} \in \WFFA_{\Sigma, \FF}(\bind{S}_{\FF})$ depicted in Figure \ref{FIG:thm_closure_counterexample_atomic_proof_b}. It is straightforward to verify that, $\beh{\A_{\F}} = \F$. 

However, for the Kleene star $\F^*$ of $\F$, we have 
\[
\F^*(\bot_{d_1, \dots, d_n}) = n + d_1 + {\dots} + d_n
\]
for all $d_1, \dots, d_n \in \DD$. By Lemma \ref{LEMMA:primitive_counterexample}, $\F \notin \beh{\WFFA_{\Sigma, \FF}(\PP_{\FF})}$.

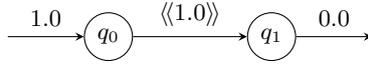
\begin{figure}[ht]
\centering
\begin{tikzpicture}[
    >=stealth,
    every node/.style={font=\footnotesize},
    state/.style={circle, draw, minimum size=18pt, inner sep=1pt}
]
\node (q0) {};
\node[state, right=1cm of q0] (q1) {$q_0$};
\node[state, right=1.5cm of q1] (q2) {$q_1$};
\node[right=1cm of q2] (q22) {};

\draw[->] 
	(q0) edge[above] node{$1.0$} (q1)
	(q1) edge[above] node{$\llangle 1.0 \rrangle$} (q2)
	(q2) edge[above] node{$0.0$} (q22);
\end{tikzpicture}
\caption{WFFA $\A_{\F}$ from the proof of Theorem~\ref{THM:closure_counterexample_primitive}(b).}
\label{FIG:thm_closure_counterexample_atomic_proof_b}
\end{figure}
\qedhere
\end{enumerate}
\end{proof}

\begin{lemma}[Lemma~\ref{LEMMA:bounded_sum_property_ff_arc_times_positive}]
Let $\FF$ be a finance semiring with an additive data-binding function. Then, $\PP_{\FF}$ satisfies the bounded sum property and $B(\PP_{\FF}) \le 2$.
\end{lemma}

\begin{proof}
For any $k \in \NN_{\ge 1}$, let $\pi_1, \dots, \pi_k \in \PP_{\FF}$ and define $e = [\pi_1 \oplus {\dots} \oplus \pi_k]$. Since $\oplus$ is commutative, we may assume that $\pi_1, \dots, \pi_l \in S$ for some $l \le k$ and $\pi_{l + 1}, \dots, \pi_{k} \in \bind{S}_{\FF}$. Then
\[
e = [s_1 \oplus {\dots} \oplus s_l \oplus \bind{s_{l + 1}} \oplus {\dots} \oplus \bind{s_k}]
\]
where $s_1, \dots, s_k \in S$. Consider the following subcases:
\begin{itemize}
\item {\em All summands in $S$} ($l = k$): $e = [s_1 \oplus {\dots} \oplus s_k]$. Let $s' = (s_1 \oplus {\dots} \oplus s_k) \in S \subseteq \PP_{\FF}$. Then, $\beh{e} = \beh{s'}$.
\item {\em All summands in $\bind{S}_{\FF}$} ($l = 0$): $e = [\bind{s_1} \oplus {\dots} \oplus \bind{s_k}]$. Let $s' = (s_1 \oplus {\dots} \oplus s_k) \in S$. By additivity of $\bb$, we have $\beh{e} = \beh{\bind{s'}}$, with $\bind{s'} \in \PP_{\FF}$.
\item {\em Mixed summands} ($0 < l < k$): Let $s' = (s_1 \oplus {\dots} \oplus s_l) \in S$ and $s'' = (s_{l + 1} \oplus {\dots} \oplus s_k) \in S$. By additivity of $\bb$, we obtain $\beh{e} = \beh{s' \oplus \bind{s''}}$, where $s', \bind{s''} \in \PP_{\FF}$.
\end{itemize}
Hence, $\PP_{\FF}$ satisfies the bounded-sum property, with $B(\PP_{\FF}) \le 2$.
\end{proof}

\begin{lemma}[Lemma~\ref{LEMMA:bounded_sum_property_ff_arc_times_rr}]
Let $\FF = \FF_{\arc, \times}^{\RR}$. Then $\PP_{\FF}$ satisfies the bounded-sum property with $B(\PP_{\FF}) \le 5$.
\end{lemma}

\begin{proof}
Let $k \in \NN_{\ge 1}$ and $e = [\pi_1 \oplus \dots \oplus \pi_k]$ with $\pi_1, \dots, \pi_k \in \PP_{\FF}$.

We distinguish the following subcases:

\begin{enumerate}[label=(\alph*)]
\item All $\pi_i \in \bind{S}_{\FF}$. Write $\pi_i = \bind{s_i}$ with $s_i \in S$ (without loss of generality, $s_i \neq -\infty$, i.e., $s_i \in \RR$).  
Consider the following subcases.
\begin{itemize}
\item all $s_1, \dots, s_k \ge 0$ or all $s_1, \dots, s_k \le 0$. Let  
\[
\sigma_1 = \bigoplus_{i=1}^k s_i, \quad \sigma_2 = -\bigoplus_{i=1}^k (-s_i).
\]  
Then $\beh{e} = \beh{\bind{\sigma_1} \oplus \bind{\sigma_2}}.$
\item $s_1, \dots, s_l \ge 0$ and $s_{l+1}, \dots, s_k < 0$ for some $0 < l < k$. Define  
\[
e_{\ge 0} = \left[\bigoplus_{i=1}^l \bind{\pi_i} \right], \quad e_{<0} = \left[\bigoplus_{i=l+1}^k \bind{\pi_i}\right].
\]  
By the previous subcase, there exist $\sigma_1, \sigma_2 \in \RR$ such that 
$
\beh{e_{\ge 0}} = \beh{\bind{\sigma_1} \oplus \bind{\sigma_2}},
$ 
and $\sigma_3, \sigma_4 \in \RR$ such that 
$
\beh{e_{<0}} = \beh{\bind{\sigma_3} \oplus \bind{\sigma_4}}.
$  
Thus,  
$
\beh{e} = \beh{\bigoplus_{i=1}^4 \bind{\sigma_i}}.
$
\end{itemize}

\item All $\pi_i \in S$. Then  
\[
\sigma = \bigoplus_{i=1}^k \pi_i \in S, \quad \text{and hence } \beh{e} = \beh{\sigma}.
\]

\item $\pi_1, \dots, \pi_l \in S$ and $\pi_{l+1}, \dots, \pi_k \in \bind{S}_{\FF}$ for some $0 < l < k$. Define  
\[
e_1 = [\bigoplus_{i=1}^l \pi_i], \quad e_2 = [\bigoplus_{i=l+1}^k \pi_i].
\]  
By (b), there exists $\sigma_1 \in S$ such that $\beh{e_1} = \beh{\sigma_1}.$
By (a), there exist $\sigma_2, \dots, \sigma_K$ with $K \le 4$ such that 
$\beh{e_2} = \beh{\bigoplus_{i=2}^K \bind{\sigma_i}}.$
Hence,  
\[
\beh{e} = \beh{\sigma_1 \oplus \bind{\sigma_2} \oplus \dots \oplus \bind{\sigma_K}}.
\]
\end{enumerate}

In all subcases, there exist $1 \le B \le 5$ and $\pi_1', \dots, \pi_B' \in \PP_{\FF}$ such that  
\[
\beh{e} = \beh{\pi_1' \oplus \dots \oplus \pi_B'}.
\]
\end{proof}

\begin{lemma}[Lemma~\ref{LEMMA:wffa_constructions_with_bounded_sum}]
Assume that $\PP_{\FF}$ satisfies the bounded-sum property.
\begin{enumerate}[label=(\alph*)]
\item Let $\A_1, \A_2 \in \WFFA^{\tau}_{\Sigma, \FF}(\PP_{\FF})$. Then, there exists $\A' \in \WFFA_{\Sigma, \FF}^{\tau}(\PP_{\FF})$ such that ${\beh{\A'} = \beh{\A_1} \cdot \beh{\A_2}}$, and $|\A'|_Q = \O(|\A_1|_Q + |\A_2|_Q)$.
\item Let $\A \in \WFFA^{\tau}_{\Sigma, \FF}(\PP_{\FF})$ such that $\beh{\A}$ is proper. Then, there exists $\A^* \in \WFFA_{\Sigma, \FF}^{\tau}(\PP_{\FF})$ such that $\beh{\A^*} = \beh{\A}^*$ and $|\A^*|_Q = \O(|\A|_Q)$.
\end{enumerate}
\end{lemma}

\begin{proof}
\begin{enumerate}[label=(\alph*)]
\item 
First, using the construction of the proof of Corollary \ref{COR:cauchy_product_kleene_star_full_wffa} for the Cauchy product, we build $\hat{\A} \in \WFFA^{\tau}_{\Sigma, \FF}(\EE_{\FF})$ with $|\hat{\A}|_Q = |\A_1|_Q + |\A_2|_Q + \O(1)$ such that $\beh{\hat{\A}} = \beh{\A_1} \cdot \beh{\A_2}$. Note that, since $\A_1$ and $\A_2$ are purely transition-weighted, we can assume that the transition weights of $\hat{\A}$ are $\FF$-expressions of the form $e = [\pi_1 \oplus {\dots} \oplus \pi_k]$ where $\pi_1, \dots, \pi_k \in \PP_{\FF}$, i.e., $e \in \clp{\PP_{\FF}}$. Since $\PP_{\FF}$ satisfies the bounded-sum property, every such $e$ can be replaced with an expression of the form $e' = [\pi_1' \oplus {\dots} \oplus \pi'_b]$ with $b \le B(\PP_{\FF})$, i.e., $|e'| \le B(\PP_{\FF})$. After all these expression adjustments we obtain $\hat{\A}_1 \in \WFFA_{\Sigma, \FF}^{\tau}(\clp{\PP_{\FF}})$ with $\beh{\hat{\A}_1} = \beh{\hat{\A}}$, $|\hat{\A}_1|_Q = |\hat{\A}|_Q$ and $|\hat{\A}_1|_{\EE_{\FF}} \le B(\PP_{\FF})$. Then, using Lemma~\ref{LEMMA:plus_closed_expressiveness}, we can construct $\A' \in \WFFA_{\Sigma, \FF}^{\tau}(\PP_{\FF})$ with $\beh{\A'} = \beh{\hat{\A}_1}$. It follows from the proof of Lemma~\ref{LEMMA:plus_closed_expressiveness} that $\A'$ has $\O(|\hat{\A}_1|_Q \cdot B(\PP_{\FF}))$ states. Since $B(\PP_{\FF})$ is a constant, we obtain that $\A'$ has $\O(|\hat{\A}_1|_Q) = \O(|\A_1|_Q + |\A_2|_Q)$ states.
\item The proof is similar to (a). Again, using the construction of Corollary \ref{COR:cauchy_product_kleene_star_full_wffa} for the Kleene star, we build $\hat{\A} \in \WFFA_{\Sigma, \FF}^{\tau}(\EE_{\FF})$ such that $\beh{\hat{\A}} = \beh{\A}^*$ and $|\hat{\A}|_Q = |\A|_Q + \O(1)$. Again, we notice that $\hat{\A} \in \WFFA_{\Sigma, \FF}^{\tau}(\clp{\PP_{\FF}})$. The rest of the proof proceeds using the same argument as in the proof of (a).
\qedhere
\end{enumerate}
\end{proof}

\subsubsection{Affine Expressions (Subsection \ref{SUBSEC:affine_expressions})}

\begin{lemma}[Lemma~\ref{LEMMA:wffa_with_affine_expressions_qfls_inclusion}]
Let $\FF = \FF_{\arc, \times}$ and $\Sigma = \{\bot\}$. Then,
\[
\beh{\WFFA_{\Sigma, \FF}(\PP_{\FF})} \subsetneq \beh{\WFFA_{\Sigma, \FF}(\PP_{\FF})} \subsetneq \beh{\WFFA_{\Sigma, \FF}(\EE_{\FF})}.
\]
\end{lemma}

\begin{proof}
Clearly, $\beh{\WFFA_{\Sigma, \FF}(\PP_{\FF})} \subseteq \beh{\WFFA_{\Sigma, \FF}(\AA_{\FF})} \subseteq \beh{\WFFA_{\Sigma, \FF}(\EE_{\FF})}$.

\begin{itemize}
\item To demonstrate that $\beh{\WFFA_{\Sigma, \FF}(\PP_{\FF})} \neq \beh{\WFFA_{\Sigma, \FF}(\AA_{\FF})}$, consider the QFL $\F: \DD \Sigma^* \to S$ with $\F(\bot_{d_1, \dots, d_n}) = n + d_1 + \dots + d_n$. By Lemma~\ref{LEMMA:primitive_counterexample}, $\F \notin \beh{\WFFA_{\Sigma, \FF}(\PP_{\FF})}$. Moreover, the WFFA $\A_{\F}$ from the proof of Lemma~\ref{LEMMA:primitive_counterexample} is in $\WFFA_{\Sigma, \FF}(\AA_{\FF})$, which implies $\F \in \beh{\WFFA_{\Sigma, \FF}(\AA_{\FF})}$.
\item Now we show that $\beh{\WFFA_{\Sigma, \FF}(\AA_{\FF})} \neq \beh{\WFFA_{\Sigma, \FF}(\EE_{\FF})}$. Let $\F': \DD \Sigma^* \to S$ be defined, e.g., by $\F'(\bot_{0.5}) = 1$ and $\F'(w) = \zero$ for all $w \in \DD \Sigma^* \setminus \{\bot_{0.5}\}$ (recall that $\bot_{0.5} = (\bot, 0.5)$). Note that $\F'$ is recognizable by $\A_{\F'} \in \WFFA_{\Sigma, \FF}(\EE_{\FF})$ defined as
\[
\A_{\F'} = (\{q_1, q_2\}, \{q_1\}, \{(q_1, \bot, q_2)\}, \{q_2\}, [q_1 \mapsto 0], \wt_T, [q_2 \mapsto 0])
\]
where $\wt_T(q_1, \bot, q_2) = [\bind{1.0} = 0.5]$. Suppose that there exists $\A \in \WFFA_{\Sigma, \FF}(\AA_{\FF})$ such that $\beh{\A} = \F'$. Then, there exist expressions $e_1, \dots, e_k \in \AA_{\FF}$ and constants $c_1, \dots, c_k \in S$ such that $\beh{\A}(w_d) = \max_{i \in [k]} (c_i + \beh{e_i}(d))$ for all $d \in \DD$. Without loss of generality, assume that, for all $i \in [k]$, $e_i = [\bind{a_i} \otimes b_i]$ and $\beh{\A}(\bot_{0.5}) = a_j + \beh{e_j}(0.5)$ for some $j \in [k]$. Then, $\beh{\A}(\bot_{0.5}) = a_j \cdot 0.5 + (b_j + c_j) = 1$, which implies $a_j, b_j, c_j \in \RR$. Then, for all $d \in \DD \setminus \{0.5\}$, $-\infty = \beh{\A}(w_d) \ge (a_j \cdot d + (b_j + c_j)) \in \RR$. This is a contradiction. Thus, $\F' \notin \WFFA_{\Sigma, \FF}(\AA_{\FF})$.
\end{itemize}
\end{proof}

\begin{lemma}[Lemma~\ref{LEMMA:bounded_sum_for_affine_expressions}]
Let $\FF = \FF_{\arc, \times}$. Then $\AA_{\FF}$ does not satisfy the bounded-sum property.
\end{lemma}

\begin{proof}
We show that, for every $N \in \NN_{\ge 1}$, there exist affine $\FF$-expressions $\alpha_1, \dots, \alpha_N \in \AA_{\FF}$ such that, for all $n \in [N - 1]$ and $\alpha'_1, \dots, \alpha'_n \in \AA_{\FF}$, $\beh{\alpha_1 \oplus {\dots} \oplus \alpha_N} \neq \beh{\alpha'_1 \oplus {\dots} \oplus \alpha'_n}$.

For every $i \in \NN_{\ge 1}$, let $\alpha_i = [\bind{2 \cdot i} \otimes (-i^2)]$. Let $N \in \NN_{\ge 1}$ and $A_N = [\alpha_1 \oplus {\dots} \oplus \alpha_N]$. Then, for every $d \in \DD$, 
\begin{equation}
\label{EQ:affine_expression_lambda_n_semantics}
\beh{A_N}(d) = \begin{cases}
\beh{\alpha_1}(d), & \text{if } d \in [0, \frac{3}{2}], \\
\beh{\alpha_i}(d), & \text{if } d \in [i-\frac{1}{2}, i+\frac{1}{2}] \text{ for } i \in \{2, \dots, N - 1\}, \\
\beh{\alpha_N}(d), & \text{if } d \in [N - \frac{1}{2}, \infty).
\end{cases}
\end{equation}
Suppose that there exist $n \in [N - 1]$ and $\alpha_1', \dots, \alpha_n' \in \AA_{\FF}$ such that $\beh{A'_n} = \beh{A_{N}}$ where $A'_n = \alpha_1' \oplus {\dots} \oplus \alpha_n'$. Assume that, for all $i \in [n]$, $\alpha_i' = [\bind{a_i} \otimes b_i]$ where $a_i, b_i \in \RR$ (if $a_i = -\infty$ or $b_i = -\infty$, then $\beh{\alpha_i} = \beh{-\infty}$ and therefore $\alpha'_i$ can be excluded from $A'_n$).

For all $i \in [n]$ and $d \in \DD$, $\beh{A'_n}(d) = \beh{\alpha'_i}(d)$ if and only if $a_i \cdot d + b_i \ge a_j \cdot d + b_j$ for all $j \in [n]$. These constraints can be viewed as a system of $n$ linear inequalities in $d$. The solution to this system is either the empty set, an interval $[x, y]$ with $x, y \in \DD$ and $x \le y$, or an interval $[x, \infty)$ for some $x \in \DD$. Since $\beh{A'_n}$ is a continuous function, in the case of the empty set or a single-point interval $[x, x]$, $\alpha_i$ does not contribute to the value of $\beh{A_n'}$ and can therefore be eliminated from $A'_n$. Thus, without loss of generality, we can assume that $\beh{A'_n}(d) = \beh{\alpha'_i}(d)$ if and only if $d \in I_i$, where either $I_i = [x_i, y_i]$ for some $x_i, y_i \in \DD$ with $x_i < y_i$ or $I_i = [x_i, \infty)$ for some $x_i \in \DD$.
Furthermore, $\bigcup_{i \in [n]} I_i = \DD$. Note also that, for every $j \in [N]$, there exists $i_j \in [n]$ such that $\beh{\alpha_j}(d) = \beh{\alpha'_{i_j}}(d)$ for all $d$ belonging to some non-empty interval $(\hat{x}_{i_j}, \hat{y}_{i_j}) \subseteq I_{i_j}$. This implies that $a_{i_j} = 2 \cdot j$ and $b_{i_j} = -j^2$. Thus, the set $\{\alpha_1, \dots, \alpha_n\}$ must contain at least $N$ distinct elements. A contradiction.
\end{proof}

\begin{lemma}[Lemma~\ref{LEMMA:affine_expressions_blowup}]
Let $\FF = \FF_{\arc, \times}$ and $\Sigma = \{\bot\}$.
For every $i \in \NN_{\ge 1}$, let $\varphi_i: \DD \to \RR$ be defined by $\varphi_i(x) = 2 \cdot i \cdot x - i^2$ for all $x \in \DD$. For every $N \in \NN_{\ge 1}$, let $f_N: \DD \to R$ be defined for all $x \in \DD$ by $f_N(x) = \max_{i \in [N^2]} \varphi_i(x)$. Let $\F_N: \DD \Sigma^* \to S$ be defined for all $d \in \DD$ by $\F_N(w_d) = f_N(d)$ and, for all $w \in \DD \Sigma^*$ with $|w| \neq 1$, by $\F_N(w) = \zero$. Then, the following hold:
\begin{enumerate}[label=(\alph*)]
\item There exists $\A_N \in \WFFA^{\tau}_{\Sigma, \FF}(\AA_{\FF})$ such that $|\A_N|_Q = 2 \cdot N$ and $\beh{\A_N} = \F_N$.
\item For every $\A \in \WFFA_{\Sigma, \FF}^{\tau}(\AA_{\FF})$ with $\beh{\A} = \F_{N} \cdot \F_{N}$, we have $|\A|_Q \ge 3 \cdot N^{\frac{4}{3}}$
\end{enumerate}
\end{lemma}

\begin{proof}
\begin{enumerate}[label=(\alph*)]
\item For each $i \in \NN_{\ge 1}$, let $e_i = [\bind{2 \cdot i} \otimes (-i^2)] \in \AA_{\FF}$. Then, $\beh{e_i} = \varphi_i$.
We let 
$
\A_{N} = (Q, I, T, F, \wt_I, \wt_T, \wt_F)
$
where:
\begin{itemize}
\item $Q = \{q_1, \dots, q_{2 \cdot N}\}$, $I = \{q_1, \dots, q_N\}$, $F = \{q_{N + 1}, \dots, q_{2 \cdot N}\}$;
\item $T = \{(q_i, \bot, q_{j + N}) \mid i, j \in [N]\}$;
\item $\wt_T(q_i, \bot, q_{j + N}) = e_{N \cdot (i - 1) + j}$ for all $i, j \in [N]$;
\item $\wt_I(q) = \one$ for all $q \in I$; $\wt_F(q) = \one$ for all $q \in F$.
\end{itemize}
It is easy to see that $\beh{\A_{N}} = \F_{N}$. 
\item Let $\F = \F_{N} \cdot \F_{N}$. Then, for all $x,y \in \DD$,
\[
\F(\bot_{x,y}) = \max \{\psi_{i,j}(x, y) \mid i, j \in [N^2]\}
\]
where, for all $i, j \in [N]$, $\psi_{i,j}: \DD^2 \to \RR$ is defined by $\psi_{i,j}(x, y) = \varphi_i(x) + \varphi_{j}(y)$. Fix any $i, j \in [N^2]$. Then, for all $i', j' \in [N^2]$, such that $i \neq i'$ or $j \neq j'$, and for $x = i$ and $y = j$, it is easy to check that
\[
\psi_{i,j}(x, y) - \psi_{i',j'}(x, y) = (i - i')^2 + (j - j')^2 > 0,
\]
which implies $\F(\bot_{x, y}) = \psi_{i,j} (x, y)$ for $x = i$ and $y = j$. Since $\psi_{i, j}$ is a continuous function, there exist $r_{i,j}^x, r_{i, j}^y > 0$ such that $\F(\bot_{\tilde{x}, \tilde{y}}) = \psi_{i, j}(\tilde{x}, \tilde{y})$ for all $\tilde{x} \in (i - r_{i, j}^x, i + r_{i, j}^x)$ and all $\tilde{y} \in (j - r_{i, j}^y, j + r_{i,j}^y)$, i.e., for infinitely many $\tilde{x}$ and infinitely many $\tilde{y}$.

Let $\A = (Q, I, T, F, \wt_I, \wt_T, \wt_F) \in \WFFA_{\Sigma, \FF}(\AA_{\FF})$ such that $\beh{\A} = \F$. Without loss of generality, we can make the following assumptions:
\begin{itemize}
\item $\zero \notin \wt_I(I) \cup \wt_F(F)$ and, for every $t \in T$, $\wt_T(t) = [\bind{s_t} \otimes s'_t]$ with $s_t \neq \zero$ and $s'_t \neq \zero$;
\item for every initial state $q \in I$, there is a valid run in $\A$, which starts in $q$;
\item for every final state $q \in F$, there is a valid run in $\A$, which ends in $q$;
\item every state $q \in Q$ belongs to a valid run of $\A$.
\end{itemize}

Note that $\A$ has valid runs only on finance words $w \in \DD \Sigma^*$ with $|w| = 2$. Then, clearly, $I \cap F = \emptyset$. Let $H = Q \setminus (I \cup F)$. Using the assumptions above, we can easily check that $T \subseteq (I \times \Sigma \times H) \cup (H \times \Sigma \times F)$ (otherwise, $\A$ would have a valid run on a finance word $w$ with $|w| \neq 2$). Then, $\A$ can have at most $|I| \cdot |H| \cdot |F|$ runs. Since for every $i,j \in [N^2]$,
\[
\beh{\A}(\bot_{x, y}) = \psi_{i, j}(x, y) = 2 \cdot i \cdot x + 2 \cdot j \cdot y - i^2 - j^2
\]
for infinitely many $x \in \DD$ and infinitely many $y \in \DD$, for every $i, j \in [N^2]$ there exists a run $\varrho = t_1 t_2$ where $t_1, t_2 \in T$, $s_{t_1} = 2 \cdot i$ and $s_{t_2} = 2 \cdot j$. Thus, $\A$ has at least $N^4$ runs and therefore
\[
N^4 \le |I| \cdot |H| \cdot |F|.
\]
Then, by the arithmetic-geometric mean inequality,
\[
|\A|_Q = |I| + |H| + |F| \ge 3 \cdot (|I| \cdot |H| \cdot |F|)^{\frac{1}{3}} \ge 3 \cdot N^{\frac{4}{3}}.
\]
\qedhere
\end{enumerate}
\end{proof}

\subsubsection{Monomials (Subsection~\ref{SUBSEC:monomials})}

For any subset $X \subseteq \DD$ and $c = \langle c_1, c_2 \rangle \in S^2$, let $\alpha^{c}_X: \DD \to S$ be defined for all $d \in \DD$ by 
\[
\alpha^{c}_X(d) =
\begin{cases}
	c_1 \cdot d + c_2, & \text{if } d \in X, \\
	\zero, & \text{otherwise}.
\end{cases}
\]
Here, we let $c_1 \cdot d + c_2 = \zero$ whenever $c_1 = \zero$ or $c_2 = \zero$. For $x \in \DD$ and $y \in \DD \cup \{\infty\}$, let $(x, y) = \{z \in \DD  \mid  x < z < y\}$.
For $k \in \NN$ and $\delta = \{d_1, \dots, d_k\} \subseteq \DD$ with $0 < d_1 < {\dots} < d_k$, let 
\[
\P(\delta) = \{ \{0\}, (0, d_1), \{d_1\}, (d_1, d_2), \dots, (d_{k-1}, d_k), \{d_k\}, (d_k, \infty)\} \subseteq 2^{\DD},
\]
the partition of $\DD$ defined by $\delta$. Note that $\P(\emptyset) = \{\{0\}, (0, \infty)\}$.
For $f_1, f_2: \DD \to S$, let $(f_1 \oplus f_2): \DD \to S$ be defined for all $x \in \DD$ by $(f_1 \oplus f_2)(x) = f_1(x) \oplus f_2(x)$.

Let $k \in \NN_{\ge 1}$. A function $f: \DD \to S$ is called {\em $k$-piecewise affine} if there exist ${\delta = \{d_1, \dots, d_{k-1}\} \subseteq \DD}$ with $0 < d_1 < {\dots} < d_{k - 1}$ and a coefficients mapping $c: \P(\delta) \to S^2$ such that 
$
f = \bigoplus (\alpha^{c(I)}_I  \mid  I \in \P(\delta)).
$
We say that $f$ is {\em piecewise affine} if $f$ is $k$-piecewise affine for some $k \in \NN_{\ge 1}$. Let ${\dist{f} = \min \{k \in \NN_{\ge 1}  \mid  f \text{ is } k\text{-piecewise affine}\}}$. If $f$ is not piecewise affine, then $\dist{f} = \infty$.

\begin{example}
\begin{enumerate}[label=(\alph*)]
\item Let $f: \DD \to S$ be defined for all $x \in \DD$ by $f(x) = 2 \cdot x + 1$. Then, $f$ is $1$-piecewise affine as it can be represented as $f = \alpha_{\{0\}}^{\langle 2, 1 \rangle} \oplus \alpha_{(0, \infty)}^{\langle 2, 1 \rangle}$. Clearly, $\dist{f} = 1$.
\item Let $f: \DD \to S$ be defined for all $x \in \DD$ by
\[
f(x) = \begin{cases}
2 \cdot x + 1, & \text{if } 0 < x \le 3, \\
-4 \cdot x + 6, & \text{if } x \ge 5, \\
-\infty, & \text{otherwise.}
\end{cases}
\]
Then, $f$ is 3-piecewise affine since
\[
f = \alpha_{\{0\}}^{\langle 0, -\infty \rangle} \oplus \alpha_{(0, 3)}^{\langle 2, 1 \rangle} \oplus \alpha_{\{3\}}^{\langle 2, 1 \rangle} \oplus \alpha_{(3,5)}^{\langle 0, -\infty \rangle} \oplus \alpha_{\{5\}}^{\langle -4, 6 \rangle} \oplus \alpha_{(5, \infty)}^{\langle -4, 6 \rangle}.
\]
Moreover, $\dist{f} = 3$.
\end{enumerate}
\end{example}

\begin{remark}
\label{REMARK:piecewise_affine_extension}
Let $\delta' \subseteq \DD$ be any set such that $\delta \subseteq \delta'$. Then, $f$ can be represented as 
$
f = \bigoplus \big(\alpha_I^{c'(I)}  \mid  I \in \P(\delta') \big)
$
where $c': \P(\delta') \to S^2$ is defined as follows. Let $I' \in \P(\delta')$. Then, there exists a unique $I \in \P(\delta)$ such that $I' \subseteq I$. Then, we let $c'(I') = c(I)$.
\end{remark}

\begin{lemma}[Lemma~\ref{LEMMA:finarc_exp_piecewise_affine}]
For every $e \in \EE_{\FF}$, the function $\beh{e} : \DD \to S$ is piecewise affine.
\end{lemma}

\begin{proof}
We show that there exist a finite set $\delta_e \subseteq \DD$ and a coefficients mapping $c_e: \P(\delta_e) \to S^2$ such that 
\begin{equation}
\label{EQ:finarc_exp_piecewise_affine}
\beh{e} = \bigoplus \big(\alpha_I^{c_e(I)} \mid  I \in \P(\delta_e) \big)
\end{equation}
for all $d \in \DD$.
We proceed by induction on the structure of $e$.
\begin{itemize}
\item Let $e = s \in S$. Then, we let $\delta_e = \emptyset$ and $c_e(\{0\}) = c_e((0, \infty)) = \langle 0.0, s\rangle$.
\item Let $e = \llangle s \rrangle$ with $s \in S$. Then, we let $\delta_e = \emptyset$ and $c_e(\{0\}) = c_e((0, \infty)) = \langle s, 0.0 \rangle$.
\item Let $e = e_1 \otimes e_2$ where $e_1, e_2 \in \EE_{\FF}$. By induction hypothesis, for ${i \in \{1, 2\}}$, there exist a finite set $\delta_{e_i} \subseteq \DD$ and $c_{e_i}: \P(\delta_{e_i}) \to S^2$ such that ${\beh{e_i} = \bigoplus_{I \in \P(\delta_{e_i})} \alpha_I^{c_{e_i}(I)}}$. By Remark~\ref{REMARK:piecewise_affine_extension}, we may assume that $\delta_{e_1} = \delta_{e_2} = \delta$.
Let $\delta_e = \delta = \{d_1, \dots, d_k\}$ where ${0 < d_1 < {\dots} < d_k}$. It remains to define $c_e$.

Let $I \in \delta_e$ and, for $i \in \{1, 2\}$, let $c_{e_i}(I) = \langle c_i, c_i' \rangle$. Then, we put $c_e(I) = \langle c_1 \otimes c_2, c_1' \otimes c_2' \rangle$.
\item Let $e = e_1 \oplus e_2$ where $e_1, e_2 \in \EE_{\FF}$. By induction hypothesis, for ${i \in \{1, 2\}}$, there exist a finite set $\delta_{e_i} \subseteq \DD$ and $c_{e_i}: \P(\delta_{e_i}) \to S^2$ such that ${\beh{e_i} = \bigoplus_{I \in \P(\delta_{e_i})} \alpha_I^{c_{e_i}(I)}}$. By Remark~\ref{REMARK:piecewise_affine_extension}, we can assume that $\delta_{e_1} = \delta_{e_2} = \delta$. Let $\delta = \{d_1, \dots, d_k\}$ where ${0 = d_0 < d_1 < {\dots} < d_k < d_{k + 1} = \infty}$. We costruct $\delta_e$ and $c_e$ as follows. For every $j \in \{0, \dots, k\}$, we let $I_j = (d_j, d_{j + 1})$ and proceed as follows.
\begin{itemize}
\item $c_e(\{d_j\}) = \langle 0.0, \beh{e}(d_j) \rangle$; if $j \neq 0$, then we put $d_j \in \delta_e$;
\item if $\beh{e_1}(d) \le \beh{e_2}(d)$ for all $d \in I_j$, then we put $c_e(I_j) = c_{e_2}(I_j)$;
\item if $\beh{e_2}(d) > \beh{e_2}(d)$ for all $d \in I_j$, then we put $c_e(I_j) = c_{e_1}(I_j)$;
\item if $\beh{e_1}(d') = \beh{e_2}(d')$ for some $d' \in I_j$ and $\beh{e_1}(d) \neq \beh{e_2}(d)$ for all $d \in I_j \setminus \{d'\}$, then we put $d' \in \delta_e$, $c_e(\{d'\}) = \langle 0, \beh{e}(d') \rangle$ and consider the following subcases:
\begin{itemize}
\item if $\beh{e_1}(d) < \beh{e_2}(d)$ for all $d \in (d_j, d')$, then we put $c_e((d_j, d')) = c_{e_2}(I_j)$ and $c_e((d', d_{j + 1})) = c_{e_1}(I_j)$;
\item otherwise, if $\beh{e_1}(d) > \beh{e_2}(d)$ for all $d \in (d_j, d')$, we put $c_e((d_j, d')) = c_{e_1}(I_j)$ and $c_e((d', d_{j + 1})) = c_{e_2}(I_j)$;
\end{itemize}
\end{itemize}
\item Let $e = (e_1 = e_2)$ where $e_1, e_2 \in \EE_{\FF}$. By induction hypothesis, for ${i \in \{1, 2\}}$, there exist a finite set $\delta_{e_i} \subseteq \DD$ and $c_{e_i}: \P(\delta_{e_i}) \to S^2$ such that ${\beh{e_i} = \bigoplus_{I \in \P(\delta_{e_i})} \alpha_I^{c_{e_i}(I)}}$. By Remark~\ref{REMARK:piecewise_affine_extension}, we can assume that $\delta_{e_1} = \delta_{e_2} = \delta$. Let $\delta = \{d_1, \dots, d_k\}$ where ${0 = d_0 < d_1 < {\dots} < d_k < d_{k + 1} = \infty}$. We costruct $\delta_e$ and $c_e$ as follows. For every $j \in \{0, \dots, k\}$, we let $I_j = (d_j, d_{j + 1})$ and proceed as follows.
\begin{itemize}
\item $c_e(\{d_j\}) = \langle 0.0, \beh{e}(d_j) \rangle$; if $j \neq 0$, then we put $d_j \in \delta_e$;
\item if $\beh{e_1}(d) = \beh{e_2}(d)$ for all $d \in I_j$, then we put $c_e(I_j) = \langle 0, \one \rangle$;
\item if $\beh{e_1}(d) \neq \beh{e_2}(d)$ for all $d \in I_j$, then we put $c_e(I_j) = \langle 0, \zero \rangle$;
\item if $\beh{e_1}(d') = \beh{e_2}(d')$ for some $d' \in I_j$ and $\beh{e_1}(d) \neq \beh{e_2}(d)$ for all $d \in I_j \setminus \{d'\}$, then we put $d' \in \delta_e$, $c_e((d_j, d')) = c_e((d', d_{j+1})) = \langle \zero, \zero \rangle$ and $c_e(\{d'\}) = \langle \one, \one \rangle$.
\end{itemize}
\item Finally, let $e = (e_1 \neq e_2)$ where $e_1, e_2 \in \EE_{\FF}$. This case is symmetric to the previous one. By induction hypothesis, for ${i \in \{1, 2\}}$, there exist a finite set $\delta_{e_i} \subseteq \DD$ and $c_{e_i}: \P(\delta_{e_i}) \to S^2$ such that ${\beh{e_i} = \bigoplus_{I \in \P(\delta_{e_i})} \alpha_I^{c_{e_i}(I)}}$. By Remark~\ref{REMARK:piecewise_affine_extension}, we can assume that $\delta_{e_1} = \delta_{e_2} = \delta$. Let $\delta = \{d_1, \dots, d_k\}$ where ${0 = d_0 < d_1 < {\dots} < d_k < d_{k + 1} = \infty}$. We costruct $\delta_e$ and $c_e$ as follows. For every $j \in \{0, \dots, k\}$, we let $I_j = (d_j, d_{j + 1})$ and proceed as follows.
\begin{itemize}
\item $c_e(\{d_j\}) = \langle 0.0, \beh{e}(d_j) \rangle$; if $j \neq 0$, then we put $d_j \in \delta_e$;
\item if $\beh{e_1}(d) = \beh{e_2}(d)$ for all $d \in I_j$, then we put $c_e(I_j) = \langle 0, \zero \rangle$;
\item if $\beh{e_1}(d) \neq \beh{e_2}(d)$ for all $d \in I_j$, then we put $c_e(I_j) = \langle 0, \one \rangle$;
\item if $\beh{e_1}(d') = \beh{e_2}(d')$ for some $d' \in I_j$ and $\beh{e_1}(d) \neq \beh{e_2}(d)$ for all $d \in I_j \setminus \{d'\}$, then we put $d' \in \delta_e$, $c_e((d_j, d')) = c_e((d', d_{j+1})) = \langle \one, \one \rangle$ and $c_e(\{d'\}) = \langle \zero, \zero \rangle$.
\end{itemize}
\end{itemize}
\end{proof}

\begin{remark}
We can show for affine expressions $\alpha \in \clp{\AA_{\FF}}$ that $\dist{\beh{\alpha}} = \Omega(|\alpha|)$. Indeed, for every $N \in \NN_{\ge 1}$, consider the expression $A_N \in \clp{\AA_{\FF}}$ as defined in Lemma~\ref{LEMMA:bounded_sum_for_affine_expressions}. Then, $|A_N| = 4 \cdot N - 1$ and, by Equation~\eqref{EQ:affine_expression_lambda_n_semantics}, $\dist{\beh{A_N}} = N$.

However, it remains open whether $\dist{\beh{e}}$ is polynomial in $|e|$ for all $\FF$-expressions $e \in \EE_{\FF}$.
\end{remark}

\begin{lemma}[Lemma~\ref{LEMMA:polynomial_for_piecewise}]
Let $f: \DD \to S$ be a piecewise affine function. Then, there exists $e \in \clp{\MM_{\FF}}$ such that $\beh{e} = f$.
\end{lemma}

\begin{proof}
Let $f = \bigoplus (\alpha_I^{c(I)}  \mid  I \in \P(\delta))$ where $\delta \subseteq \DD$ is a finite set and $c: \P(\delta) \to S^2$.
For the proof that $f$ is definable by an $\FF$-polynomial, it suffices to show that, for every $I \in \P(\delta)$ and every $c = (c_1, c_2) \in S^2$, $\alpha_I^{c(I)}$ is definable by some $\mu \in \MM_{\FF}$. Consider the following cases.
\begin{itemize}
\item $c_1 = \zero$ or $c_2 = \zero$. Let $\mu \in \MM_{\FF}$ be defined as $(\one = \one) \otimes \zero$. Then, $\beh{\mu} = \alpha_I^c$.
\item $c_1, c_2 \in \RR$ and $I = \{d\}$ for some $d \in \RR_{\ge 0}$. Consider the monomial $\mu \in \MM_{\FF}$ defined as $(\bind{1} = d) \otimes (\bind{c_1} \otimes c_2)$. Then, $\beh{\mu} = \alpha_I^c$.
\item $c_1, c_2 \in \RR$ and $I = (d, d')$ with $d, d' \in \RR_{\ge 0}$ such that $d < d'$. Consider the monomial $\mu \in \MM_{\FF}$ defined as 
\[
\big((\llangle 1 \rrangle \oplus d) \neq d \big) \otimes \big((\llangle 1 \rrangle \oplus d') \neq \bind{1} \big) \otimes \big(\bind{c_1} \otimes c_2\big).
\]
Then, $\beh{\mu} = \alpha_I^c$.
\item $c_1, c_2 \in \RR$ and $I = (d, \infty)$ with $d \in \RR_{\ge 0}$. Consider the monomial $\mu \in \MM_{\FF}$ defined as $((\llangle 1 \rrangle \oplus d) \neq d) \otimes (\bind{c_1} \otimes c_2)$. Then, $\beh{\mu} = \alpha_I^c$.
\qedhere
\end{itemize}
\end{proof}

\subsubsection{Negation Operator (Remark \ref{REM:negation})}
\label{SUB:application_reduction_to_monomials}

As an application of Theorem~\ref{THM:ff_arc_times_reduction_to_monomials}, we now illustrate certain expressiveness limitations of WFFAs.

Let $\FF = \FF_{\arc, \times}$. For any QFL $\F: \DD \Sigma^* \to S$, where $S = \RR \cup \{-\infty\}$, define the {\em negation} $(-\F): \DD \Sigma^* \to S$ by
\[
(-\F)(w) = \begin{cases}
-\F(w), & \text{if } \F(w) \neq -\infty, \\
-\infty, & \text{otherwise}.
\end{cases}
\]

As observed in Example~\ref{EX:european_call_option_wffa}, such a negation operator is of practical relevance: it permits switching between the profit functions associated with long and short positions in a given financial instrument by simply reversing the sign of the payoff.
As demonstrated in Example~\ref{EX:european_call_option_wffa}, the profit functions of both the long and the short position in a European call option are WFFA-recognizable with respect to the finance semiring $\FF_{\arc, \times}$. However, for more complex financial instruments, this property may fail. The underlying intuition is that our model lacks the ability to compare the finance values carried by different transitions.

Let $S_{\arc} = \RR \cup \{-\infty\}$ and $S_{\trop} = \RR \cup \{-\infty\}$. Consider a QFL $\F: \DD \Sigma^* \to S_{\arc}$. Define the mapping $\T(\F): \DD \Sigma^* \to S_{\trop}$ by
\[
\T(\F)(w) =
\begin{cases}
\F(w), & \text{if } \F(w) \neq -\infty,\\[1mm]
+\infty, & \text{otherwise}.
\end{cases}
\]

Intuitively, $\T(\F)$ transforms the $-\infty$ values of $\F$ into $+\infty$, allowing us to interpret the QFL in the tropical semiring $\FF_{\trop, \times}$ while preserving the meaningful finite values of $\F$.

The following theorem clarifies the situation. While the negation of a WFFA-recognizable QFL need not be WFFA-recognizable over the arctic semiring, applying the transformation $\T$ allows one to represent the negated function within the tropical semiring.

\begin{theorem}
\label{THM:qfl_negation_unrec}
\begin{enumerate}[label=(\alph*)]
\item Let $\Sigma = \{\bot\}$ be a singleton alphabet and let $\F: \DD \Sigma^* \to S$ be defined for all $w \in \DD \Sigma^*$ by
\[
\F(w) = \begin{cases}
\max\{d_1, d_2\}, & \text{if } w=\bot_{d_1, d_2} \text{ for some } d_1, d_2 \in \DD, \\
-\infty, & \text{otherwise}.
\end{cases}
\] Then $\F \in \beh{\WFFA_{\Sigma, \FF_{\arc, \times}}}$, but $(-\F) \notin \beh{\WFFA_{\Sigma, \FF_{\arc, \times}}}$.
\item Let $\Sigma$ be any alphabet and let $\F \in \beh{\WFFA_{\Sigma, \FF_{\arc, \times}}}$. Then
$
\T(-\F) \in \beh{\WFFA_{\Sigma, \FF_{\trop, \times}}}.
$
\end{enumerate}
\end{theorem}

\begin{proof}
\begin{enumerate}[label=(\alph*)]
\item Let $\FF = \FF_{\arc, \times}$.
Consider the QFL ${\F: \DD \Sigma^* \to S}$ defined for all finance words $w = \bot_{d_1, \dots, d_n} \in \DD \Sigma^*$ over the singleton alphabet $\Sigma = \{\bot\}$ by
\[
\F(w) = \begin{cases}
\max\{d_1, d_2\}, & \text{if } n = 2, \\
-\infty, & \text{otherwise}.
\end{cases}
\]

The QFL $\F$ is recognized by the WFFA $\A_{\F} \in \WFFA_{\Sigma, \FF}$ depicted in Figure~\ref{FIG:thm_qfl_negation_unrec_proof}.

\begin{figure}[ht]
\centering
\begin{tikzpicture}[
    >=stealth,
    every node/.style={font=\footnotesize},
    state/.style={circle, draw, minimum size=18pt, inner sep=1pt}
]
\node (q0) {};
\node[state, right=1cm of q0] (q1) {$q_0$};
\node[state, above right=0.5cm and 1.5cm of q1] (q2) {$q_1$};
\node[state, below right=0.5cm and 1.5cm of q1] (q3) {$q_2$};
\node[state, below right=0.5cm and 1.5cm of q2] (q4) {$q_3$};
\node[right=1cm of q4] (q44) {};

\draw[->] 
    (q0) edge[above] node{$0.0$} (q1)	
    (q1) edge[above] node[above left,pos=.62]{$\llangle 1.0 \rrangle$} (q2)
    (q1) edge[above] node[below left,pos=.62]{$0.0$} (q3)
	(q2) edge[above] node[above right,pos=.38]{$0.0$} (q4)
    (q3) edge[above] node[below right,pos=.38]{$\llangle 1.0 \rrangle$} (q4)
    (q4) edge[above] node{$0.0$} (q44)
    ;
\end{tikzpicture}
\caption{WFFA $\A_{\F}$ from the proof of Theorem~\ref{THM:qfl_negation_unrec}(b).}
\label{FIG:thm_qfl_negation_unrec_proof}
\end{figure}
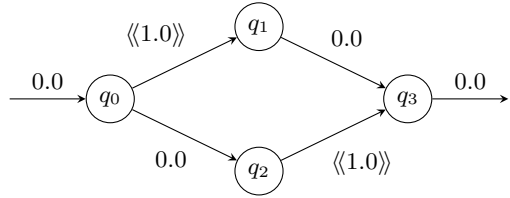


Now let $\F' = -\F$. Then for all finance words $w = \bot_{d_1, \dots, d_n} \in \DD \Sigma^*$ ,
\[
  \F'(\bot_{d_1,d_2}) = \begin{cases} \min\{-d_1, -d_2\}, & \text{if } n = 2, \\ -\infty, & \text{otherwise.} \end{cases}
\]

We show that $\F' \notin \beh{\WFFA_{\Sigma, \FF}}$.
Suppose, towards a contradiction, that there exists an automaton
$\A = (Q, I, T, F, \wt_I, \wt_T, \wt_F) \in \WFFA_{\Sigma,\FF}$ such that $\beh{\A} = \F'$.  
By Theorem~\ref{THM:ff_arc_times_reduction_to_monomials}, we may assume
without loss of generality that 
$\A \in \WFFA_{\Sigma,\FF}(\MM_{\FF})$.
More precisely, we may assume that every transition weight expression $e$
occurring in $\A$ satisfies 
$\beh{e} = \alpha_{I}^{\langle a, b\rangle}$,
where $a,b \in \RR_{\ge 0} \cup \{-\infty\}$ and
either $I = \{x\}$ for some $x \in \RR_{\ge 0}$,
or $I = (x,y)$, an open interval with $x \in \RR_{\ge 0}$ and
$y \in \RR_{\ge 0} \cup \{\infty\}$.

For all $k, l \in \NN_{\ge 1}$, choose a run $\varrho_{k,l} \in \Run_{\A}(\bot_{k,l})$ such that $\beh{\A}(\bot_{k,l}) = \wt_{\A}(\varrho_{k,l})$. Since $\A$ has only finitely many transitions, there exist transitions $t_1, t_2 \in T$, an infinite set $K \subseteq \NN_{\ge 1}$ and for each $k \in K$ an infinite set $L_k \subseteq \NN_{\ge 1}$, such that $\tau = t_1 t_2$ is the transition sequence of $\varrho_{k,l}$ for all $k \in K$ and $l \in L_k$.

Let $\beh{\wt_T(t_1)} = \alpha^{\langle a_1, b_1 \rangle}_{I_1}$ and $\beh{\wt_T(t_2)} = \alpha^{\langle a_2, b_2 \rangle}_{I_2}$. Clearly, $a_i, b_i \neq -\infty$ and each interval has the form $I_i = (c_i, \infty)$ for some $c_i \in \RR_{\ge 0}$. Denote by $z_0$ and $z_f$ the initial and final weights associated with $\tau$, respectively. Set $L_k' = L_k \cap [k + 1, \infty)$. Then for all $k \in K$ and all $l \in L_k'$ (so that $k < l$) we obtain:
\[
-l = \beh{\A}(\bot_{k, l}) = \wt_{\A}(\varrho_{k, l}) = z_0 + a_1 \cdot k + b_1 + a_2 \cdot l + b_2 + z_f
\]
Let $C = -(z_0 + b_1 + b_2 + z_f)$. Then the above equality rewrites as
\begin{equation}
\label{EQ:min_unrec_linear_eq}
a_1 \cdot k + (a_2 + 1) \cdot l = C.
\end{equation}
Fix $k \in K$. Since equality \eqref{EQ:min_unrec_linear_eq} holds for infinitely many $l \in L'_k$, the left-hand side ${a_1 \cdot k + (a_2 + 1) \cdot l}$ must remain constant while $l$ ranges over an infinite set. This is possible only if $a_2 = -1$. With $a_2 = -1$, Equation~\eqref{EQ:min_unrec_linear_eq} simplifies to $a_1 \cdot k = C$, and this equality holds for infinitely many $k \in K$. Hence the only possibility is $a_1 = 0$ and $C = 0$.

Now choose $\delta = 1 + \max\{c_1, c_2\}$. Then $\delta > c_1$ and $\delta > c_2$, so $\tau$ is a valid transition sequence for $\bot_{\delta + 1, \delta}$. Let $\varrho' \in \Run_{\A}(\bot_{\delta + 1, \delta})$ be the run over $\tau$. Then
\[
\wt_{\A}(\varrho') = (-C) + a_1 \cdot (\delta + 1) + a_2 \cdot \delta = -\delta
\]
On the other hand, by definition of the automaton's behavior,
\[
-(\delta + 1) = \beh{\A}(\bot_{\delta + 1, \delta}) \ge \wt_{\A}(\varrho') = -\delta,
\]
which is a contradiction.
\item Let
$
\A_{\arc} = (Q, I, T, F, \wt_I, \wt_T, \wt_F) \in \WFFA_{\Sigma, \FF_{\arc}}
$
be such that $\beh{\A} = \F$.
By Theorem~\ref{THM:ff_arc_times_reduction_to_monomials}, we may assume without loss of generality that, for every $t \in T$, 
\[
\beh{\wt_T(t)} = \alpha^{\langle a_t, b_t \rangle}_{I_t}
\]
where $a_t, b_t \in S_{\arc}$ and $I_t$ is either a singleton or an open interval.
Moreover, we may assume that $a_t \neq -\infty$ and $b_t \neq -\infty$, since otherwise the evaluation of such a transition weight would be $-\infty$, and transitions of this kind can be removed without affecting the behavior. This implies that $a_t, b_t \in \RR$ for all $t \in T$. Analogously, we may assume that all initial and final weights in $\A$ are real numbers. 

We construct a WFFA
$
\A_{\trop} = (Q, I, T, F, \tilde{\wt}_I, \tilde{\wt}_T, \tilde{\wt}_F) \in \WFFA_{\Sigma, \FF_{\trop, \times}}
$
as follows:
\begin{itemize}
\item For each $q_i \in I$, define $\hat{\wt}_I(q_i) = -\wt(q_i)$.
\item For each $q_f \in F$, define $\hat{\wt}_F(q_f) = -\wt(q_f)$.
\item For each $t \in T$, let $\tilde{c}_t \in \CC_{\FF_{\trop}}$ be a constraint  encoding the set $I_t$ (constructed as in the proof of Lemma~\ref{LEMMA:finarc_exp_piecewise_affine}). Then set
\[
\tilde{\wt}_T(t) = \tilde{c}_t \otimes [\llangle -a_t \rrangle \otimes [-b_t]]
\]
\end{itemize}

Clearly, for every finance word $w = (a_1, d_1) \dots (a_n, d_n) \in \DD \Sigma^*$, we have 
\[
\Run_{\A_{\trop}}(w) = \Run_{\A_{\arc}}(w).
\]

Let $\delta = d_1 \dots d_n$. Then, for every run $\varrho \in \Run_{\A_{\trop}}(w)$,
with $\wt_{\A_{\arc}}(\varrho, \delta) \neq -\infty$, the construction ensures
$
\wt_{\A_{\trop}}(\varrho, \delta) = -\wt_{\A_{\arc}}(\varrho, \delta).
$
Hence,
\[
\begin{aligned}
\beh{\A_{\trop}}(w)
&= \min_{\varrho \in \Run_{\A_{\trop}}(w)} \wt_{\A_{\trop}}(\varrho, \delta) \\
&= \T\!\left(- \max_{\varrho \in \Run_{\A_{\arc}}(w)} \wt_{\A_{\arc}}(\varrho, \delta) \right)
 = \T\!\left(-\beh{\A_{\arc}}(w)\right).
\end{aligned}
\]
where the last equality uses the definition of $\T$ (mapping $-\infty$ to $+\infty$).

Thus, $\beh{\A_{\trop}} = \T(-\F)$.
\qedhere
\end{enumerate}
\end{proof}

\subsection{Proof of the Kleene--Sch\"utzenberger Theorem (Section~\ref{SEC:proof_kleene_schuetzenberger})}

\begin{lemma}[Lemma~\ref{LEMMA:primitive_regexp_recognizable}]
Let $E \subseteq \EE_{\FF}$ be an arbitrary collection of $\FF$-expressions.
\begin{enumerate}[label=(\alph*)]
\item For every $e \in E$ and $a \in \Sigma$, $\beh{e_a} \in \beh{\WFFA_{\Sigma, \FF}^{\tau}}(E)$.
\item For every $s \in S$, $\beh{s_{\eps}} \in \beh{\WFFA_{\Sigma, \FF}(E)}$.
\end{enumerate}
\end{lemma}

\begin{proof}
Construct automata $\A_{e_a} \in \WFFA_{\Sigma, \FF}^{\tau}(E)$ and $\A_{s_{\eps}} \in \WFFA_{\Sigma, \FF}(E)$ as depicted in Figure \ref{FIG:wffa_atomic_regexp}. It is straightforward to verify that $\beh{\A_{e_a}} = \beh{e_a}$ and $\beh{\A_{s_{\eps}}} = \beh{s_{\eps}}$.
Thus, the claims (a) and (b) follow.

\begin{figure}[ht]
\centering
\begin{tikzpicture}[
    >=stealth,
    every node/.style={font=\footnotesize},
    state/.style={circle, draw, minimum size=18pt, inner sep=1pt}
]
\node (q0) {};
\node[state, right=1cm of q0] (q1) {$q_0$};
\node[state, right=2cm of q1] (q2) {$q_1$};
\node[right=1cm of q2] (q22) {};
\node[left=0.2cm of q0] {$\A_{e_a}$};

\node[right=2cm of q22] (p0) {};
\node[state, right=1cm of p0] (p1) {$p_0$};
\node[right=1cm of p1] (p11) {};
\node[left=0.2cm of p0] {$\A_{s_{\eps}}$};

\draw[->] 
	(q0) edge[above] node{$\one$} (q1)
	(q1) edge[above] node{$a \mid e$} (q2)
	(q2) edge[above] node{$\one$} (q22)
	(p0) edge[above] node{$s$} (p1)
	(p1) edge[above] node{$\one$} (p11);
\end{tikzpicture}
\caption{WFFAs $\A_{e_a}$ and $\A_{s_{\eps}}$ from the proof of Lemma~\ref{LEMMA:primitive_regexp_recognizable}.}
\label{FIG:wffa_atomic_regexp}
\end{figure}
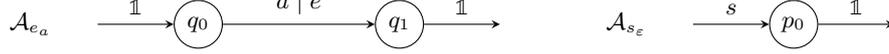

\end{proof}

\begin{lemma}[Lemma~\ref{LEMMA:epsfree_regexp_recognizable}]
Let $E \subseteq \EE_{\FF}$ be an arbitrary collection of $\FF$-expressions and $\R \in \Reg_{\Sigma, \FF}^{\epsfree}(E)$. Then, there exists $\A \in \WFFA_{\Sigma, \FF}^{\tau}(E)$ such that $\beh{\A} = \beh{\R}$.
\end{lemma}

\begin{proof}
We proceed by induction on the structure of $\R$.
\begin{itemize}
\item For $\R = e_a$ with $e \in E$ and $a \in \Sigma$, apply Lemma~\ref{LEMMA:primitive_regexp_recognizable}(a).
\item For $\R = \R_1 \oplus \R_2$, apply the induction hypothesis for $\R_1$ and $\R_2$ and Corollary \ref{COR:closure_properties_purely_transition_weighted}(a).
\item For $\R = \R_1 \cdot \R_2$, apply the induction hypothesis and Remark~\ref{REM:epsfree_proper} for $\R_1, \R_2 \in \Reg_{\Sigma, \FF}^{\epsfree}(E)$. Then, the claim follows from Corollary \ref{COR:closure_properties_purely_transition_weighted}(b).
\item Let $\R = \R_1 \cdot \R_2^*$ with $\R_1, \R_2 \in \Reg_{\Sigma, \FF}^{\epsfree}(E)$. By induction hypothesis, $\beh{\R_1}, \beh{\R_2} \in \beh{\WFFA_{\Sigma, \FF}^{\tau}(E)}$. By Remark~\ref{REM:epsfree_proper}, $\beh{R_1}(\eps) = \beh{\R_2}(\eps) = \zero$. By Corollary \ref{COR:closure_properties_purely_transition_weighted}, $\beh{\R_2^*} \in \beh{\WFFA_{\Sigma, \FF}^{\tau}(E)}$. Since $\beh{R_2^*}(\eps) = \one$, it follows from Corollary \ref{COR:closure_properties_purely_transition_weighted}(b) that $\R \in \beh{\WFFA_{\Sigma, \FF}^{\tau}(E)}$.
\item Let $\R = \R_1^* \cdot \R_2$. This case is symmetric to the previous one.
\end{itemize}
\end{proof}

\begin{lemma}[Lemma~\ref{LEMMA:restricted_regexp_recognizable}]
Let $E \subseteq \EE_{\FF}$ be an arbitrary collection of $\FF$-expressions and $\R \in \Reg_{\Sigma, \FF}^{\res}(E)$. Then, there exists $\A \in \WFFA_{\Sigma, \FF}(E)$ such that $\beh{\A} = \beh{\R}$.
\end{lemma}

\begin{proof}
We proceed by induction on the structure of $\R$.
\begin{itemize}
\item For regular expressions $e_a$ and $s_{\eps}$, apply Lemma~\ref{LEMMA:primitive_regexp_recognizable}.
\item For $\R = \R_1 \oplus \R_2$, apply the induction hypothesis for $\R_1$ and $\R_2$ and Theorem~\ref{THM:closure_properties_e_weighted_wffa}(a).
\item For $\R = \R_1 \cdot \R_2$, apply the induction hypothesis for $\R_1, \R_2$ and Theorem~\ref{THM:closure_properties_e_weighted_wffa}(c).
\item Let $\R = (\R_{\epsfree})^*$ where $\R_{\epsfree} \in \Reg_{\Sigma, \FF}^{\epsfree}(E)$. If follows from Lemma~\ref{LEMMA:epsfree_regexp_recognizable} that $\beh{\R_{\epsfree}} \in \beh{\WFFA_{\Sigma, \FF}^{\tau}(E)}$. By Remark~\ref{REM:epsfree_proper}, $\beh{\R_{\epsfree}}$ is proper. Then, by Corollary~\ref{COR:closure_properties_purely_transition_weighted}(c), ${\beh{(\R_{\epsfree})^*} \in \beh{\WFFA_{\Sigma, \FF}^{\tau}(E)}}$.
\end{itemize}
\end{proof}

\begin{lemma}[Lemma~\ref{LEMMA:from_wffa_to_restricted_regexp}]
Let $E \subseteq \EE_{\FF}$ and $\A \in \WFFA_{\Sigma, \FF}(E)$. Then, there exists $\R \in \Reg_{\Sigma, \FF}^{\res}(E)$ such that $\beh{\R} = \beh{\A}$.
\end{lemma}

\begin{proof}
Let $\A = (Q, I, T_c, T_m, F, \wt_I, \wt_T, \wt_F)$. Assume that $Q = \{1, \dots, m\}$.

For any transition $t = (p, a, q) \in T$, let $\R_t = (\wt_T(t))_a \in \Reg_{\Sigma, \FF}^{\epsfree}(E)$.
We proceed by induction on $k \in \{0, \dots, m\}$. For all $p, r \in Q$, we define the regular expression $\R_{p, r}^{(k)} \in \Reg_{\Sigma, \FF}^{\epsfree}(E)$ as follows:
\begin{itemize}
\item $\R_{p,r}^{(0)} = \bigoplus_{a \in \Sigma} \mathcal \R_{(p, a, r)}$
\item $\R_{p,r}^{(k+1)} = \R_{p,r}^{(k)} \oplus \R_{p,k+1}^{(k)} \cdot \big(\R_{k + 1, k + 1}^{(k)} \big)^* \cdot \R_{k+1,r}^{(k)}$ for $k \ge 0$. Here, since $\R_{k+1,k+1}^{(k)}$ is $\eps$-free, $\beh{\R_{k + 1, k + 1}^{(k)}}$ is proper by Remark~\ref{REM:epsfree_proper}.
\end{itemize}
One easily verifies, via a standard induction on $k \ge 0$, that $\R_{p, r}^{(k)}$ is $\varepsilon$-free.

For any word $u = a_1 \dots a_n \in \Sigma^*$, let $\textsc{Paths}_{p, r}^{(k)}(u)$ denote the set of all paths (i.e., run fragments) of $\A$ 
\[
\pi = \bigl(q_0 \xrightarrow{a_1} q_1 \xrightarrow{a_2} {\dots} \xrightarrow{a_n} q_n \bigr)
\]
where $q_0 = p$, $q_n = r$ and, for all $i \in [n]$, $t_i = (q_{i-1}, a_i, q_i) \in T$ and $q_i \le k$.
Given a {\em path context} $\delta = d_1 \dots d_n \in \DD^*$, we define the {\em weight} of $\pi$ in the usual manner by
\[
\wt^{\mathrm{path}}_{\A}(\pi, \delta) = \bigotimes_{i \in [n]} \beh{\wt_T(t_i)}(d_i)
\]
where we do not include the initial and final weights into the computation.
We say that $\pi$ is {\em valid} in context $\delta$ if $\wt^{\mathrm{path}}(\pi, \delta) \neq \zero$. Given $w = (a_1, d_1) \dots (a_n, d_n) \in \DD \Sigma^*$, let $\textsc{Paths}_{p,r}^{(k)}(w)$ denote the set of all $\pi \in \textsc{Paths}_{p,r}^{(k)}(a_1 \dots a_n)$ which are valid in path context $d_1 \dots d_n$.

Following the standard inductive argument used in proofs of the classical Kleene-Schützenberger theorem, one can show by induction on $k$ that, for all $p,r \in Q$ and all finance data words $w = (a_1, d_1) \dots (a_n, d_n) \in \DD \Sigma^+$,
\[
\beh{\R_{p,r}^{(k)}}(w) = \bigoplus_{\pi \in \textsc{Paths}_{p,r}^{(k)}(a_1 \dots a_n)} \wt^{\mathrm{path}}_{\A}(\pi, d_1 \dots d_n).
\]
It can also be verified that 
\[
\Run_{\A}(w) = \bigcup_{i \in I, f \in F} \textsc{Paths}_{i, f}^{(m)}(w).
\]
Let $\delta = d_1 \dots d_n$. Then, for any $\pi \in \textsc{Paths}_{i, f}^{(m)}(w)$, we have $\pi \in \Run_{\A}(w)$ and
\[
\wt_I(i) \otimes \wt^{\mathrm{path}}_{\A}(\pi,  \delta) \otimes \wt_F(f) = \wt_{\A}(\pi, \delta).
\]
Let $\R \in \Reg_{\Sigma, \FF}^{\res}(E)$ be defined as
\[
\R = \bigoplus_{i \in I, f \in F} \bigl((\wt_I(i))_{\eps} \cdot \R_{i, f}^{(m)} \cdot (\wt_F(f))_{\eps} \bigr).
\]
Then, $\beh{\R}(w) = \beh{\A}(w)$ for each $w \in \DD \Sigma^*$ with $w \neq \eps$. Hence, if $\beh{\A}$ is proper, we have $\beh{\R} = \beh{\A}$. If $\A$ is not proper, let $\R' = \R \oplus s_{\eps}$ with $s = \beh{\A}(\eps)$. Then, $\R' \in \Reg_{\Sigma, \FF}^{\res}(E)$ and $\beh{\R'} = \beh{\A}$.
\end{proof}

\subsection{Matrix Representation of WFFAs (Section \ref{SEC:matrix_representation})}

\begin{lemma}[Remark~\ref{REM:finance_homo_counterexample}]
Let $\Sigma = \{\bot\}$ and $\FF = \FF_{\arc, \times}$. Let $\F: \FF \Sigma^* \to S$ be the QFL defined for all $d_1, \dots, d_n \in \DD$ by
$
\F(\bot_{d_1, \dots, d_n}) = d_1^2 + {\dots} + d_n^2.
$
Then:
\begin{enumerate}[label=(\alph*)]
\item For $k = 1$,  there exist vectors $\lambda \in S^{1 \times k}$, $\nu \in S^{k \times 1}$ and a morphism ${\mu: \DD \Sigma_{\free}^* \to \SS_{\boxtimes}^{k \times k}}$ such that $\F(w) = \lambda \boxtimes \mu(w) \boxtimes \nu$ for all $w \in \FF \Sigma^*$.
\item $\F \notin \beh{\WFFA_{\Sigma, \FF}}$.
\end{enumerate}
\end{lemma}

\begin{proof}
\begin{enumerate}[label=(\alph*)]
\item Let $\lambda = \nu = (\one)$. Let $\mu: \DD \Sigma^* \to S^{1 \times 1}$ be the morphism defined by $\mu(\eps) = (\one)$ and, for all $d_1, \dots, d_n \in \DD$ with $n \in \NN_{\ge 1}$ by $\mu(\bot_{d_1, \dots, d_n}) = (d_1^2 + {\dots} + d_n^2)$. Then, $\F(w) = \lambda \boxtimes \mu(w) \boxtimes \nu$ for all $w \in \DD \Sigma^*$.
\item Suppose that there exists $\A = (Q, I, T, F, \wt_I, \wt_T, \wt_F) \in \WFFA_{\Sigma, \FF}$ such that $\beh{\A} = \F$. Let $T' \subseteq T$ be the set of all transitions $t = (i, \bot, f)$ of $\A$ where $i \in I$ and $f \in F$. We also let $e_t = [\wt_I(i) \otimes \wt_T(t) \otimes \wt_F(f)] \in \EE_{\FF}$. Consider $e = [\bigoplus_{t \in T'} e_t] \in \EE_{\FF}$. Then, for all $d \in \DD$, $d^2 = \beh{\A}(w_d) = \beh{e}(d)$. By Lemma~\ref{LEMMA:finarc_exp_piecewise_affine}, $\beh{e}(d): \DD \to S$ is a piecewise affine function. Then, there exist $d' \in \DD$ and constants $A, B \in S$ such that $\beh{e}(d) = A \cdot d + B$ for all $d \ge d'$. Hence, $d^2 = A \cdot d + B$ for all $d \ge d'$. A contradiction. 
\qedhere
\end{enumerate}
\end{proof}

\begin{lemma}[Lemma~\ref{LEMMA:wffa_to_morphism}]
Let $E \subseteq \EE_{\FF}$ with $\beh{\zero} \in \beh{E}$ and $\A \in \WFFA_{\Sigma, \FF}(E)$. Then, there exist $k \in \NN_{\ge 1}$, vectors $\lambda \in S^{1 \times k}$ and $\nu \in S^{k \times 1}$, and an $E$-generated morphism $\mu: \DD \Sigma_{\free}^* \to \SS^{k \times k}_{\boxtimes}$ such that $\beh{\A}(w) = \lambda \boxtimes \mu(w) \boxtimes \nu$ for all $w \in \DD \Sigma^*$.
\end{lemma}

\begin{proof}
Let $e_{\zero} \in E$ such that $\beh{e_{\zero}} = \beh{\zero}$.
Let $\A = (Q, I, T, F, \wt_I, \wt_T, \wt_F)$. Assume that $Q \neq \emptyset$. Let $(q_1, \dots, q_k)$ be an enumeration of $Q$ and let:
\begin{itemize}
\item $\lambda = (\lambda_1, \dots, \lambda_k)$ where, for each $i \in [k]$, $\lambda_i = \begin{cases}
\wt_I(q_i), & \text{if } q_i \in I, \\
\zero, & \text{otherwise};
\end{cases}$
\item $\nu = (\nu_1, \dots, \nu_k)^T$ where, for each $i \in [k]$, $\nu_i = \begin{cases} \wt_F(q_i), & \text{if } q_i \in F, \\ \zero, & \text{otherwise};
\end{cases}$
\item $\xi: \Sigma \to E^{k \times k}$ be defined for all $a \in \Sigma$ and $i, j \in [k]$ by
$$
(\xi(a))_{i,j} = \begin{cases}
\wt_T(q_i, a, q_j), & \text{if } (q_i, a, q_j) \in T, \\
e_{\zero}, & \text{otherwise}.
\end{cases}
$$
\end{itemize}
It can be shown that $\beh{\A}(w) = \lambda \boxtimes \mu_{\xi}(w) \boxtimes \nu$ for all $w \in \DD \Sigma^*$.

\end{proof}

\begin{lemma}[Lemma~\ref{LEMMA:morphism_to_wffa}]
Let $k \in \NN_{\ge 1}$, $E \subseteq \EE_{\FF}$, $\lambda \in S^{1 \times k}$, $\nu \in S^{k \times 1}$ and $\xi: \Sigma \to E^{k \times k}$. Let $\F: \DD \Sigma^* \to S$ be defined by $\F(w) = \lambda \boxtimes \mu_{\xi}(w) \boxtimes \nu$ for all $w \in \DD \Sigma^*$. Then, $\F \in \beh{\WFFA_{\Sigma, \FF}(E)}$.
\end{lemma}

\begin{proof}
Let $\lambda = (\lambda_1, \dots, \lambda_k)$ and $\nu = (\nu_1, \dots, \nu_k)^T$. Let $\A \in \WFFA_{\Sigma, \FF}(E)$ be defined as $\A = (Q, I, T, F, \wt_I, \wt_T, \wt_F) $ where:
\begin{itemize}
\item $Q = I = F = [k]$, $T = Q \times \Sigma \times Q$;
\item for all $i \in [k]$, $\wt_I(i) = \lambda_i$ and $\wt_F(i) = \nu_i$;
\item for all $t = (i, a, j) \in T$, $\wt_T(t) = (\xi(a))_{i,j} \in E$.
\end{itemize}
Then, $\A \in \WFFA_{\Sigma, \FF}(E)$ and $\beh{\A} = \L$.
\end{proof}

\subsection{Decision Problems for WFFAs (Section~\ref{SEC:decision_problems})}

\begin{lemma}[Lemma~\ref{LEMMA:expression_support_decidable}]
Let $\FF = \FF_{\arc, \times}$ and $e \in \EE_{\FF}$. Then, it is decidable whether $\supp(\beh{e}) \neq \emptyset$.
\end{lemma}

\begin{proof}
By Lemma~\ref{LEMMA:finarc_exp_piecewise_affine}, we can effectively find a representation of $\beh{e}$ as in Equation~\eqref{EQ:finarc_exp_piecewise_affine}.
Then, $\supp(\beh{e}) \neq \emptyset$ if and only if $\supp\!\left(\alpha_I^{c_e(J)} \right) \neq \zero$ for some $J \in \P(\delta_e)$ (i.e., $J \neq \emptyset$).
For every $J \in \P(\delta_e)$, let $c_e(J) = \langle c_J, c'_J \rangle$. Then, $\alpha_J^{c_e(J)}: J \to S$ with
$\alpha_J^{c_e(J)}(d) = c_J \cdot d + c'_J$ and hence $\supp(\alpha_I^{c_e(J)}) \neq \zero$ 
if and only if $c_J \neq \zero$ and $c'_J \neq \zero$. This means that $\supp(\beh{e}) \neq \zero$ if and only if there exists 
$I \in \P(\delta)$ such that $c_I \neq \zero$ and $c'_J \neq \zero$.
Hence, it is decidable whether $\supp(\beh{e}) \neq \emptyset$.
\end{proof}

\begin{lemma}[Lemma~\ref{LEMMA:rec_qfl_support_regular}]
Let $\FF = \FF_{\arc, \times}$ and $\A \in \WFFA_{\Sigma, \FF}$. Then, $\pi_{\Sigma}(\supp(\beh{\A})) \subseteq \Sigma^*$ is a regular language.
\end{lemma}

\begin{proof}
Let $\A = (Q, I, T, F, \wt_I, \wt_T, \wt_F)$. We define the non-deterministic finite automaton $\A' = (Q, I', T', F')$ by letting:
\begin{itemize}
\item $I' = \{i \in I  \mid  \wt_I(i) \neq \zero\}$, $F' = \{f \in F  \mid  \wt_F(f) \neq \zero\}$;
\item $T' = \{t \in T  \mid  \supp(\beh{\wt_T(t)}) \neq \emptyset\}$.
\end{itemize} 
Then, for every $w \in \DD \Sigma^*$, $\beh{\A}(w) \neq \zero$ if and only if $\A'$ accepts the word $\pi_{\Sigma}(w)$. Hence, $\A'$ accepts the language $\pi_{\Sigma}(\supp(\beh{\A})) \subseteq \Sigma^*$, i.e., $\pi_{\Sigma}(\supp(\beh{\A}))$ is regular.
\end{proof}

\begin{lemma}[Lemma~\ref{LEMMA:expression_sup_computable}]
Let $e \in \EE_{\FF}$ and either $\I = \DD$ or $\I = [l, h] \subseteq \DD$. Then, $\sup_{x \in \I} \beh{e}(x)$ is computable.
\end{lemma}

\begin{proof}
By Lemma~\ref{LEMMA:finarc_exp_piecewise_affine}, $\beh{e}$ is a piecewise affine function, i.e., there exist a finite set $\delta = \{d_1, \dots, d_k\} \subseteq \DD$ with $0 = d_0 < d_1 < {\dots} < d_k < d_{k + 1} = \infty$ and a mapping $c: \P(\delta) \to S^2$ such that $\beh{e} = \bigoplus_{J \in \P(\delta)} \alpha_J^{c(J)}$. For every $i \in \{0, \dots, k\}$, let $J_i = (d_i, d_{i + 1})$ and $c(J_i) = \langle c_i, c'_i \rangle$. Then, for all $i \in \{0, \dots, k\}$,
$$
M_i = \sup_{x \in J_i} \alpha_{J_i}^{c(J_i)}(x) = c_i' + \max\{c_i \cdot d_i, c_i \cdot d_{i + 1}\}
$$
where for $d_{k + 1} = \infty$, we assume $c_i \cdot \infty = \infty$ whenever $c_i > 0$, 
$c_i \cdot \infty = -\infty$ whenever $c_i < 0$, $0 \cdot \infty = 0$ and $(-\infty) \cdot \infty = -\infty$. Note that $M_i \in \RR \cup \{-\infty, \infty\}$.
Then,
$$
\sup_{x \in \DD} \beh{e}(x) = \max\{\beh{e}(d_0), {\dots}, \beh{e}(d_k), M_0, \dots, M_k \}.
$$
Therefore, $\sup_{x \in \DD} \beh{e}(x)$ is computable.

For $\I = [l, h]$, we can use Remark~\ref{REMARK:piecewise_affine_extension} and assume without loss of generality that $l = d_{p}$ and $h = d_{q}$ for some $0 \le p < q \le k$.
Then,
$$
\sup_{x \in [l, h]} \beh{e}(x) = \max\{ \beh{e}(d_p), {\dots}, \beh{e}(d_q), M_p, {\dots}, M_q \}
$$
and hence $\sup_{x \in [l, h]} \beh{e}(x)$ is also computable.
\end{proof}

Now ce can give the proof of Theorem~\ref{THM:threshold_problem_wffa}.
\begin{proof}
Let $\A = (Q, I, T, F, \wt_I, \wt_T, \wt_F) \in \WFFA_{\Sigma, \FF}$, $\I = \DD$ or $\I = [l, h]$ and $\theta \in \RR$ a threshold. Our goal is to determine whether there exists a finance word $w \in \DD_{\I} \Sigma^*$ such that $\beh{\A}(w) > \theta$. For each transition $t \in T$, we use Lemma~\ref{LEMMA:expression_sup_computable} to compute $M_t = \sup_{x \in \I} \beh{\wt_T(t)}(x)$.
Define the sets:
\begin{itemize}
\item $T_{-\infty} = \{t \in T  \mid  M_t = -\infty\}$,
\item $T_{\infty} = \{t \in T  \mid  M_t = \infty\}$.
\end{itemize}
Observe that for every $t \in T_{-\infty}$ and $d \in \I$, we have $\beh{\wt_T(t)}(d) = -\infty$. Consequently, any run of $\A$ that includes a transition in $T_{-\infty}$ is not valid.

First, for every transition $t = (p, \sigma, q) \in T_{\infty}$ check whether:
\begin{itemize}
\item $p$ is reachable from an initial state from $I$ via transitions in $T \setminus T_{-\infty}$, and
\item some final state from $F$ is reachable from $q$ via valid transitions in $T \setminus T_{-\infty}$
\end{itemize} 
If both conditions hold, then for every threshold $\vartheta \in \RR$ there exists a run $\varrho$ of $\A$ on some word $w \in \DD_{\I}(\Sigma^*)$ such that $\varrho$ touches the transition $t$ and $\wt_{\A}(\varrho) > \vartheta$. As a consequence, there exists $w \in \DD_I \Sigma^*$ such that $\beh{\A}(w) \ge \wt_{\A}(\varrho) > \theta$.

Otherwise, if no transition in $T_{\infty}$ appears in any valid run of $\A$, 
we construct a semiring-weighted automaton $\A' = (Q, I, T', F, \wt_I, \wt_{T'}, \wt_F)$ over $\Sigma$ and $\SS_{\arc}$ where $T' = T \setminus (T_{-\infty} \cup T_{\infty})$ and $\wt_{T'}: T' \to \RR$ is defined for all $t \in T'$ by $\wt_{\T'}(t) = M_t$. Then, $\beh{\A}(w) > \theta$ for some $w \in \DD_I \Sigma^*$ if and only if $\beh{\A'}(u) > \theta$ for some $u \in \Sigma^*$.
Since the $>$-threshold problem for weighted automata over the max-plus semiring is decidable (cf.~\cite{ABK22}), and the present setting can be reduced to this case, the $>$-threshold problem for $\WFFA_{\Sigma, \FF}$ over $\DD_{\I} \Sigma^*$ is also decidable.
\end{proof}

\end{appendices}

\end{document}